\documentclass[11pt]{article}

\usepackage{color}
\usepackage{amsfonts,amsmath,amssymb,amsthm,boxedminipage,color,url,fullpage}
\definecolor{weborange}{rgb}{.8,.3,.3}
\definecolor{webblue}{rgb}{0,0,.8}
\definecolor{internallinkcolor}{rgb}{0,.5,0}
\definecolor{externallinkcolor}{rgb}{0,0,.5}

\usepackage{float}
\restylefloat{figure}

\usepackage[pagebackref,
hyperfootnotes=false,
colorlinks=true,
urlcolor=externallinkcolor,
linkcolor=internallinkcolor,
filecolor=externallinkcolor,
citecolor=internallinkcolor,
breaklinks=true,
pdfstartview=FitH,
pdfpagelayout=OneColumn]{hyperref}

\usepackage{enumerate,paralist}
\usepackage[numbers]{natbib} 
\usepackage[labelfont=bf]{caption}
\usepackage{aliascnt}
\usepackage{cleveref}
\usepackage{xspace}

\usepackage{relsize}
\usepackage{graphics}

\newcommand{\remove}[1]{}
\newcommand{\Draft}[1]{\ifdefined\IsDraft\texttt{ #1} \fi}

\newcommand{\TLLNCS}[2]{\ifdefined\IsLLNCS#1\else #2 \fi}

\ifdefined\IsDraft
\newcommand{\authnote}[2]{{\bf [{\color{red} #1's Note:} {\color{blue} #2}]}}
\else
\newcommand{\authnote}[2]{}
\fi

\ifdefined\IsDraft

\else

\fi

\ifdefined\IsDraft
\newcommand{\deleted}[1]{{\color{blue} ~Deleted:~{\color{red} #1}}}
\else
\newcommand{\deleted}[1]{}
\fi

\newcommand{\sdotfill}{\textcolor[rgb]{0.8,0.8,0.8}{\dotfill}} 

\newenvironment{algorithm}{\begin{algo}}{\vspace{-\topsep}\sdotfill\end{algo}}
\newenvironment{experiment}{\begin{expr}}{\vspace{-\topsep}\sdotfill\end{expr}}

\newcommand{\Ensuremath}[1]{\ensuremath{#1}\xspace}
\newcommand{\MathAlg}[1]{\mathsf{#1}}
\newcommand{\RND}{{\mathsf{rnd}}}

\newcommand{\rnd}[1]{\mathsf{rnd}_{#1}}
\newcommand{\MathAlgX}[1]{\Ensuremath{\MathAlg{#1}}}

\newcommand{\ie}  {i.e.,\xspace}
\newcommand{\eg}  {e.g.,\xspace}

\newcommand{\wrt} {with respect to\xspace}
\newcommand{\wlg} {without loss of generality\xspace}
\newcommand{\cf}{cf.,\xspace}

\newcommand{\set}[1]{\ens{#1}}

\newcommand{\floor}[1]{\left \lfloor#1 \right \rfloor}

\newcommand{\vct}{{\bf v}}
\newcommand{\ith}{\Ensuremath{i^{\rm th}}}

\newcommand{\Lap}[1]{\mathsf{Lap}(#1)}

\newcommand{\eqdef}{:=}

\newcommand{\R}{{\mathbb R}}
\newcommand{\N}{{\mathbb{N}}}
\newcommand{\Z}{{\mathbb Z}}

\newcommand{\zo}{\{0,1\}}

\newcommand{\zs}{{\zo^\ast}}
\newcommand{\ot}{\set{1,2}}

\newcommand{\eps}{\varepsilon}

\newcommand{\la}{\gets}

\newcommand{\poly}{\operatorname{poly}}
\newcommand{\polylog}{\operatorname{polylog}}
\newcommand{\loglog}{\operatorname{loglog}}

\newcommand{\negl}{\operatorname{neg}}

\newcommand{\Supp}{\operatorname{Supp}}
\newcommand{\Img}{\operatorname{Im}}

\newcommand{\MathFam}[1]{\mathcal{#1}}

\newcommand{\Rng}{\MathFam{R}}

\renewcommand{\P}{\class{P}}

\renewcommand{\cref}{\Cref}

\TLLNCS{
	\newaliascnt{claiml}{theorem}
	\newtheorem{claiml}[claiml]{Claim}
	\aliascntresetthe{claiml}
	
	\renewenvironment{claim}{\begin{claiml}}{\end{claiml}}
}
{
	\newtheorem{theorem}{Theorem}[section]

	\newaliascnt{lemma}{theorem}
	\newtheorem{lemma}[lemma]{Lemma}
	\aliascntresetthe{lemma}
	
	\newaliascnt{claim}{theorem}
	\newtheorem{claim}[claim]{Claim}
	\aliascntresetthe{claim}

	\newaliascnt{corollary}{theorem}
	\newtheorem{corollary}[corollary]{Corollary}
	\aliascntresetthe{corollary}
	
	\newaliascnt{proposition}{theorem}
	\newtheorem{proposition}[proposition]{Proposition}
	\aliascntresetthe{proposition}

	\newaliascnt{conjecture}{theorem}
	
	\aliascntresetthe{conjecture}

	\newaliascnt{definition}{theorem}
	\newtheorem{definition}[definition]{Definition}
	\aliascntresetthe{definition}

	\newaliascnt{remark}{theorem}
	\newtheorem{remark}[remark]{Remark}
	\aliascntresetthe{remark}

	\newaliascnt{example}{theorem}
	
	\aliascntresetthe{example}
}

\crefname{lemma}{Lemma}{Lemmas}
\crefname{figure}{Figure}{Figures}
\crefname{claim}{Claim}{Claims}
\crefname{corollary}{Corollary}{Corollaries}
\crefname{proposition}{Proposition}{Propositions}
\crefname{conjecture}{Conjecture}{Conjectures}
\crefname{definition}{Definition}{Definitions}
\crefname{remark}{Remark}{Remarks}
\crefname{exmaple}{Example}{Examples}

\newcommand{\stepref}[1]{Step \ref{#1}}

\newaliascnt{construction}{theorem}

\aliascntresetthe{construction}
\crefname{construction}{Construction}{Constructions}

\newaliascnt{fact}{theorem}
\newtheorem{fact}[fact]{Fact}
\aliascntresetthe{fact}
\crefname{fact}{Fact}{Facts}

\newaliascnt{notation}{theorem}
\newtheorem{notation}[notation]{Notation}
\aliascntresetthe{notation}
\crefname{notation}{Notation}{Notation}

\crefname{equation}{Equation}{Equations}

\newaliascnt{proto}{theorem}

\newtheorem{proto}[proto]{Protocol}

\aliascntresetthe{proto}
\crefname{proto}{protocol}{protocols}

\newaliascnt{algo}{theorem}
\newtheorem{algo}[algo]{Algorithm}
\aliascntresetthe{algo}
\crefname{algo}{algorithm}{algorithms}

\newaliascnt{expr}{theorem}
\newtheorem{expr}[expr]{Experiment}
\aliascntresetthe{expr}
\crefname{experiment}{experiment}{experiments}

\def\FullBox{$\Box$}
\def\qed{\ifmmode\qquad\FullBox\else{\unskip\nobreak\hfil
		\penalty50\hskip1em\null\nobreak\hfil\FullBox
		\parfillskip=0pt\finalhyphendemerits=0\endgraf}\fi}

\def\qedsketch{\ifmmode\Box\else{\unskip\nobreak\hfil
		\penalty50\hskip1em\null\nobreak\hfil$\Box$
		\parfillskip=0pt\finalhyphendemerits=0\endgraf}\fi}

\newcommand{\eex}[2]{\Ex_{#1}\left[#2\right]}
\newcommand{\ex}[1]{\Ex\left[#1\right]}
\newcommand{\Ex}{{\mathbf E}}
\renewcommand{\Pr}{{\mathrm {Pr}}}
\newcommand{\pr}[1]{\Pr\left[#1\right]}
\newcommand{\ppr}[2]{\Pr_{#1}\left[#2\right]}

\newcommand{\Ac}{\MathAlgX{A}}

\newcommand{\Pc}{\mathsf{P}}

\renewcommand{\H}{H}

\newcommand{\ens}[1]{\left\{#1\right\}}
\newcommand{\size}[1]{\left|#1\right|}

\newcommand{\out}{\mathsf{out}}

\newcommand{\Uni}{{\mathord{\mathcal{U}}}}

\newcommand{\prob}[1]{\mathsf{\textsc{#1}}}

\newcommand{\SD}{\prob{SD}}

\newcommand{\ppt}{{\sc ppt}\xspace}
\newcommand{\pptm}{{\sc pptm}\xspace}
\newcommand{\polynu}{\ensuremath{{\text{\pptm}}^{\sf NU}}\xspace}

\newcommand{\cH}{{\cal{H}}}
\newcommand{\cA}{\mathcal{A}}
\newcommand{\cB}{\mathcal{B}}

\newcommand{\cU}{\mathcal{U}}

\newcommand{\cD}{\mathcal{D}}
\newcommand{\cR}{\mathcal{R}}
\newcommand{\cT}{\mathcal{T}}
\newcommand{\cP}{\mathcal{P}}

\newcommand{\cC}{\mathcal{C}}

\newcommand{\st}{\text{ s.t.\ }}

\newcommand{\cu}{{\cal{U}}}
\newcommand{\cs}{{\cal{S}}}
 \newcommand{\ocs}{{\overline{\cs}}}

\newcommand{\cx}{{\cal{X}}}

\theoremstyle{definition}

\theoremstyle{plain}

\theoremstyle{remark}

\renewcommand{\c}{\mathbf{c}}

\newcommand{\supp}{\mathrm{supp}}

\newcommand{\wt}[1]{\widetilde{#1}}
\newcommand{\A}{\mathcal{A}}
\renewcommand{\S}{\mathcal{S}}

\renewcommand{\H}{\mathcal{H}}

\newcommand{\wb}[1]{\overline{#1}}

\renewcommand{\P}{\mathcal{P}}

\newcommand{\rmv}[1]{\setminus  #1 }

\newcommand{\Biased}{{\mathcal{N}\!\mathsf{on\mathcal{S}imilar}}}
\newcommand{\Unbiased}{{\mathcal{S}\mathsf{imilar}}}

\newcommand{\np}{n}

\newcommand{\hp}{\chi}
\newcommand{\GvAttack}{\mathsf{MartAttack}}
\newcommand{\DpAttack}{\mathsf{DpAttack}}
\newcommand{\ApproxAttack}{\mathsf{SingAttack}}

\newcommand{\aux}{{\mathsf{NoJump}}}
\newcommand{\hist}{{\mathsf{agt}}}

\newcommand{\backup}{\mathsf{Bckp}}
\newcommand{\AvgBackup}{\mathsf{AvgBckp}}
 
\newcommand{\cQ}{\mathcal{Q}} 
\newcommand{\adv}{\mathsf{adv}}

\newcommand{\trigg}{\mathsf{trig}} 

\newcommand{\coef}[1]{\mathsf{coef}_{\np}(k,#1)}

\newcommand{\ks}{{k^\ast}}
\newcommand{\bS}{{\mathbb{S}}}
\newcommand{\bE}{{\mathbb{E}}}

\newcommand{\LapExp}{\mathsf{LapExp}}

\newcommand{\mh}{{\setminus h}}
\newcommand{\mH}{{\setminus H}}
\newcommand{\tsh}{\mathsf{tsh}}
\newcommand{\party}{b}	
\newcommand{\bj}{{1-\party}}

\newcommand{\noh}{s}

\newcommand{\Inote}[1]{\authnote{Iftach}{#1}}
\newcommand{\Nnote}[1]{\authnote{Nikos}{#1}}
\newcommand{\Enote}[1]{\authnote{Eran}{#1}}

\newcommand{\enote}{\Enote}

\title{Tighter Bounds on   Multi-Party Coin Flipping   via \\Augmented Weak Martingales and Differentially Private   Sampling\footnote{An extended abstract of this work appears in the 59th Annual IEEE Symposium on 
Foundations of Computer Science (FOCS'2018).} \Draft{\\{\small \sc Working Draft: Please Do Not Distribute}}}

  \author{Amos Beimel\thanks{Department of Computer Science, Ben Gurion University. 
	E-mail: \texttt{amos.beimel@gmail.com}. Research supported by ISF grant 152/17.}
 	\and Iftach Haitner\thanks{School of Computer Science, Tel Aviv University. E-mail: \texttt{iftachh@cs.tau.ac.il}. Director of the Check Point Institute for Information Security.} %
 	\footnote{Research supported by ERC starting grant 638121.}
 	\and Nikolaos Makriyannis\thanks{School of Computer Science, Tel Aviv University. E-mail: \texttt{n.makriyannis@gmail.com}.}~\footnotemark[3]
 	\and Eran Omri\thanks{Department of Computer Science, Ariel University. E-mail: \texttt{omrier@ariel.ac.il}. Research supported by ISF grant 152/17.}
 	}
 
\begin{document}

	\sloppy
	  \maketitle

\begin{abstract}
In his seminal work, \citeauthor{Cleve86} [STOC '86] has proved that any $r$-round coin-flipping protocol can be efficiently  biased  by $\Theta(1/r)$. This  lower bound  was met for the two-party case by  \citeauthor*{MoranNS16} [Journal of Cryptology '16], and  the three-party case (up to a $\polylog$ factor) by \citeauthor{HaitnerT17} [SICOMP '17], and was approached for $\np$-party protocols when $\np< \loglog r$ by     \citeauthor*{BHLT17} [SODA '17]. For $\np> \loglog r$, however, the best bias for $\np$-party coin-flipping protocols remains  $O(\np/\sqrt{r})$ achieved by the majority protocol of  \citeauthor*{AwerbuchBCGM1985}  [Manuscript '85].

Our main result is a  tighter lower bound on the bias of   coin-flipping protocols, showing that, for every constant $\eps >0$,  an $r^{\eps}$-party $r$-round coin-flipping protocol can be efficiently biased by $\widetilde{\Omega}(1/\sqrt{r})$.  As far as we know, this  is the first improvement of  \citeauthor{Cleve86}'s bound, and is only $\np=r^{\eps}$ (multiplicative) far from the aforementioned upper bound of \citeauthor{AwerbuchBCGM1985} 

We  prove the above   bound using   two new results that we believe are of independent interest. The first result is that a sequence  of (``augmented'') weak martingales have large gap:  with constant probability there exists two adjacent variables whose gap is  at least the ratio between the gap between the first and last variables and  the square root of the number of variables. This  generalizes over the  result of \citeauthor{CleveI93} [Manuscript '93], who showed that the above holds for strong martingales, and allows in some setting to exploit this gap by efficient algorithms.  We prove the above using   a  novel argument  that does not follow the more complicated approach of  \cite{CleveI93}. The second result is a new sampling algorithm  that uses a differentially private mechanism to minimize  the effect of data divergence.

\end{abstract}



\section{Introduction}\label{sec:intro}
In a coin-flipping protocol, introduced by \citet{Blum83}, the parties wish to output a common (close to) unbiased bit, even though some of the parties may be corrupted and try to bias the output. More formally, such protocols should satisfy the following two properties: first, when all parties are honest (\ie follow the prescribed protocol), they all output the \emph{same} unbiased bit. Second, even when some parties are corrupted (\ie collude and arbitrarily deviate from the protocol), the remaining parties should still output the same bit, and this bit should not be too biased (\ie its distribution should be close to being uniform over $\zo$). We emphasize that
the above requirements stipulate  
that the honest parties should \emph{always} output a common bit, regardless of what the corrupted parties do, and in particular  they are not allowed to abort if a cheat was detected.\footnote{Such protocols are typically  addressed as having \emph{guaranteed output delivery}, or, abusing terminology, as \emph{fair}.} Coin flipping is a fundamental  primitive with numerous  applications, and thus  lower bounds on coin flipping  protocols
imply analogous  bounds on many other basic cryptographic primitives,  including other inputless primitives and secure computation of functions that have input (\eg XOR).  

 In his seminal work, \citet{Cleve86} showed that for \emph{any} efficient two-party $r$-round coin-flipping protocol, there exists an efficient  adversarial strategy  that biases the output of the honest party by $\Theta(1/r)$, and his bound  extends to the multi-party case with no honest majority,  via a simple reduction. The above lower bound on coin-flipping protocols was met for the two-party case by \citet*{MoranNS16} and for the three-party case (up to a $\polylog$ factor) by \citet{HaitnerT17}, and was approached for $\np$-party coin-flipping protocols when $\np< \loglog r$ by   \citet*{BHLT17}. For $\np> \loglog r$, however, the smallest bias for $\np$-party coin-flipping  protocol remains  $\Theta(\np/\sqrt{r})$, achieved by the  majority protocol of  \citet*{AwerbuchBCGM1985}.

\subsection{Our Results}
Our main result is  the following  lower bound on the security of  coin-flipping protocols.
\begin{theorem}[Main result, informal]\label{thm:Infmain}
	For any $\np$-party $r$-round  coin-flipping protocol with  $\np^k\geq r$  for some $k \in \N$,  there exists a fail-stop\footnote{Acts honestly, but might abort prematurely. }  adversary running in time $\np^k$, corrupting all parties but one,  that biases the output of honest party  by $1/(\sqrt{r}\cdot \log(r)^k)$. 
\end{theorem}

As a concrete example, assume the number of parties is $\np= r^{1/100}$. The above theorem yields an attack of bias  $\widetilde{\Omega}(1/\sqrt{r}) = \widetilde{\Omega}(1/r^{0.5})$, to be compared to the $\np/\sqrt{r}= 1/r^{0.49}$ upper bound of \citet{AwerbuchBCGM1985}. As far as we know, \cref{thm:Infmain}  is the first improvement over the  $\Omega(1/r)$ bound of \citeauthor{Cleve86}~\cite{Cleve86}.

 \cref{thm:Infmain}  is only applicable when the adversary is able to corrupt all parties but one. However, by grouping parties together, we note that any $\np$-party protocol is a $ \floor{\np/\noh}$-party protocol, for any $s<n$, and thus the following theorem dealing with more versatile corruption strategies follows by simple reduction.  

\begin{theorem}[Main result, fewer  corruptions variant, informal]\label{thm:InfMainLessCorrupt}
	For  $\np$-party $r$-round  coin-flipping protocol with  $(\np/\noh)^k\geq r$ for some $\noh<\np/2$ and $k\in \N$,  there exists an  adversary running in time $(\np/\noh)^k$, corrupting all parties but a subset of size $\noh$, that biases  the output of honest parties  by $1/(\sqrt{r}\cdot \log(r)^k)$. 
\end{theorem}
For instance, if $\np^{k} > r$, by corrupting all parties but a subset of size $\np^{1/2}$, the adversary achieves a bias of $1/(\sqrt{r}\cdot \log(r)^{2k})$. That is, up to a factor of $1/\log(r)^k$ in the bias, we derive the same result as \cref{thm:Infmain}, but with fewer corrupted parties (only all parties but a subset of size $\np^{1/2}$ instead of all parties but one).  

\Inote{When \citet*{CHOR18} is published, we should address it, as it shows that for any constant $r$, there is an optimal coin flipping protocol against adversary controlling all but a constant number of the parties (the large enough constant is a function of $r$). Our result shows, that this cannot happen for small enough constant. }
We prove the above theorems using  the following  two results that we believe to be  of independent interest.

\subsubsection{Augmented Weak Martingales have Large Gap}\label{sec:intro:martingales}
A sequence  $X_1,\dots,X_r$ of random variables is  a \textit{(strong) martingale},  if $\ex{X_i\mid X_{\le i-1}} = X_{i-1}$ for every $i\in [r]$ (letting $X_{\le j} = (X_1,\ldots,X_j)$). \citet{CleveI93} showed that any strong martingale  sequence with $X_1 = \frac12$ and $X_r\in \zo$  has a $1/\sqrt{r}$  gap with  constant probability: with constant probability, $\size{X_{i} - X_{i-1}} \ge \Omega(1/\sqrt{r})$ for some $i\in[r]$. This result is the core of their proof showing that there exists an \emph{inefficient} (fail-stop) attack for any coin-flipping protocol that yields a bias of order $1/\sqrt{r}$ (see \cref{sec:intro:Technique}). The  result of \cite{CleveI93} is used \wrt the \textit{Doob martingale} sequence defined by  $X_i = \ex{f(Z) \mid  Z_{\le i}}$, for  random variables $Z= (Z_1,\ldots,Z_r)$  and  a function $f$ of interest. \enote{Should we cite Doob?}\Nnote{Wikipedia has a paper from the 40's we can cite, but I think it s folklore} To be applicable in a computational setting, we require that  $X_{i} = \ex{f(Z) \mid  Z_{\le i}}$ is an efficiently computable function of  $Z_{\le i}$. In many cases however, including the one considered by \cite{CleveI93}, $\Supp(Z_{\le i})$ is huge, resulting in $X_i$ not being  efficiently computable.

\textit{Weak martingales}, introduced by \citet{Nelson70}, is a relaxation of strong martingales where it is only required that $\ex{X_{i}\mid X_{i-1}} = X_{i-1}$. Namely, the conditioning is only on the value of the preceding variable, and not on the whole ``history''.  As in the case of (strong) martingales, for arbitrary $Z= (Z_1,\ldots,Z_r)$ and a function $f$ of interest, we can consider the Doob-like sequence $X_i = \ex{f(Z) \mid  Z_i,X_{i-1}}$. The support size of the function for computing $X_{i}$ is only of size $\size{\Supp(Z_{i}) \times \Supp(X_{i-1})}$, and we can use discretization to further reduce the support size of $X_i$ (\ie we let $X_{i}$ be a \emph{rounding} of $\ex{f(Z) \mid  Z_{i},X_{i-1}}$).  Hence, if the support of $Z_i$ is small, the computation of the $X_i$'s can be done efficiently. (Discretization is not useful for the (strong) Doob martingale described above, since, even if the support of each individual $Z_1$ is small, even $2$, the domain of $Z_1,\ldots,Z_r$ is typically huge). Unfortunately, it is unclear whether weak martingales have large gaps, and thus we are unable to apply the attack of \citet{CleveI93} using such a sequence.

We prove that a slightly different variant of the Doob construction results in a sequence that is efficiently computable and has a large gap, at the same time. A sequence $X_1,\dots,X_r$ of random variables is a \emph{sum-of-squares-augmented weak martingales}, if  $\ex{X_{i} \mid X_{i-1}, \sum_{j\in[ i-1]} (X_j - X_{j-1})^2}=X_{i-1}$.   Namely, $X$ has the ``martingale property'' when conditioning on some small amount of information about the past. For such a sequence, we prove the following result:

\begin{theorem}[Informal]\label{theorem:SosMartingalesInf}
	Let $X_1,\ldots,X_r$ be a sequence of sum-of-squares-augmented weak martingale with $X = 1/2$ and $X_r\in \zo$, then   $$\pr {\exists i\in [r] \colon \size{X_i - X_{i-1}} \ge 1/\sqrt{r}} \in \Omega(1).$$   
\end{theorem}

We  prove that the above holds for a \emph{rounded} variant of $X_i$, \ie $X_i$ are rounded to the closest multiplicative of some  $\delta>0$. 

Consider the sequence of sum-of-squares-augmented weak martingales defined by the Doob-like sequence $X_i = \ex{f(Z) \mid  Z_i,X_{i-1},\sum_{j\in[ i-1]} (X_j - X_{j-1})^2}$, for arbitrary  $Z = (Z_1,\ldots,Z_r)$  and  a function $f$ of interest. If the support of the $Z_i$'s small, the computation of the (rounding of) $X_i$'s can be done \emph{efficiently}.  This efficiency plays a critical role in our attack on coin-flipping protocols, allowing us, in some cases, to mount an efficient variant of the attack of \cite{CleveI93}.

Our proof actually yields the following stronger statement. 
\begin{theorem}[Informal]\label{theorem:SosMartingalesSSSstrongInf}
	Let $X_1,\ldots,X_r$ be a sequence of sum-of-squares-augmented weak martingales with $X = 1/2$ and $X_r\in \zo$, then 
	
	$$\Pr\bigr[\sum_{i\in[ r]} (X_i - X_{i-1})^2 \ge 1\bigr] \in \Omega(1).$$
\end{theorem}
Namely, the  sum-of-squares is constant with constant probability.  In particular, the probability that $\size{X_i - X_{i-1}} \ge 1/\sqrt{r}$ , for some $i$,  is also constant, implying \cref{theorem:SosMartingalesInf}.    But  \cref{theorem:SosMartingalesSSSstrongInf} yields a stronger result: if we are guaranteed that all gaps are at most $1/\sqrt{r}$ (\ie  $\size{X_i - X_{i-1}} \in O(1/\sqrt{r})$ for all $i$), then \cref{theorem:SosMartingalesSSSstrongInf} implies that, with  constant probability, the sequence has a \emph{linear} number of $1/\sqrt{r}$-gaps (as opposed to only one such gap guaranteed by \cite{CleveI93}). 

 Our proof for \cref{theorem:SosMartingalesSSSstrongInf} is surprisingly simple, and does not follow the more complicated approach of \citet{CleveI93}.\footnote{To be fair, \citet{CleveI93} derive their result  by  proving an Azuma-like tail inequality for bounded strong martingales that  have large gap with only small probability, a bound that we do not prove  here. }

\subsubsection{Oblivious Sampling via Differential Privacy}\label{sec:intro:DP}
   Consider the following $r$-round game in which the goal is to maximize the revenue of the chosen party: in the beginning, a party $H$ is drawn  uniformly  from $\H$ (for  $\H$ being a finite set of parties). In each round $i$, values $\set{s_i^h\in [0,1]}_{h\in \H}$ are assigned to the  parties   of  $\H$, and  the values  of all parties but $H$, \ie  $\set{s^{h}_i}_{h\in \H\setminus \set{H}}$,   are published. Seeing the published values, you can either decide to \emph{abort}, and then party $H$ is rewarded with  (the unseen) value $s_i^H$, or to continue to the next round. If you never choose to abort, then party $H$ is rewarded with $s_r^H$ (the value of the last round). Your goal is to get a reward as close to the optimal value $\gamma = \max_i \set{s_i \eqdef  \eex{h\la \H}{s_i^h} }$.  To make the game reasonable, it is guaranteed that the values assigned to the parties in each round  are \emph{similar}: $\size{s_i^h - s_i}\le \sigma$ for every $h \in \H$. Namely, the individual values are $\sigma$-close to the mean.  
   
   We will be interested in a distributional variant of the above game in which the  values of $\set{s_i^h}$ are not fixed, but rather drawn from some underlying distribution (in our setting, the values of $\set{s_i^h}$ will be induced by the randomness of the attacked coin-flipping protocol), while satisfying the above guarantees with regards to $\gamma$ and $\sigma$ with good enough probability. \enote{I think it would be instructive to make the connection to the actual way we use the attack clearer. Specifically, we should probably say that the $s_i$s stand for the changes in the probability of the average back-up value (or just game-value). }\Nnote{they denote the gap between the average of the backups of two distinct sets of tuples the honest party belongs to} We refer to the resulting game as an \emph{oblivious sampling game} with parameters $r$, $\size{\H}$, $\gamma$, and $\sigma$. An aborting  strategy for the above game can  only depend on the game parameters (\ie $r,\size{\H},\gamma,\sigma$) and the values published online.

   The simplest aborting strategy for such a game is to abort if the average of all other parties, \ie $\set{s_{h}}_{h\in \H \setminus \set{H}}$, is  larger than (roughly) $\gamma - \sigma$. The reward of such a strategy is roughly $\gamma - \sigma$, which is useless if $\sigma\ge \gamma$. As we show next, this linear loss in $\sigma$ is inherent for this strategy; consider  a deterministic threshold strategy that aborts if $s_i^{\setminus h} = \eex{h' \la \H \setminus h}{s_i^{h'}} \ge \tsh$ for some threshold $\tsh \in[0, \gamma]$. Namely, an aborts occurs if the average value at hand in a given round is greater than $\tsh$.  Consider the game defined by  $\H = [r-1]$, $s_r^h= \gamma$ for all $h$, and for  $i\in [r-1]$:  $s_i^h = \tsh - \sigma$ if  $i=h$, and $\tsh$ otherwise.  It follows that for every value of $h$, the strategy seeing the values  of $\set{s^{\setminus  h}_{i}}$ aborts at round $h$,  and gets reward $\tsh- \gamma$. Hence, the reward of this strategy is $\tsh -\sigma \le \gamma -\sigma$.

   We show that  using a \textit{differentially private} mechanism, and in particular adding Laplace  noise to the estimated revenue $s_i^{\setminus h} = \eex{h' \la \H \setminus h}{s_i^{h'}}$,  significantly improves  upon the above deterministic strategy. By introducing such noise, the aborting decision is less correlated to the choice of the random party $H$.  More accurately, the value of $H$ is $\sigma$-\emph{differentially private}, according to the definition of \citet*{dwork2016calibrating}, from the aborting decision, and thus we avoid the pitfalls caused by strong correlation between  $H$ and the aborting round, as illustrated by the above example for the deterministic threshold strategy. We exploit this  ``privacy'' guarantee to prove the following improvement in the expected reward.

   \begin{theorem}[Informal]\label{thm:LaplaceInf} 
   	For every oblivious sampling game, the  randomized strategy that adds  Laplace noise in every round (whose magnitude depends on the game parameters) to $s_i^{\setminus h}$, and  aborts if the result is greater than $\gamma/2$, achieves expected reward   $\gamma/2  - \sigma^2$.   
   \end{theorem}
Namely, the penalty for having imperfect similarity  is reduced from $\sigma$ to $\gamma/2 +\sigma^2$, a significant improvement when $ \gamma< \sigma <1$.  We also  prove a generalization  of the above theorem where each party has a different similarity  guarantee.

   \subsection{Our Techniques}\label{sec:intro:Technique}
   
     Below, we  describe   the approach for proving  \cref{thm:Infmain} using \cref{theorem:SosMartingalesInf,thm:LaplaceInf}.  We do not discuss here the proofs of these theorems, but we do explain in \cref{sec:intro:CompDoob} why the weak martingale used by the attack is computable by an efficient uniform algorithm.

		Let $\Pi$ be an $r$-round $\np$-party coin-flipping protocol and let $\out$ denote the (always common) output of the parties in a  random honest execution. By definition,  $\out\in \zo$ and $\ex{\out}=1/2$.   Our goal is to obtain an efficient attacker that, by controlling $\np-1$ of the parties, biases the honest parties'  output by  $1/\sqrt{r}$ (we ignore $\log$ factors).	We start by describing the $1/\sqrt{r}$ inefficient attack of \citet{CleveI93}.

\subsubsection{\citeauthor{CleveI93}'s Inefficient Attack} \label{sec:intro:CI}
Let $\np=2$ and let $(\Pc_0,\Pc_1)$ be the parties of $\Pi$. Let $T_1,\ldots,T_r$ denote the messages in a random execution of $\Pi$. Let  $X_i = \ex{\out \mid T_{\le i} }$; namely, $X_i$ is the expected outcome of the protocol given the first $i$ messages $T_{\le i} = T_1,\ldots,T_i$. It is easy to see that $X_1,\ldots,X_r$ is a (strong) martingale sequence. Hence, the result of \cite{CleveI93} described in \cref{sec:intro:martingales} yields that  (omitting  absolute values and constant factors)
\begin{align}\label{eq:CIGap}
\text{Jump:}& \hfill  &\pr{\phantom{,}\exists i\in [r] \colon  X_i- X_{i-1} \ge 1/\sqrt{r}\phantom{,}} \in \Omega(1).
\end{align}

  \paragraph{Backup values.} 	
  For $\party \in \zo$, let the \emph{backup value} $Z_i^\party$ denote the output of party $\Pc_\party$ if party $\Pc_\bj$ aborts \emph{after} the \ith message was sent, letting $Z^\party_r$ be the final output of $\Pc_\party$ (if no abort occurs).  Using this notation, $\ex{Z_{i}^\party \mid  T_{\le i}}$ is the expected outcome of $\Pc_\party$ if $\Pc_\bj$ aborts after the \ith round. 
  We can assume \wlg that
\begin{align}\label{eq:CIGapZZ}
\text{Backup value follows game value:}& \hfill  & \pr{\exists i\in [r] \colon  \size{X_{i} -  \ex{Z_{i}^\party \mid  T_{\le i}}} \ge 1/\sqrt r} \in o(1).
\end{align}
for both $\party \in \zo$. Otherwise, the attacker controlling $\Pc_\bj$ that computes $X_{i}$ and $\ex{Z_{i}^\party \mid  T_{\le i}}$ for each round $i$, and aborts if $X_{i}-\ex{Z_{i}^\party \mid  T_{\le i}}\ge 1/\sqrt{r}$, would bias $\Pc_\party$'s output towards $0$ by $1/\sqrt{r}$.\footnote{To be more precise, at least one of two attacks would succeed, depending on the aimed direction of the bias.}

 \paragraph{The martingale attack.} 	
 The above two observations yield the following attack.  From \cref{eq:CIGap,eq:CIGapZZ}, it follows that \wlg   
\begin{align}\label{eq:CIGapZ}
\text{Attack slot:}& \hfill  & \pr{\exists i\in [r] \colon \text{$\Pc_\party$ sends the \ith message} \land     X_{i} -  \ex{Z_{i-1}^\bj \mid  T_{\le i} }  \ge 1/2\sqrt{r}} \in \Omega(1).
\end{align}
This yields the following attack for party $\Pc_\party$ to bias the output of party $\Pc_\bj$ towards zero. Before sending the \ith message $T_i$,  party $\Pc_\party$ aborts if  $ X_{i} - \ex{Z_{i-1}^\bj \mid  T_{\le i} } \ge 1/2\sqrt{r}$. By \cref{eq:CIGapZ}, under this attack, the output of $\Pc_\bj$ is biased towards zero by $\Omega(1/\sqrt{r})$.\footnote{In more detail, assume for simplicity that $\Pc_0$ sends the messages $T_1,T_3,\dots$ and $\Pc_1$ sends the messages $T_2,T_4,\dots$. For at least one  party $\Pc_\party$, \cref{eq:CIGapZ} holds when limiting $i$ to be a round where $\Pc_\party$ is supposed to send the \ith message. The above attack is effective when executed by the relevant party.}

The clear limitation of the above attack is that, in many cases, the values of both $X_i=\ex{\out \mid T_{\le i} }$ and $\ex{Z_i^\bj \mid T_{\le i}}$  are \emph{not} efficiently computable (given  $T_{\le i}$). Indeed (assuming the existence of oblivious transfer), the above $\Theta(1/\sqrt{r})$ lower bound does not  hold  for  $\np< \loglog r$  \cite{BHLT17,HaitnerT17,MoranNS16}. 

\subsubsection{Towards an Efficient Attack via  Augmented Weak Martingales}
 The first step towards making the above attack efficient is \emph{not} to  define the $X_i$'s as a function of the  transcript. Indeed, even given the first message $T_1$, computing  $\ex{\out \mid T_1}$ might involve inverting a one-way function. Our solution is to define $X^\party_i= \ex{\out \mid Z^\party_{\le i}}$; namely, the expected outcome given $\Pc^\party$'s  backup values.  The immediate advantage is that the backup values are only bits. Thus, $X^\party_1$ has only two possible values, and  computing it from $Z_1$ can be done efficiently. Yet, for  large values of $i$, the computation of $X^\party_i$ (depending on $Z_1^\party,\ldots,Z_i^\party$) might still be infeasible. 
    
Thankfully, our new result for sum-of-squares-augmented weak martingales (\cref{theorem:SosMartingalesInf}) circumvents this problem. Let $f(Z_1^\party,\ldots,Z_r^\party) = \ex{\out \mid Z^\party_{\le r}}$.  By definition, it holds that $f(Z_1^\party,\ldots,Z_r^\party) = Z^\party_r \in \zo$, and thus $\ex{f(Z_1^\party,\ldots,Z_r^\party)} = 1/2$. \cref{theorem:SosMartingalesInf} yields that for the Doob-like  sequence $X_i^\party= \ex{\out \mid  Z^\party_{i},X_{i-1}^\party,\sum_{j\in [i-1]} (X_j^\party - X_{j-1}^\party)^2}$, it holds that  (again, omitting  absolute values and constant factors)
    
 \begin{align}\label{eq:AWMGap}
   \text{Jump:}& \hfill  &  \pr{\exists i\in [r] \colon   X^\party_{i} - X^\party_{i-1} \ge 1/\sqrt{r}} \in \Omega(1).
 \end{align}
    
Using a rounded variant of the $X^\party_i$'s, the value of $X^\party_i$ is only a function of $\size{\Supp(Z^\party_{i})} \cdot r^2  \in O(r^3)$ bits, and thus can be computed efficiently.  Namely, the martingale attack of \cite{CleveI93}  (\ie aborting in the event of an observed gap) \wrt this definition of $X_i$ is now efficient.  Similarly to   \cite{CleveI93}, we obtain an $\Omega(1/\sqrt{r})$ attack  if   
\begin{align}\label{eq:eq:AWMGapZ}
\text{Attack slot:}& \hfill  & \pr{\exists i\in [r] \colon X^\party_{i}  -  \ex{Z_{i-1}^\bj  \mid Z^\party_{i},X^\party_{i-1}, \sum_{j\le i-1} (X^\party_{j} - X^\party_{j-1})^2} \ge 1/2\sqrt{r}} \in \Omega(1).
\end{align}
The coin-flipping protocols of \cite{BHLT17,HaitnerT17,MoranNS16} show that the equation above does not hold in general. Nevertheless, we show that  (for a suitable variant of) the above inequality does hold for the case $\np \ge  r$, and thus the ``martingale'' attack achieves the desired bias for this case. The case $\np^{k} \ge  r$ for $k \ge 2$ is significantly more complex, but follows the same principle.  Details below.

\subsubsection{An Efficient Attack for $\np = r$}\label{sec:intro:tec:3}
 Let $(\Pc_1,\ldots,\Pc_\np)$ be the parties of $\Pi$. For $\party\in [\np]$,  let   $Z_i^\party \in \zo$  be the  output   (backup value) party $\Pc_\party$ outputs if \emph{all} other parties abort right after  the \ith round, and for $\cs \subseteq [\np]$ let  $Z_i^\cs =  \frac{1}{\size{\cs}} \cdot \sum_{s \in \cs} Z_i^s$.  For a  subset   $\cs\subseteq [\np]$,   consider the sequence of augmented weak martingales $X^\cs_i = \ex{\out \mid  Z^\cs_{i},X_{i-1}^\cs,\sum_{j\in [i-1]} (X_j^\cs - X_{j-1}^\cs)^2}$.  As before,  with constant probability  $X^\cs_{i} - X^\cs_{i-1} \ge 1/\sqrt{r}$  for some $i\in [r]$.   Hence, \wlg,  
 \begin{align}\label{eq:AWMGapZMany}
  \text{Jump:}& \hfill  &  \pr{\exists i\in [r] \colon X^\cs_{i}  - Z^\cs_{i-1}    \ge 1/2\sqrt{r}} \in \Omega(1).
 \end{align} 
 A crucial observation, and the reason why considering a number of parties that is \emph{linear} in the round complexity is rewarding, is that, with high probability over the choice of $\cs$ of size $\np/2$, it holds that
 \begin{align}\label{eq:intro:appx}
  \text{Similar backup values:}& \hfill  &   \forall i\in[r] \colon  \ \ Z^\cs_i = Z_i^{\ocs} \pm 1/3\sqrt{r}.
 \end{align} 
 Namely, $Z^\cs_i $ is a good estimation for $Z^\ocs_i$, for all rounds $i\in [r]$ \emph{simultaneously}.\footnote{Actually, this requires $\np = r\log r$, but we ignore such log factors in  this informal discussion.}  
 
 Indeed, since $\cs$ is chosen at random,  $Z_i^\cs$ ($=\frac{1}{\size{\cs}} \cdot \sum_{s \in \cs} Z_i^s$) is  a $1/3\sqrt{r}$ approximation of $Z_i^{[\np]}$ and thus of   $Z_i^{ \ocs}$. Fix such a good set $\cs$. The following martingale attack biases the output of a random  party  $\Pc_h$ not in $\cs$ (\ie  $h\la {\ocs}$) towards zero.  In the \ith round, the attacker  aborts all parties but $\Pc_h$ if  $ X^\cs_{i} - Z_{i-1}^\cs  \ge 1/6\sqrt{r}$. \cref{eq:AWMGapZMany,eq:intro:appx} implies that  the above adversary biases the output of $\Pc_h$ towards zero by $\Omega(1/\sqrt{r})$.

    \subsubsection{An Efficient Attack for $\np^{k} \ge r$  via  Differentially Private   Sampling}\label{sec:intro:tec:4}
    We describe the attack for $\np^2 \ge r$, and then briefly highlight the extension for $k\ge 3$. 
    
    A critical part of the above  attack  for $\np = r$ (stated in \eqref{eq:intro:appx}) is that for a random (and thus for some) subset $\cs \subseteq [\np]$ of size $\np/2$, it holds that $Z^\cs_i$  is at most  $O(1/\sqrt{r})$-far from $Z_i^{\ocs}$. This is not the case for  $\np^2 = r$, where  we are only guaranteed that $Z^\cs_i$  is at most  $O(1/\sqrt{\np}) = O(1/\sqrt[4]{r})$-far from $Z_i^{\ocs}$, a too-rough approximation for our needs, since the error is larger than the potential gain of $O(1/\sqrt{r})$). 
    
    Our  solution is to consider the \emph{joint} backup values for \emph{pairs} of parties.  That is, the joint output of such a pair given that all other parties abort. Considering the pairs' backup values, however, raises a different  problem. The adversary can no longer examine the  values of a random large subset $\cP \subsetneq\binom{[\np]}{2} $ of backup values, as we described in the case $\np=r$, since \emph{each}  party in $[\np]$ (and, in particular, the honest party) takes part in $\Theta(1/\np)$ fraction of $\cP$, with high probability. Rather, we let the attacker examine the backup values of the pairs $\binom{\cs}{2}$, for some subset  $\cs \subsetneq [\np]$. If (the average of) these backup values are a good approximation for the backup value of pairs that contain the honest party, then the previous aborting strategy results in a bias of suitable magnitude. If not, then  we can employ a different type of attack using differentially private sampling  (\cref{thm:LaplaceInf}). More details follow.
   
    For a pair $p  = ( j_1,j_2)\in \binom{[\np]}{2}$,  let   $Z_i^p \in \zo$  be the  joint output   (backup value) of the parties  $\Pc_{j_1}$ and $\Pc_{j_2}$, if \emph{all} parties but them abort right after  the \ith round. For $\cP \subseteq  \binom{\np}{2}$,  let  $Z_i^\cP =  \frac{1}{\size{\cP}} \cdot \sum_{p \in \cP} Z_i^p$. Consider  the sequence of augmented weak martingales $X^\cs_i = \ex{\out \mid  Z^{\binom{\cs}{2}}_{i},X^\cs_{i-1},\sum_{j\in [i-1]} (X^\cs_{j} - X^\cs_{j-1})^2}$, for some  subset   $\cs\subseteq [\np]$. As before,  with constant probability  $X^\cs_{i+1} - X^\cs_i \ge 1/\sqrt{r}$  for some $i$. Assuming that 
   \begin{align}\label{eq:intro:2:1}
   \text{Similar backup values:}& \hfill  &  Z_i^{\binom{\cs}{2}}  =    Z_i^{\cs \times \overline{\cs}}  \pm o(1/\sqrt{r}),
   \end{align} 
   for every $i$. Namely, the average backup values of pairs of parties seen by an attacker controlling all parties in $\cs$, is very close to the average of the backup values of pairs containing one party in $\cs$ and one party not in $\cs$. Similarly to the case $\np = r$, the above assumption enables the following martingale attack biasing the output of a random party $\Pc_h$ not in $\cs$ (\ie $h\gets \overline{\cs}$) towards zero by $\Omega(1/\sqrt{r})$. In the \ith round, if  $X^\cs_{i} - X^\cs_{i-1} > 1/2\sqrt{r}$ the attacker aborts  all parties but $(\Pc_h,\Pc_{s})$, for a random $s \la\cs$.

  Unlike  the  case $\np = r$, \cref{eq:intro:2:1} might be false (for any set $\cs$).   Fortunately, if this happens, we can mount a different attack, described below.

   \paragraph{The differentially private sampling attack.} 
      Assume for simplicity that  
    \begin{align}\label{eq:intro:2:p}
  \text{Non-similar backup values:}&   \hfill  & \pr{\exists i\in [r]\colon Z_i^{\binom{\cs}{2}}  -   Z_i^{\cs \times \overline{\cs}}  > 1/\sqrt{r}} \in \Omega(1).
   \end{align} 
   This calls for the following attack biasing the output of a random honest party $\Pc_h$, for $h\la \cs$,  towards zero. For a pair-subset $\cP \subseteq \binom{[n]}{2}$, let $\cP \rmv{h}$  stand for all pairs in $\cP$ that do not include $h$. In the \ith round,  the attacker checks whether  $G_i^{\setminus  h} = Z_i^{\binom{\cs}{2} \rmv{h}} - Z_i^{(\cs \times\ocs)\rmv{h}}  > 1/2\sqrt{r}$. If so, it  aborts  all parties but  $(\Pc_h,\Pc_{s})$ for a random $s \la\cs$.   The attack performs well if the backup values of the corrupted parties are a good approximation of the expected value of the honest party output.  In particular, if for every $h$ and $i$:
   \begin{align}\label{eq:intro:25:p}
   \text{Strong gap similarity:}&   \hfill  & G_i^{\setminus  h} = G_i^h \pm o(1/\sqrt{r}),
    \end{align}
    for $G_i^h = Z_i^{\set{h} \times (\cs\setminus\set{h}) }- Z_i^{ \set{h} \times \ocs}$. Unfortunately,  \cref{eq:intro:25:p} might be false, without yielding any useful consequences. Rather, we can only assume the weaker guarantee that 
   \begin{align}\label{eq:intro:3:p}
  \text{Weak gap similarity:}&   \hfill  & G^{\setminus h}_i = G_i^h \pm o(1/\sqrt{\np})= G_i^h \pm o(1/\sqrt[4]{r}),
   \end{align} 
    for every $h$ and $i$. Indeed,  if \cref{eq:intro:3:p} does not hold, then (w.l.o.g) for some party $h'  \in \cs $  it holds that $Z_i^{ \set{h} \times \ocs} \neq Z_i^{ \set{h'} \times \ocs}\pm  \Theta(1/\sqrt[4]{r}) $.  That is, when restricting our attention to the $\np-1$ pairs containing $h'$, the gap is $\Theta(1/\sqrt{\np})$. Such gap  yields that an attack in  the spirit of the one used for  the $\np = r$ case induces a large bias on an honest party chosen randomly from $\ocs$.

    The guarantee of \cref{eq:intro:3:p} does not suffice for the simple attack described below \cref{eq:intro:2:p} to go through. Roughly, the reason is that the approximation error (\ie $ o(1/\sqrt[4]{r})$) is larger than the expected gain of $\Omega(1/\sqrt{r})$. Indeed, we find ourselves  in the setting of the oblivious sampling game considered in \cref{sec:intro:DP}, letting $\gamma = 1/\sqrt{r}$, $s_i^h = G_i^h $ and $s_i = G_i =  \eex{h\la \H}{s_i^h} $, and $\sigma = \size{s_i - s_i^h} \in o(1/\sqrt[4]{r}) > 1/\sqrt{r}$. As explained in \cref{sec:intro:DP}, a threshold  deterministic attack for this sampling game, i.e.~seeing the values of $\set{s_i^h}_{i\in [r]}$ one by one, until one decides to abort, might achieve no reward for random $h$, and neither for any fixed $h$. In particular, it might hold that   $s_i^h=0$ for $i$ being the aborting round.

      Fortunately, since we are in the setting of the oblivious sampling game, \cref{thm:LaplaceInf} yields that  randomized aborting online strategy that adds noise to its halting decision,  in every round, performs significantly better. Specifically, the strategy that adds the right Laplace noise to $s_i^{\setminus h}= \eex{h' \la \H \setminus h}{s_i^{h'}}$, and aborts if the result is greater than  $\gamma/2$, achieves expected revenue $\gamma/2 - \sigma^2 =  1/2\sqrt{r} -  o(1/\sqrt[4]{r})^2 > 1/4\sqrt{r}$.  The above holds for a random $h$, and thus for some fixed $h$ as well. It follows that the induced attacker on the protocol, controlling all parties by $\Pc_h$, and applying the same randomized aborting strategy, obtains   $Z_i^{\set{h} \times (\cs\setminus\set{h}) }- Z_i^{ \set{h} \times \ocs} \ge 1/4\sqrt{r}$ where $i$ denotes the aborting round.  This translates, using the same means as in the deterministic threshold attack for choosing the non aborting pair, into a $1/4\sqrt{r}$-bias of $\Pc_h$'s output.

	Intuitively, the point of using the differentially private sampling mechanism is to avoid identifying the choice of honest party. That is, the effect of the back-up values of pairs containing any single party on the adversary's decision to trigger the attack is diminished (\cref{eq:intro:3:p} keeps the sensitivity of this decision small).

   \paragraph{The case $\np^{k} \ge r$ for $k\ge 3$.}
    To begin, assume that $k=3$ (\ie $\np^{3} \ge r$).  For such value of $\np$, it holds that  $1/\sqrt{\np} = 1/{\sqrt[6]{r}} \gg 1/\sqrt[4]{r}$. Thus, the promise $G_i = G_i^h \pm o(1/\sqrt{\np}) $  does not suffice for the  differentially private based attack to go through. Rather, we need to assume that  $G_i = G_i^h \pm o(1/\sqrt[4]{r}) = G_i^h \pm o(1/\np^{3/4})$. We show that if the latter does not hold, the attacker can fix a party and never abort it (\ie  we restrict the subset of all backup values to those containing  this party) we are essentially in the setting of   $\np^{2} \ge r$. Namely, either we have differentially private based attack, or we have a martingale attack (both \wrt the above fixing of a never aborting party). 
    
    For larger values of  $k$, we iterate the above, fixing non-aborting parties one after the other, until one of the  differentially private based attacks or the  martingale attack go through.

\subsubsection{Computing  Doob-like Weak Martingales} \label{sec:intro:CompDoob}
In \cref{sec:intro:martingales}, we claimed that  if the support size of the $Z_i$'s is small, then  the sum-of-squares-augmented weak martingales $X_1,\ldots,X_r$ defined by the rounded  Doob-like sequence $X_i = \RND(\ex{f(Z) \mid  Z_i,X_{i-1},\sum_{j\in[ i-1]} (X_j - X_{j-1})^2})$ can be efficiently computable, where $\RND$ is a  small support rounding function. We use this guarantee above to argue that our attack is efficient. While this claim trivially holds when considering non-uniform algorithms,  the argument for \emph{uniform} algorithms is more subtle, and since we  believe it to be of independent interest, we highlight it below. 

 For simplicity, we focus on  the weak martingales defined by the Doob-like sequence $X_i = \RND(\ex{f(Z) \mid  Z_i,X_{i-1}})$. Consider the mappings $\chi_1,\ldots,\chi_r$ inductively defined by  $\chi_i(z) =  \RND(\ex{f(Z) \mid  Z_i = z_i,X_{i-1} = \chi_{i-1}(z)})$. It is easy to verify that  $X_i = \chi_i(Z_{\le r})$, and since each of these mappings has a small description, the sequence  $X_0,\ldots,X_r$  can be computed from $Z_1,\ldots,Z_r$ by a small circuit holding these mappings.   Arguing that the above can be performed by an efficient (uniform) \emph{algorithm}, things get slightly more involved.  While we can  estimate well the mapping $\chi_1(z) = \RND(\ex{f(Z) \mid  Z_1= z_1})$ via sampling, even a small unavoidable error in the estimation might cause a larger error in the estimation of $\chi_2 = \RND(\ex{f(Z) \mid  Z_2 = z_2,X_{i-1} = \chi_1(z)})$. This is since the dependency on $\chi_1$ is in the  conditioning, and thus estimating $\chi_2$ using an estimate of $\chi_1$ amplifies the error. This might lead to very large errors when trying to use the estimated mapping for calculating $X_i$'s of large indices.

\newcommand{\Xu}{\widehat{X}}
\newcommand{\Mu}{\mu}

So rather, we consider the efficiently computable sequence $\Xu = (\Xu_1,\ldots,\Xu_r)$ defined by $\Xu_i = \Mu(Z_{\le i})$,  for $\Mu_1,\ldots,\Mu_r$ being an  \emph{estimation} for  $\chi_1,\ldots,\chi_r$ done via sampling.\footnote{The mapping $\Mu_1,\ldots,\Mu_r$  are constructed   iteratively. After constructing $\Mu_1,\ldots,\Mu_{i-1}$,  the value  of $\Mu_i(z)$ is set by approximating  via sampling (a rounding of) $\ex{f(Z) \mid Z_i = z_i, \Mu_{i-1}(Z) = \Mu_{i-1}(z)}$.} Since $\Xu$ is defined \wrt the approximated mappings, it is  a weak martingale, even if the approximated  mappings \emph{wrongly approximate} the real ones. The reason is that the quality of $\Xu_i$ as a ``Doob-like sequence'' --- i.e.~how well it approximates $\ex{f(Z) \mid Z_i, \Xu_{i-1}}$ --- is not affected by the quality of $\Mu_1,\ldots,\Mu_{i-1}$,  and thus errors do not accumulate.  Taking the same approach for the sum-of-squares-augmented weak martingales, our construction yields that  with high probability over the choice of the estimated mappings $\Mu_1,\ldots,\Mu_r$, the sequence $\Xu$ satisfies all the properties required by \cref{theorem:SosMartingalesInf}, and thus we can  invoke our  attack  using this  sequence.


\subsection{Related Work}

\subsubsection{Coin Flipping}

A coin-flipping protocol  is  $\delta$-fair, if no efficient attacker (controlling any number of parties) can bias the output (bit) of the honest parties by more than $\delta$. 


\paragraph{Upper bounds.} 
\citet{Blum83} presented a two-party two-round coin-flipping protocol with  bias $1/4$. \citet{AwerbuchBCGM1985} presented an $\np$-party $r$-round protocol with bias $O(\np/\sqrt{r})$ (the two-party case appears also in  \citet{Cleve86}).  This was improved to (almost) $O(1/\sqrt{r})$ in \cite{BeimelOO10, CHOR18}, for the case where the fraction of honest parties is constant. \citet*{MoranNS09} resolved the two-party case, presenting a two-party $r$-round  coin-flipping protocol with bias $O(1/r)$. \citet{HaitnerT14}  resolved the three-party case up to poly  logarithmic factor,  presenting a three-party coin-flipping protocol with bias  $O(\polylog(r)/r)$. \citet{BHLT17} constructed an $\np$-party $r$-round  coin-flipping protocol with bias $\widetilde{O}(\np^3 2^\np/r^{\frac{1}{2}+\frac{1}{2^{\np-1}-2}})$. In particular,  their four-party coin-flipping protocol the bias is $\widetilde{O}(1/r^{2/3})$, and for $\np = \log \log r$ their  protocol has bias smaller than  \cite{AwerbuchBCGM1985}.

For the case where less than $2/3$ of the parties are corrupt, \citet{BeimelOO15} have constructed 
an $\np$-party $r$-round  coin-flipping protocol with bias  $2^{2^k}/r$, tolerating up to $t=(\np+k)/2$ corrupt parties.  \citet{AO16} constructed an $\np$-party $r$-round coin-flipping protocol with bias $\widetilde{O}(2^{2^\np}/r)$, tolerating up to $t$ corrupted parties, for constant  $\np$ and $t<3\np/4$.

\paragraph{Lower bounds.} 

\citet{Cleve86} proved that for every $r$-round two-party coin-flipping protocol there exists is an efficient adversary that can bias the output by  $\Omega(1/r)$.  \citet{CleveI93} proved that   for every $r$-round two-party coin-flipping protocol there exists an inefficient  fail-stop adversary that  biases the output by  $\Omega(1/\sqrt{r})$. They also showed that a similar  attack  exists also if the parties have access to an ideal  commitment scheme. All above bounds extend to multi-party protocol  (with no honest majority) via a simple reduction. 

A different line of work examines the  minimal assumptions required to achieve an  $o(1/\sqrt{r})$-bias  two-party coin-flipping protocols. \citet{Dachman11} have shown that any fully black-box construction of $O(1/r)$-bias two-party protocols based on  one-way functions with $r$-bit input and output needs $\Omega(r/\log r)$ rounds. \citet{DachmanMM14} have shown that there is no fully black-box and function \textit{oblivious}
construction of  $O(1/r)$-bias two-party protocols from one-way functions (a protocol is function oblivious if the outcome of the protocol is independent of the choice of the one-way function used in the protocol).   Very recently,  \citet{HaitnerMO18} have used  an attack in the spirit of the one used in this paper, together with the dichotomy result of \citet{HaitnerNOSS18}, to prove that key-agreement is a necessary assumption for \emph{two-party} $r$-round coin-flipping  protocol of bias $o(1/\sqrt{r})$, as long as $r$ is independent of the security parameter.

A different type of lower bound was given by \citet{CohenHOR2016}. They focused on the communication model required for fully secure computation,  and in particular showed that in the setting where broadcast is impossible   (\eg peer-to-peer network  with $1/3$ fraction  of dishonest parties), there exists no many-party  coin-flipping protocol with  non-trivial bias (\ie noticeably smaller then $1/2$).

\subsubsection{$1/p$-Secure Protocols}
 \citet{Cleve86} result implies that for many functions fully-secure computation without an honest majority is not possible. \citet{GordonK10} suggested the notion of $1/p$-secure computation to bypass this impossibility result. Very informally, a  protocol is $1/p$-secure if every poly-time adversary can harm the protocol with probability at most $1/p$ (e.g., with probability $1/p$ the adversary can learn the inputs of honest parties, get the output and prevent the honest parties from getting the output, or bias the output).
\citet{GordonK10} constructed for every polynomial $p(\kappa)$ (where $\kappa$ is the security parameter) an efficient two-party  $1/p(\kappa)$-secure protocol for computing a function $f$, provided that  the size of the domain of at least one party in $f$ or the size of the range of $f$ is bounded by a polynomial. 
\citet{BLOO11} generalized this result to multi-party protocols when the number of parties is constant -- for every function $f$ with $O(1)$ inputs such that the domain of each party (or the size of the range of $f$) is bounded by a polynomial and for every polynomial $p(\kappa)$, they presented  
an efficient $1/p(\kappa)$-secure protocol for computing the function.

\citet{GordonK10} and \citet{BLOO11} also provided impossibility results explaining why their protocols require bounding the size of the domain or range of the functions. Specifically, \citet{GordonK10} described a two-party function whose size of domain of each party and size of range is $\kappa^{\omega(1)}$ such that this function cannot be computed by any poly-round protocol achieving $1/3$-security. \citet{BLOO11} used this result to construct a function
$f\colon \zo^{\omega(\log n)}\rightarrow \kappa^{\omega(1)}$ (\ie a function with $\omega(\log n)$ parties where the domain of each party is Boolean) such that this function cannot be computed by any poly-round protocol achieving $1/3$-security. They also showed the same impossibility result for a function with $\omega(1)$ parties where the domain of each party is bounded by a polynomial is the security parameter. We emphasize that these impossibility results do not apply to coin-flipping protocols, where the parties do not have inputs.

\subsubsection{Complete Fairness Without Honest Majority}
\citet{Cleve86} result was interpreted as saying that non-trivial functions cannot be computed with complete fairness without an honest majority. In a surprising result, \citet{GordonHKL11}
have shown that the millionaire problem with a polynomial size domain and other interesting functions can be computed with complete fairness in the two-party setting. The two-party functions that can be computed with complete fairness were further studied in~\cite{AsharovLR13,Asharov14,Makriyannis14,AsharovBMO15}; in particular,  \citet{AsharovBMO15} characterized the Boolean functions that can be computed with complete fairness.
\citet{GordonK09} have studied complete fairness in the multi-party case and constructed completely-fair protocols for non-trivial functions in this setting.



\subsubsection{Differential Privacy}
Differential privacy, introduced by \citet{dwork2016calibrating}, provides  a provable guarantee of privacy for data of individuals. Assume there is a database containing private information of individuals and there is an algorithm computing some function of the database. We say that such randomized algorithm  is differentially private if changing the data of one individual has small affect on the output of the algorithm. For example, if, for a database $D$, a function $f(D)$ returns a numerical value in $[0,1]$, then an algorithm returning $f(D)+\text{\em noise}$, where {\em noise} is distributed according to the Laplace distribution (with suitable
parameters), is a differentially private algorithm. Since the introduction of differential privacy in 2006, many algorithms satisfying differential privacy were introduced, see, \eg \citet{DworkRoth14}. In this work we use differential privacy (i.e., Laplace noise) not for protecting privacy, but rather to provide oblivious sampling. This is similar in spirit to the usage of differential privacy, by \citet{DworkFHPRR17},  to enable adaptive queries to a database.

\subsection*{Open Questions} \Nnote{see edits}
Our lower bound is only applicable if the number of parties $\np$ is greater than $\log(r)$, and it is very close to the $\np/\sqrt r$-upper bound (protocol) of  \cite{AwerbuchBCGM1985}, if the number of parties is $\np = r^{\eps}$ for ``small'' $\eps>0$.   For $\np <\loglog r$ parties, the upper bounds of \cite{MoranNS16,HaitnerT17,BHLT17} tell us that no attack achieving $\widetilde{O}(1/\sqrt r)$-bias exists. For $ \loglog (r) < \np \le \log(r)$ parties, we know of no such limitation, yet our attack is either innaplicable, or it yields a bias that is smaller than $1/r$. Thus,  for the latter choice of parameters, the only known  bound remains the  $1/r$-lower bound of \cite{Cleve86}, which is  far from meeting the upper bound of \cite{AwerbuchBCGM1985}.

\subsection*{Paper Organization}

Basic definitions and notation used throughout the paper, are given in \cref{sec:Prlim},  we also prove therein some useful inequalities used by the different sections.   Our result  for augmented week martingales is stated and proved in \cref{sec:martingales},  and the oblivious sampling result is given in \cref{sec:Laplace}.  The proof of the main theorem is given in \cref{sec:attack}.

\remove{
Our results are formally stated in \cref{sec:OurResults}. Their full proofs  are given in the appendix. 

}

\newcommand{\sosa}{\textsf{SoS}-augmented }

\section{Preliminaries}\label{sec:Prlim}

\subsection{Notation}\label{sec:prelim:notation}
We use calligraphic letters to denote sets, uppercase for random variables and functions,  lowercase for values, and boldface for vectors. All logarithms considered here are in base two. For a vector $\vct$, we denote its \ith entry by $\vct_i$ or $\vct[i]$. 
For $a\in \R$ and $b\geq 0$,  let $a\pm b$ stand for the interval $[a-b,a+b]$. Given sets $\cs_1,\ldots,\cs_k$ and $k$-input function $f$, let  $f(\cs_1,\ldots,\cs_k) \eqdef \set{f(x_1,\ldots,x_j) \colon x_i\in \cs_i}$, \eg $f(1\pm 0.1) = \set{f(x) \colon x\in [.9,1.1]}$. For $n\in \N$, let $[n] \eqdef \set{1,\ldots,n}$ and $(n) \eqdef \set{0,\ldots,n}$. Given a vector $\vct \in \zo^\ast$, let $w(\vct)\eqdef \sum_{i\in [\size{\vct}]} \vct_i$.  For $x,\delta\in [0,1]$ let $\rnd{\delta} (x)= k\delta$, for $k\in \Z$ being the largest number with $k\delta \le x$.  For a function $f\colon \cA \mapsto \cB$, let $\Img(f)= \set{f(a)\colon a\in \cA}$. 

Let $\poly$ denote the set of all polynomials, let \ppt stand for probabilistic  polynomial time, let  \pptm denote a \ppt algorithm (Turing machine) and   let \polynu stands for a \emph{non-uniform} \pptm.  A function $\nu \colon \N \to [0,1]$ is \textit{negligible}, denoted $\nu(n) = \negl(n)$, if $\nu(n)<1/p(n)$ for every $p\in\poly$ and large enough $n$.

\subsection{Coin-Flipping Protocols}
Since the focus of this paper is showing the non-existence of coin-flipping protocols with small bias, we  will only focus on the correctness and bias of such protocols.   See \cite{HaitnerT17} for a complete definition of such protocols. 

\begin{definition}[correct coin-flipping protocols]\label{def:CorrectCT}
	A multi-party protocol is a {\sf correct coin-flipping protocol}, if
	\begin{itemize}
		\item When interacting with an efficient adversary controlling a subset of the parties, the honest parties {\sf always} output the same bit, and
		\item The common  output in  a random {\sf honest} execution of the protocol is a uniform bit.
	\end{itemize}
\end{definition}

\begin{definition}[Biassing coin-flipping protocols]\label{def:BiassingCT}
	An adversary $\Ac$ controlling a strict subsets of the parties of a correct coin-flipping protocol {\sf biases}  its output by $\delta\in [1/2,1]$, if when interacting with  the parties controlled by $\Ac$, the remaining honest parties output some  a priory fixed bit $b\in\zo$ with probability $\frac12 + \delta$.
	
	Such an adversary is called {\sf fail stop}, if the parties in its control honestly follow the prescribed protocol, but might abort prematurely. The adversary is a rushing adversary, that is, in each round, first the honest parties send their messages, then the adversary might instruct some of the parties to abort (that is, send a special ``abort'' message to all other parties), and finally, all corrupt parties that have not aborted send their messages.
\end{definition}


\subsection{Basic Probability Facts}

Given a distribution $D$, we write $x\gets D$ to indicate that $x$ is selected according to $D$. Similarly, given a random variable $X$, we write $x\gets X$ to indicate that $x$ is selected according to $X$. Given a finite set $\cs$, we let $s\la \cs$ denote that $s$ is selected according to the uniform distribution on $\cs$. Let $D$ be a distribution over a finite set $\Uni$, for $u\in\Uni$, denote $D(u) =\ppr{X\la D}{X=u}$ and for $\cs\subseteq\Uni$ denote $D(\cs) =\ppr{X\la D}{X\in\cs}$. Let the support of $D$, denoted $\Supp(D)$, be defined as $\set{u\in\Uni: D(u)>0}$. The \emph{statistical distance} between two distributions $P$ and $Q$ over a finite set $\Uni$, denoted as $\SD(P,Q)$, is defined as $\max_{\cs\subseteq \Uni} \size{P(\cs)-Q(\cs)} = \frac{1}{2} \sum_{u\in \Uni}\size{P(u)-Q(u)}$.

\begin{fact}[Hoeffding's inequality]\label{fact:Hoeffding}
	Let $\cx=\{x_i\in \set{0,1}\}_{i=1}^n$ and $\mu=\frac{1}{n}\cdot \sum_{i=1}^n x_i$. Let $E\la {[n]\choose n/2}$ i.e.~ $E$ denotes a random subset of $[n]$ of size $n/2$. For any $\eps\geq 0$, it holds that 
	\begin{align*}
		\pr{\,\size{\mu-\frac{2}{n}\cdot \sum_{\ell\in E} x_i}\geq \eps/\sqrt{n}} \leq 2\exp(-\eps^2)\enspace .
	\end{align*}
\end{fact} 
\Nnote{don't forget to cite}

\subsubsection{The Laplace Distribution}
\begin{definition}\label{def:LaplaceDis} The Laplace distribution with parameter $\lambda \in \R^{+}$, denoted $\Lap{\lambda}$,  is defined by the density function $f(x)=\exp(-\size{x}/\lambda)/2\lambda$ .  
\end{definition}

The following facts easily follow from the definition of the Laplace distribution. 
\begin{fact}\label{fact:LaplacePr} For every $x\in \R$, it holds that 
\begin{eqnarray*}
\pr{\mathsf{Lap}(\lambda)\geq \lambda \size{x}} &=& \frac{1}{2}\cdot \exp({-\size{x}}),\\ 
\pr{\Lap{\lambda}\geq -\lambda\size{x}}&= &1-\frac{1}{2}\cdot \exp({-\size{x}}).
\end{eqnarray*}
\end{fact}

\begin{fact}\label{fact:laplace}
	Let  $\gamma, \gamma' \in \R$ and $\lambda \in \R^{+}$. Let $p=\pr{\Lap{\lambda}\geq \lambda\gamma}$ and $p'=\pr{\Lap{\lambda}\geq \lambda\gamma'}$. If $ \size{\eps =  \gamma' - \gamma} \leq 1 $, then $p/ p'\in 1\pm 5\eps$.
\end{fact}
For completeness, the proof of \cref{fact:laplace} is given in \cref{sec:MissinProofs}.

\subsubsection{Useful Observations about Iterated Bernoulli Trials}
 

The next lemma bounds the statistical distance between the first success for two experiments of $r$ independent Bernoulli trials satisfying a certain notion of closeness. 

\begin{lemma}\label{fact:basic2}
	Consider two iterative sequences, each of $r$ independent Bernoulli trials.  Let $p_i,p'_i\in [0,1]$ denote the success probability of the \ith trial of the first and second sequence, respectively. Assume that $p_r=p'_r=1$. For $i\in[r]$, let $q_i=  p_i\cdot \prod_{j<i}(1-p_j)$ and $q'_i=  p'_i\cdot \prod_{j<i}(1-p'_j)$. 
		Let $\eps$ be such that  for all $i\in[r]$, it holds that $\frac{p_i}{p'_i},\frac{p'_i}{p_i}, \frac{(1-p'_i)}{(1-p_i)},\frac{(1-p_i)}{(1-p'_i)}\in (1\pm \eps)$. Then, $\sum_{i=1}^{r-1} \size{q_i -q_i'} \le 4\eps(1 - q_r)$.
\end{lemma} 
The proof of \cref{fact:basic2} is given in \cref{sec:MissinProofs} (see \cref{fact:basic2-appendix}).

\subsubsection{Useful Observations about Conditional Expectation} 
The proofs of the following facts are given in \cref{sec:MissinProofs}.


\begin{fact} \label{lemma:z:2}
	For arbitrary random variables $A$ and $B$, it holds that
	\begin{align*}
	\ex{AB\mid B }= B\cdot \ex{A\mid B}\enspace .
	\end{align*}
\end{fact}

\begin{fact}[Tower Law]\label{lemma:tower}
	For arbitrary random variables $A$, $B$ and $C$, it holds that
	\begin{align*}
	\ex{\ex{A\mid B, C}\mid B }= \ex{A\mid B}\enspace .
	\end{align*}
\end{fact}

\begin{fact}\label{lemma:z:4}
	For arbitrary random variables $A$, $B$, and arbitrary function $f$, it holds that
	\begin{align*}
	\ex{A\mid \ex{A\mid B}, f(B) }= \ex{A\mid B} \enspace .
	\end{align*}
\end{fact}

\begin{fact}\label{lemma:z:3}
	Let $A$, $B$, and $C$ be random variables such that $\supp(B)\subseteq \R$. If \/ $\ex{A\mid B, C}=B$ then
	\begin{align*}
	\ex{A\mid \rnd{\delta}(B), C }\in \rnd{\delta}(B) \pm \delta \enspace .
	\end{align*}
\end{fact}

\subsection{Martingales}\label{sec:prelims:martingales}

In this section we define weaker variants of martingales.

\begin{definition}[$\delta$-martingales]\label{def:DMartingales} Let $X_0, \ldots, X_r$ be a sequence of random variables. 
We say that the sequence is a  {\sf $\delta$-strong} martingale sequence if \/ $\ex{X_{i+1} \mid X_{\le i} = x_{\le i} } \in x_{i} \pm \delta$ for every $i\in[r-1]$. 
We say that the sequence is a  {\sf $\delta$-weak} martingale sequence if\/  $\ex{X_{i+1} \mid X_{i} = x_{i} } \in x_{i} \pm \delta$ for every $i\in[r-1]$.
If $\delta=0$, the above are just  called {\sf strong} and {\sf weak} martingale sequence respectively. 
\end{definition}

\noindent In plain terms, a sequence is a strong martingale if the expectation of the next point conditioned on the entire history is exactly the last observed point. Analogously, a sequence is a weak martingale if the expectation of the next point conditioned on the previous point is equal to the previous point. 

 \begin{definition}[\sosa $\delta$-weak martingale]\label{def:SOSADMartingale} Let $X_0, \ldots, X_r$ be a sequence of random variables. We say that the sequence is a \sosa {\sf $\delta$-weak} martingale sequence if \/ $\ex{X_{i+1} \mid X_{  i} = x_{   i} , \sum_{j<i} (X_{j+1}-X_j)^2=\sigma} \in x_{i} \pm \delta$ for every $i\in[r-1]$.
\end{definition}

\noindent In a sense, a sequence is a \sosa weak martingale if it satisfies the weak martingale property and it is ``distance-oblivious'', i.e.~the expectation is unaffected by conditioning on the quantity $\sum_{j<i} (X_{j+1}-X_j)^2$, which captures the distance the sequence has traveled thus far.

\begin{definition}[Associated difference sequence]\label{def:DMartingaleDif} 
	Let $X_0, \ldots, X_r$ be a an arbitrary sequence and define $Y_{i}=X_{i}-X_{i-1}$, for every $i\in [r]$. The sequence $Y_1\ldots Y_r$ is referred to as the difference sequence associated with $X_0, \ldots, X_r$.
\end{definition}

\noindent By \cref{def:DMartingales,def:DMartingaleDif}, it follows immediately that a sequence $Y_1\ldots Y_r$ is a {\sf $\delta$-strong} martingale difference if and only if $\ex{Y_i \mid Y_1,\ldots, Y_{i-1}}\in \pm \delta$, 
and that a sequence $Y_1\ldots Y_r$ is a {\sf $\delta$-weak} martingale difference if and only if $\ex{Y_i \mid \sum_{\ell< i} Y_\ell}\in \pm \delta$.  By \cref{def:SOSADMartingale,def:DMartingaleDif}, a sequence $Y_1\ldots Y_r$ is a {\sf \sosa $\delta$-weak} martingale difference if and only if $\ex{Y_i \mid \sum_{\ell< i} Y_\ell, \sum_{\ell< i} Y_\ell^2}\in \pm \delta$.

\paragraph{Sequences that behave like martingales most of the time.} We also define sequences that satisfy the different flavors of the martingale property with high probability. Such sequences are referred to as {\sf $(\gamma,\delta)$-martingales}. 

\begin{definition}[$(\gamma,\delta)$-martingales]\label{def:GDMartingales} Let $X_0, \ldots, X_r$ be a sequence of random variables. 
We say that the sequence is a  {\sf$(\gamma,\delta)$-strong} martingale sequence if $$\ppr{x_{\le r}\la X_{\le r}}{\exists i \st\ex{X_{i+1} \mid X_{\le i} = x_{\le i} } \notin x_{i} \pm \delta }\le \gamma \enspace .$$ 
We say that the sequence is a  {\sf $(\gamma,\delta)$-weak} martingale sequence if  $$\ppr{x_{\le r}\la X_{\le r}}{\exists i \st \ex{X_{i+1} \mid X_{i} = x_{i} } \notin x_{i} \pm \delta}\le \gamma \enspace .$$ 
\end{definition}

 \begin{definition}[\sosa $(\gamma,\delta)$-weak martingale]\label{def:SOSAGDMartingale} Let $X_0, \ldots, X_r$ be a sequence of random variables. We say that the sequence is a \sosa {\sf $\delta$-weak} martingale sequence if 
 $$\ppr{x_{\le r}\la X_{\le r}}{\exists i \st\ex{X_{i+1} \mid X_{  i} = x_{   i} , \sum_{j<i} (X_{j+1}-X_j)^2=\sigma} \notin x_{i} \pm \delta}\le \gamma\enspace .$$ 
\end{definition}
 
\section{Augmented Weak Martingales have Large Gaps}\label{sec:martingales}

In this section, we prove a result about sequences that satisfy a weaker version of the ``martingale property''. Namely, we show that for any sequence satisfying the \sosa $\delta$-weak martingale property, if $X_0=1/2$ and $X_r\in \set{0,1}$, then the quantity $\sum_{i=1}^r (X_{i}-X_{i-1})^2$ is greater than $1/16$ with constant probability. As a corollary, we obtain a generalization of the result of \citet{CleveI93}, who showed that (strong) martingales have large gap between consecutive points. We emphasize that our results extend immediately to the usual notion of (strong) martingale sequences. The reader is referred to \cref{sec:intro:martingales} for an informal discussion and motivation for the present section. 

Recall (\cf \cref{sec:prelims:martingales}) that a sequence $X_0,\ldots, X_r$ is a $\delta$-weak martingale, if $\ex{X_{i+1} \mid X_{i} = x_{i} } \in x_{i} \pm \delta$ for every $i\in[r-1]$ and $x_i\in \supp(X_i)$. Further recall that the difference sequence associated with $X_0,\ldots, X_r$ is the sequence $Y_1,\ldots, Y_r$ defined by $Y_i=X_i-X_{i-1}$, for every $i\in [r]$. We begin by extending to weak martingales a result of \citeauthor{dasgupta2011} \cite{dasgupta2011} for strong martingales. We will use this result in the proof of our main theorem.

\begin{lemma}\label{prop:exMachina}
Let $X_0\ldots X_r$ be a $\delta$-weak martingale and let $Y_i = X_i - X_{i-1}$. If $X_i\in [0,1]$ for every $i\in [r]$, then $\ex{X_r^2-X_0^2}\in \ex{\sum_{i\in [r]} Y_i^2} \pm 2r\delta$.
\end{lemma}

\begin{proof}
Write  $\ex{\sum_i Y_i^2}=\ex{\sum_i (X_i-X_{i-1})^2}=\ex{\sum_i \left (X_i^2 -2X_{i}X_{i-1} + X_{i-1}^2\right )}$, and let $\Delta_i=-X_{i-1}+\ex{X_i\mid X_{i-1}}$. By the $\delta$-weak martingale property,  $\Delta_i\in \pm \delta$.  Since   $X_{i-1}\in [0,1]$ and $\Delta_i\in \pm \delta$, it holds that

\begin{align}
 \ex{X_{i-1} \cdot \Delta_i} \in \pm \ex{ \size{\Delta_i} }\in \pm \delta.
\end{align}
Furthermore, 
\begin{align}
\ex{ X_iX_{i-1}}&= \ex{ \ex{X_iX_{i-1} \mid X_{i-1}}}\label{eqn:expoexp}\\
&=  \ex{X_{i-1} \cdot\ex{X_i \mid X_{i-1}}} \label{eqn:exptako}\\
&=  \ex{X_{i-1} \cdot\left (X_{i-1}  -X_{i-1}+\ex{X_i \mid X_{i-1}}\right )} \nonumber \\
& = \ex{X_{i-1}^2 } + \ex{X_{i-1}\Delta_i} \nonumber\\
& \in \ex{X_{i-1}^2 } \pm \delta.\nonumber
\end{align}
\cref{eqn:expoexp,eqn:exptako} follow from \cref{lemma:tower} and \cref{lemma:z:2}, respectively. To conclude, we observe that $\ex{\sum_{i\in [r]} Y_i^2}\in \ex{\sum_{i\in [r]} (X_i^2-X_{i-1}^2) } \pm 2r\delta = \ex{X_r^2-X_{0}^2 }\pm 2r\delta$.
\end{proof}

 
\noindent
Recall  that a sequence $X_0\ldots X_r$ is a \sosa $\delta$-weak martingale if $\ex{X_{i+1} \mid X_{i} = x_{i}, \sum_{\ell\le i} (X_i-X_{i-1})^2=\sigma } \in x_{i} \pm \delta$ for every $i\in[r-1]$, $x_i\in \supp(X_i)$ and $\sigma\in \supp(\sum_{\ell\le i} (X_i-X_{i-1})^2 )$. Following is the main result of this section.


\begin{theorem}\label{claim:SOSAsquares}
For $\delta<1/100r$, let $X_0,\ldots, X_r$ be a \sosa $\delta$-weak martingale sequence such that $X_i\in [0,1]$ for every $i\in [r]$. Assuming $X_0=1/2$ and $\pr{X_r\in \set{0,1}}=1$, then $\pr{\sum_{i\in [r]} (X_i-X_{i-1})^2\ge 1/16 }\ge 1/20$.
\end{theorem}

\begin{remark}
	\cref{claim:SOSAsquares} also holds for the ``standard'' flavor of martingales, \ie strong martingales. Readers who are only interested in the strong case are advised to carry on reading by replacing  below ``\sosa $\delta$-weak'' with ``strong'' and taking $\delta=0$; the proof remains coherent.
\end{remark}

\cref{claim:SOSAsquares} is proven below, but we first sketch its proof. Assume \wlg that $\pr{X_r=1}\ge 1/2$ (otherwise apply the argument to the sequence $X_0',\ldots, X_r'$ defined by $X_i'=1-X_i$, for every $i\in [r]$). Notice that if $\pr{\sum_{i=1}^r (X_{i}-X_{i-1})^2\ge \frac{1}{16}}=0$ then $\ex{\sum_{i=1}^r (X_{i}-X_{i-1})^2}\le \frac{1}{16}$, in contradiction with \cref{prop:exMachina} which states that $\ex{\sum_{i=1}^r (X_{i}-X_{i-1})^2}=\ex{X_r^2-X_0^2}\ge \frac{1}{4}$. We argue that a similar contradiction can be derived if $\pr{\sum_{i=1}^r (X_{i}-X_{i-1})^2\ge \frac{1}{16}}<1/20$. Unfortunately, we cannot apply the same inequality as before because we have no control over the quantity $\sum_{i=1}^r (X_{i}-X_{i-1})^2$ when it is greater than $1/16$ (a crude upper bound is $r$ which is utterly unhelpful). Our solution is to construct a weak martingale sequence $U_0,\ldots, U_r$ which is ``coupled'' with the $X$-sequence in the following way: $U_i$ is equal to $X_i$ as long as $\sum_{\ell=1}^{i-1}(X_{\ell}-X_{\ell-1})^2\le \frac{1}{16}$, and $U_i=U_{i-1}$ otherwise. Then, we argue that $\ex{U_r^2-U_0^2}\ge \frac{1}{4}- \pr{\sum_{i=1}^r (X_{i}-X_{i-1})^2\ge \frac{1}{16}}$ by observing that $\pr{\sum_{i=1}^r (X_{i}-X_{i-1})^2\ge \frac{1}{16}}$ roughly corresponds to the probability that the two sequences diverge. We then upper bound the latter by applying \cref{prop:exMachina} to the sequence $U_0,\ldots, U_r$ which we have a much better grasp on, since, by construction, $\sum_{i=1}^r (U_{i}-U_{i-1})^2$ can never exceed $1/16$ by much.

\begin{proof}[Proof of \cref{claim:SOSAsquares}]
Assume \wlg that $\pr{X_r=1}\ge 1/2$. Further assume that $\pr{\exists i\st \size{Y_i}\ge \frac{1}{4}} < \frac{1}{20}$, as otherwise our theorem is trivially true.    Define the sequence $U_0\ldots U_r$ by $U_i=X_i$ if $\sum_{j<i} Y^2_i \le  \frac{1}{16} $, and $U_i=U_{i-1}$ otherwise. We show that $U_0,\ldots, U_r$ is a $\delta$-weak martingale, \ie  $\ex{U_i\mid U_{i-1}=u}\in u\pm \delta$.  Write $Z_i = U_i - U_{i-1}$ and fix  $u\in \supp(U_{i-1})$ and $\sigma\in \supp(\sum_{j<i} Z_j^2)$ . Observe that $\ex{Z_i\mid U_{i-1}=u, \sum_{j<i} Z_i^2=\sigma}=0$ if $\sigma> \frac{1}{16}$, and $\ex{Z_i\mid U_{i-1}=u, \sum_{j<i} Z_i^2=\sigma}=\ex{Y_i\mid X_{i-1}=u, \sum_{j<i} Y_i^2=\sigma}$ otherwise. Thus,  by the \sosa property of $Y$, it holds that 

\begin{align}
\ex{Z_i\mid U_{i-1}=u, \sum_{j<i} Z_i^2=\sigma}\in \pm \delta.
\end{align}

\noindent
Since $u$ and $\zeta$ were chosen arbitrarily, we deduce that $\ex{Z_i\mid U_{i-1}=u}\in \pm \delta$, and thus $U_0,\ldots, U_r$ is a $\delta$-weak martingale sequence. Furthermore, since $U_i\in [0,1]$, by \cref{prop:exMachina}, 
\begin{align}\label{eq:upper} 
\ex{U_r^2} - \frac{1}{4}&\le  \ex{\sum_{ i\in [r]} Z_i^2} +2r\delta\\
&\le   \pr{\exists i\in [r]\st \size{Y_i}\ge \frac{1}{4}}\cdot\left  ( \frac{1}{16}+1\right ) + \pr{\forall i\in [r] \;  \size{Y_i}< \frac{1}{4}}\cdot     \left ( \frac{1}{16}+\frac{1}{16}\right ) +2r\delta \label{eqn:chopoff}\\
&\le   \frac{1}{20}\cdot\left  ( \frac{1}{16}+1\right ) + 1\cdot   \left  ( \frac{1}{16}+\frac{1}{16}\right ) +2r\delta \nonumber\\
&\le  0.18 +2r\delta.\nonumber
\end{align}

\noindent
\cref{eqn:chopoff} follows from the fact that, by construction, the quantity $\sum_{i=1}^{r} Z_i^2$ is equal to $\sum_{i=1}^{\Gamma+ 1} Y_i^2$, where $\Gamma$ is the largest index such that $\sum_{i=1}^{\Gamma} Y_i^2\le \frac{1}{16}$. On the other hand, by noting that $\pr{X_r\neq U_r}\le \pr{\sum_{i\in [r]} Y_i^2> \frac{1}{16}}$,
\begin{align}\label{eq:lower} 
\ex{U_r^2}&\ge 1^2\cdot \pr{U_r=1}\\
& \ge  1^2\cdot \pr{X_r=1 \land X_r=U_r}\nonumber\\
& \ge  1^2\cdot \left (\pr{X_r=1}-\pr{ X_r\neq U_r}\right )\nonumber\\
&\ge 1/2 - \pr{X_r\neq U_r}\nonumber \\
&\ge 1/2 -  \pr{\sum_{i\in [r]} Y_i^2> \frac{1}{16}} .\nonumber
\end{align}

\noindent
Combine \cref{eq:upper,eq:lower} we deduce that 
\begin{align*}
\pr{\sum_i Y_i^2> \frac{1}{16}} &\ge \frac{1}{4}- 0.18 -2r\delta \ge \frac{1}{20},
\end{align*}
where the last inequality is true since $\delta\le 1/100r$.

\end{proof}
 
\cref{claim:SOSAsquares} immediately yields the following corollary. 
\begin{corollary}\label{claim:SOSAjumps}
For $\delta<1/100r$, let $X_0,\ldots, X_r$ be a \sosa $\delta$-weak martingale sequence such that $X_i\in [0,1]$ for every $i\in [r]$. If $X_0=1/2$ and $\pr{X_r\in \set{0,1}}=1$, then $\pr{\exists i\in [r]\st  \size{X_i-X_{i-1}}\ge \frac{1}{4\sqrt{r}} }\ge \frac{1}{20}$.
\end{corollary}

Ans also the following corollary follows via a simple coupling argument.

\begin{corollary}\label{claim:SOSAjumpsWHP}
For $\gamma<1/1000$ and $\delta<1/100r$, let $X_0,\ldots, X_r$ be a \sosa $(\gamma,\delta)$-weak martingale sequence such that $X_i\in [0,1]$ for every $i\in [r]$. If $X_0=1/2$ and $\pr{X_r\in \set{0,1}}=1$, then $\pr{\exists i\in [r]\st  \size{X_i-X_{i-1}}\ge \frac{1}{4\sqrt{r}} }\ge \frac{1}{20}-\gamma$.
\end{corollary}

\begin{remark}
	We mention that for any constant $\gamma<1/2$, it can be shown that the sequence has gaps of order $1/\sqrt{r}$, with constant probability. For the specific choice of parameters $1/4\sqrt{r}$ and $1/20$, the value of $\gamma$ should be smaller than $1/1000$. 
\end{remark}

\begin{remark}[Extensions] We remark that we can replace the requirement $X_i\in [0,1]$ for the less restrictive $\size{X_i-X_{i-1}}\le 1$ in \cref{claim:SOSAsquares} and its corollaries. However, the resulting claims offer no gains for the purposes of the present paper and their proofs are significantly longer, as far as we can tell. Therefore, we only prove here the more restrictive versions as stated in the present section.
\end{remark}

\section{ Oblivious Sampling via  Differential Privacy}\label{sec:Laplace}
Consider the following $r$-round game in which your goal is to maximize the revenue of a random ``party'' $H\la \H$. In the beginning, a party $H$ is chosen with uniform distribution from $\H$ (where $\H$ is a finite set of parties). In each round,   values $\set{s_i^h\in [0,1]}_{h\in \H}$ are assigned to the  parties   of  $\H$, but   only the values   $\set{s_{h}}_{h\in \H\setminus \set{H}}$ of  the other parties  are published. You can decide to \emph{abort}, and then party $H$ is rewarded by  $s_i^H$, or to continue to the next round. If an abort never occurs, party $H$ is rewarded by   $s_r^H$ (last round value). You have the  \emph{similarity} guarantee that  $\size{s_i^h - s_i}\le \sigma$ for every $h \in \H$,  letting $s_i = \eex{h\la \H}{s_i^h}$. You are also guarantee that  $\max_i \set{s_i} \ge \gamma$.  

In this section we analyze  the following ``differentially private based''  approach for this task, which is described by the following experiment (the basic game described above is captured by the experiment for $p=1/\np$).
\begin{experiment} [$\LapExp$: Oblivious sampling experiment]\label{Exp:Laplace}
\item Parameters:  $\H=\set{1,\ldots, n}$,  $\cs = \set{s_i^{h} \in [-1,1]}_{i\in [r],h\in \H}$, $p\in[0,1/2]$, $\gamma \in[0,1]$ and $\lambda\in \R^+$.

\item Notation: Let $s_i=\frac{1}{n}\sum_{h\in \H} s_i^{h}$ and for $h\in \H$ let $s_i^{\mh}=\frac{1}{1-p}(s_i-p\cdot s_i^h)$. 

\item Description:

\begin{enumerate}
\item Sample $h\la \H$.
\item For $i=1,\ldots, r-1$: 
\begin{enumerate}
\item  Sample $\nu_i\la  \Lap{\lambda}$. 
\item If $s_i^{\mh}+\nu_i\geq \gamma$, output $s_i^{h}$ and halt. 
\label{step:halt}
\end{enumerate}
\item Output $s_r^{h}$.
\end{enumerate} 
\end{experiment}

Let  $\LapExp(\H,\cs,\gamma,\lambda)$ denote the above experiment with parameters $\H$, $\cs$, $\gamma$ and $\lambda$. \cref{thm:Laplace} analyzes the expected value of the output of $\LapExp(\H,\cs,\gamma,\lambda)$.

\begin{theorem}[Quality of the  oblivious sampling experiment]\label{thm:Laplace} 
Let $\H$, $\cs$, $\gamma$, $\lambda$ and $p$ be as in \cref{Exp:Laplace},  with  $s^h_r = s_r$ for every $h\in \H$. Let $\sigma^{h}=\max_i\set{\size{ s_i-s_i^{h}}}$, let $\Unbiased=\set{h\in \H \colon \sigma^h \le  \lambda\cdot (1-p)/p}$ and $\Biased = \H \setminus \Unbiased$.  

Let $H$ be the value of $h$ and  $J$ be the halting round  (set to $r$ if  \cref{Exp:Laplace} does not halt in step~(\ref{step:halt})) in a random execution of  $\LapExp(\H,\cs,\gamma,\lambda)$. Then  $\ex{s_J^H} \ge \ex{v_H}- r \cdot e^{-\gamma/2\lambda}$, where 
$$
v_h=
\begin{cases}
\pr{J\neq r\mid H=h} \cdot \left(\frac{\gamma}2  - \frac{40  (\sigma^h)^2 }{\lambda   }\cdot \frac{p}{1-p} \right), &  h\in  \Unbiased,\\
-4\sigma^h, & h\in  \Biased.
\end{cases}
$$
If $s_i \ge \gamma$ for some $ i\in [r-1]$, then  $\pr{J\neq r\mid H=h}\ge  1/6$, for every $h\in \Unbiased$.
\end{theorem}


When using \cref{thm:Laplace} in our proofs, the values $\cs = \set{s_i^{h} \in [-1,1]}_{i\in [r],h\in \H}$ are calculated for a fixed transcript of the coin-flipping protocol.
\cref{thm:LaplaceGen} analysis the  expected value of the output when first a transcript $\tau$ is chosen, then the values $\cs_\tau$ are computed, and finally $\LapExp(\H,\cs_\tau,\gamma,\lambda)$ is executed.  

\begin{corollary}\label{thm:LaplaceGen} 
	Let $\H$, $\gamma$, $\lambda$ and $p$ be as in \cref{Exp:Laplace}. Let $\cs=\set{\cs_\tau=\set{s_i^h(\tau)}_{i\in [r], h\in \H}}_{\tau\in \cT}$ denote a set of numbers in $[-1,1]$ indexed by $i\in [r]$, $h\in \H$ and $\tau$ taking values in some set $\cT$. Define $\sigma^{h}(\tau)=\max_i\set{\size{ s_i(\tau)-s_i^{h}(\tau)}}$. Let $T$ be a random variable taking values in $\cT$, and let $H$ be the value of $h$ and $J$ be the halting round (set to $r$ if  \cref{Exp:Laplace} does not halt in step~(\ref{step:halt})) in a random execution of  $\LapExp(\H,\S_{\tau\la T},\gamma,\lambda)$. Further assume that there exist real numbers $\alpha$, $\beta$, $\gamma$, $\delta\in [0,1]$ such that  
	\begin{itemize}
	\item $\alpha\leq \lambda(1-p)/2p$,
	\item $\ppr{\tau\la T}{ \sigma^h(\tau)\geq \rho\cdot \alpha } \leq \frac{1}{\rho}\cdot \beta$, for every $h\in \H$ and $\rho\geq 1$,
	\item $\ppr{\tau\la T}{\max_{i\in [r]} s_i(\tau)\geq \gamma   }\geq \delta $.
	\end{itemize} 
	Then, 
	\begin{align*}
		\ex{s_J^H(T)}\geq \frac{1}{6}\cdot (\delta - \beta/2)\cdot \left (\frac{\gamma}{2} - \frac{40\cdot \alpha^2p}{\lambda (1-p)}\right )-168\alpha\beta - 8 \alpha\beta\log(1/\lambda)-\frac{r}{2}\cdot e^{-\gamma/2\lambda} .
	\end{align*}
	In particular, if $\gamma\geq \frac{1}{256\sqrt{r}}$, $\lambda=\gamma/(4\log(r))$, $\alpha\leq \frac{\gamma  \sqrt{ 4(1-p)/p}}{32\log(r)}$ and $\beta\leq \frac{\delta}{16\sqrt{(1-p)/p}}$, then $\ex{s_J^H(T)}\geq \gamma\delta/125 - \frac{1}{2r}$, {\sf for $r$ large enough}.
\end{corollary}


\subsection{Proving \cref{thm:Laplace}}

\begin{proof}[Proof of  \cref{thm:Laplace}]
For  $h\in \H$ and $i\in [r]$, let $d_i^h = s_i^\mh - s_i^h$. We next compute $\ex{s_J^H}$.
\begin{align}\label{eq:Laplace:1}
\ex{s_J^H}&= \sum_{i\in [r], h\in \H} s_i^h\cdot \pr{H=h\,\wedge J=i}\\
&= \sum_{i,h}  (s_i^\mh-d^h_i)\cdot \pr{H=h\,\wedge J=i}\nonumber\\
&= \sum_{i,h}  s_i^\mh\cdot \pr{H=h\,\wedge J=i}-\sum_{i,h}   d^h_i \cdot \pr{H=h\,\wedge J=i}\nonumber\\
&=  \ex{s_J^\mH}- \sum_{i,h}   d^h_i \cdot \pr{H=h\,\wedge J=i}\nonumber\\
&=  \ex{s_J^\mH}- \frac1n \cdot \sum_{i \in [r],h\in h}   d^h_i \cdot \pr{J=i \mid H=h}\nonumber\\
&=  \ex{s_J^\mH}- \frac1n \cdot \sum_{i \in [r-1],h\in h}   d^h_i \cdot \pr{J=i \mid H=h}.\nonumber
\end{align} 
The last equality holds since, by assumption, $s_r = s_r^h$ for any $h$, thus, $d^h_r=0$.

\vspace*{\baselineskip}

\noindent 
We start by upper bounding the right hand term above (\ie $\sum_{i,h}   d^h_i \cdot \pr{J=i \mid H=h}$). For $h\in \H$ and $i\in [r-1]$, let $$p_i^h=\pr{\Lap{\lambda}+ s_i^{\mh}\geq \gamma}, \  p_r^h =1, \text{ and } q_i^h = p_i^h\cdot \prod_{j<i}(1-p^h_j).$$ 
Note that $q_i^h = \pr{J=i \mid H=h}$.  For $i\in [r]$, let $$p_i=\pr{\Lap{\lambda}+ s_i \geq \gamma} \text{ and }  q_i=  p_i\cdot \prod_{j<i}(1-p_j).$$
 Let $\sigma_i^h = s_i - s_i^h$ and $\sigma_i^{\mh} = s_ i - s_i^{\mh}$.
 Note that $d_i^h =-\sigma_i^\mh + \sigma_i^h$. Since $s_i = (1-p) \cdot s_i^{\mh} + p \cdot s_i^{h}$, it holds that   $\sigma_i^{\mh}= - p\cdot \sigma_i^h/ (1-p)$. 
In particular, for any $h \in \Unbiased$ it holds that  $\size{\sigma_i^h} \leq  \sigma_i^h \leq \lambda (1-p)/p$ and $\size{\sigma_i^{\mh}} \le \lambda$.  Hence,  \cref{fact:laplace} yields  that $p_i^h/p_i \in 1 \pm 5 \sigma^{\mh}/\lambda$ for any $h\in \Unbiased$. Therefore,  by \cref{fact:basic2} 
\begin{align}\label{eq:Laplace:2}
\sum_{i\in [r-1]} \size{q_i - q_i^h} \le   \frac{20}{\lambda} \cdot \sigma^\mh \cdot (1- q_r^h) \le \frac{20p}{\lambda (1-p) } \cdot \sigma^h \cdot (1- q_r^h).
\end{align}
for any $h\in \Unbiased$. 
Define $d^h=\max_i\set{\size{d^h_i}}$.
It follows that
\begin{align}\label{eq:Laplace:3}
\sum_{i\in [r-1],h\in \H} d^h_i \cdot \pr{J=i \mid H=h} 
&=  \sum_{i,h}  d_i^h \cdot q^h_i  \\
&=  \sum_{i,h}  d_i^h \cdot q_i  + \sum_{i,h}  d_i^h \cdot ( q_i^h  - q_i )\nonumber\\
&=  \sum_{i,h}  d_i^h \cdot( q_i^h  - q_i )\nonumber\\
&\le  \sum_{h\in \Unbiased}  d^h  \sum_{i\in [r-1] } \size{q_i^h  - q_i}  +\sum_{h\in \Biased} 2d^h \nonumber\\
&\le \frac{20p}{\lambda\cdot (1-p)} \cdot \sum_{h\in \Unbiased} d^h \cdot\sigma^h \cdot  (1- q_r^h)  +  \sum_{h\in \Biased} 2d^h\nonumber\\
&\le \frac{20p}{\lambda\cdot (1-p)} \cdot \sum_{h\in \Unbiased} 2(\sigma^h)^2 \cdot  (1- q_r^h)  +  \sum_{h\in \Biased} 4\sigma^h\nonumber.
\end{align}
The second equality holds since $\sum_{h\in \H} \sigma_i^h  = 0$ for any $i\in [r]$, and thus $\sum_{h\in \H} \sigma_i^\mh  = 0$ and $\sum_{h\in \H} d_i^h  = 0$. The last inequality holds since $p\leq 1/2$ and thus $d^h =|-p\sigma^h/(1-p)-\sigma^h|\le 2 \sigma^h$.

\vspace*{\baselineskip}

\noindent 
The next step is to lower bound $\ex{s_J^\mH}$. By \cref{fact:LaplacePr},
\begin{align}\label{eq:Laplace:4}
\pr{J \neq r \land s_J^\mH \le \gamma/2 } = \sum_{i=1}^{r-1} \pr{J =i  \land s_i^\mH \le \gamma/2 }
 \le \frac{ r}{2}\cdot e^{-\gamma/2\lambda}.
\end{align}
Hence, since $\pr{J \neq r}= \pr{J \neq r \land s_J^\mH > \gamma/2} +  \pr{J \neq r \land s_J^\mH\le \gamma/2}$,
\begin{align}\label{eq:Laplace:5}
\ex{s_J^\mH} &\ge   \pr{J \neq r  \land s_J^\mH>\gamma/2}  \cdot  \gamma/2 - 1\cdot  \pr{J \neq r  \land s_J^\mH\le \gamma/2} \\
&\ge \left ( \pr{J \neq r } -  \frac{ r}{2} \cdot e^{-\gamma/2\lambda} \right ) \cdot  \gamma/2 - \frac{ r}{2} \cdot e^{-\gamma/2\lambda}\nonumber\\
& \ge \pr{J \neq r }\cdot  \gamma/2 - r \cdot e^{-\gamma/2\lambda} \nonumber\\
& = \ex{1- q_r^H} \cdot  \gamma/2 - r \cdot e^{-\gamma/2\lambda} \nonumber \\
&\ge \left  (\frac1n \sum_{h\in \Unbiased} (1- q_r^h) \cdot \gamma/2 \right) - r\cdot e^{-\gamma/2\lambda}\nonumber.
\end{align}
Thus, by \cref{eq:Laplace:1,eq:Laplace:3,eq:Laplace:5}, 
\begin{align}\label{eq:Laplace:6}
\lefteqn{\ex{s_J^H}}\\
& \ge \frac1n \left( \sum_{h\in \Unbiased} (1- q_r^h) \cdot \left (\gamma/2 - \frac{40}{\lambda\cdot (1-p)/p} (\sigma^h)^2 \right )  \right)- \frac{1}{n}\sum_{h\in \Biased} 4\sigma^h - r\cdot e^{-\gamma/2\lambda} \nonumber .
\end{align}
To conclude the proof we need to show  that if $s_i \ge \gamma$ for some $ i\in [r-1]$, then  $(1- q_r^h) \ge 1/6$ for all $h\in \Unbiased$.  Let $i \in [r-1]$ be a round with $s_i \ge \gamma$. For every $h\in \Unbiased$, we have shown that
$\size{\sigma_i^{\mh}} \le \lambda$, thus, $s_i^\mh\geq \gamma - \lambda$. By \cref{fact:LaplacePr}, it holds that $p^h_i \ge \pr{\Lap{\lambda}\geq \lambda }=  \exp(-1)/2\geq 1/6$. Hence, $1-q_r^h \ge 1/6$, for all $h\in \Unbiased$.  
\end{proof}

\subsection{Proving \cref{thm:LaplaceGen}}

\noindent 
Before proving the theorem, we state and prove two claims regarding the expectation of the similarity gap.

\begin{claim}\label{clm:expsigma} Under the hypothesis of \cref{thm:LaplaceGen}, it holds that 
\begin{align}
\ex{\sigma^h(T)  \mid \sigma^h(T)\geq (1-p)\lambda/p }\cdot \pr{ \sigma^h(T)\geq (1-p)\lambda/p } &\leq 2\alpha\beta\log(1/\lambda)+2\alpha\beta ,\label{eqn:outlier} \\
\ex{(\sigma^{h }(T))^2  \mid \sigma^h(T)\in [\alpha, (1-p)\lambda/p] }\cdot \pr{ \sigma^h(T)\in [\alpha, (1-p)\lambda/p] } &\leq 4(1-p)\lambda\alpha\beta/p  .
\label{eqn:exsquared}
\end{align}
 for every $h\in \H$.
\end{claim}

\begin{proof} We begin by showing \eqref{eqn:outlier}.
\begin{align*}
\lefteqn{\ex{\sigma^h(T)  \mid \sigma^h(T)\geq (1-p)\lambda/p }\cdot \pr{ \sigma^h(T)\geq (1-p)\lambda/p }}\\
&\leq \sum_{i=\log((1-p)\lambda/p\alpha)}^{\log(1/\alpha)} \alpha2^{i+1}\cdot  \pr{\sigma^h(T)\in \alpha\cdot [2^i, 2^{i+1}]}\\
&\leq \sum_{i=\log((1-p)\lambda/p\alpha)}^{\log(1/\alpha)}  \alpha2^{i+1}\cdot  \pr{\sigma^h(T)\geq  2^i\cdot \alpha}\\
&\leq \sum_{i=\log((1-p)\lambda/p\alpha)}^{\log(1/\alpha)}  \alpha2^{i+1}\cdot 2^{-i}\beta\\
&= 2 \alpha \beta \left ( \log(1/\alpha)  - \log((1-p)\lambda/p\alpha) +1  \right )\\
&= 2 \alpha \beta \log(p/(1-p)\lambda)+2\alpha\beta
\leq 2 \alpha \beta \log(1/\lambda)  +2\alpha\beta .
\end{align*} Next, we show \eqref{eqn:exsquared}.
\begin{align*}
\lefteqn{\ex{(\sigma^{h }(T))^2  \mid \sigma^h(T)\in [\alpha,  (1-p)\lambda/p ] }\cdot \pr{ \sigma^h(T)\in [\alpha,  (1-p)\lambda/p ] }}\\
&\leq \sum_{i=0}^{\log( (1-p)\lambda/p\alpha )-1} \alpha^2 2^{2i+2}\cdot  \pr{\sigma^h(T)\in \alpha\cdot [2^i, 2^{i+1}]}\\
&\leq \sum_{i=0}^{\log((1-p)\lambda/p\alpha)-1} \alpha^2 2^{2i+2}\cdot  2^{-i}\cdot \beta\\
&= 4\alpha^2\beta \sum_{i=0}^{\log((1-p)\lambda/p\alpha)-1} 2^i \\
&=  4\alpha^2\beta ((1-p)\lambda/p\alpha -1)\leq    4\alpha\beta   (1-p)\lambda/p  .
\end{align*}
\end{proof}

\begin{proof}[Proof of \cref{thm:LaplaceGen}]
Using the notation from \cref{thm:Laplace}, since $r\cdot e^{-\gamma/2\lambda}\le -\frac{1}{2r}$, for $r$ large enough, it suffices to bound $\ex{v_h\mid H=h}$ for an arbitrary $h$. To prove the claim, we will combine \cref{thm:Laplace} with \eqref{eqn:outlier} and \eqref{eqn:exsquared} from \cref{clm:expsigma}. Having fixed $H=h$, all probabilities and expectations below are conditioned on $H=h$. To alleviate notation, we will omit specifying that $H=h$, i.e.~instead of $\ex{\cdots \mid \cdots \land H=h}$ and $\pr{\cdots \mid \cdots \land H=h}$, we write $\ex{\cdots \mid \cdots }$ and $\pr{\cdots \mid \cdots }$.  Using the notation from  \cref{thm:Laplace} , we compute 

\begin{align*}
\ex{v_h}&= \ex{v_h\mid \sigma^h(T)\leq \alpha} \cdot \pr{\sigma^h(T)\leq \alpha} \\
 & \qquad+  \ex{v_h\mid \sigma^h(T)\in [\alpha, (1-p)\lambda/p]} \cdot \pr{\sigma^h(T)\in [\alpha, (1-p)\lambda/p]}\\ 
 &\qquad\qquad + \ex{v_h\mid \sigma^h(T)  \geq (1-p)\lambda/p} \cdot \pr{\sigma^h(T)  \geq (1-p)\lambda/p}  .
\end{align*} 
We compute each of these terms separately. 
First, by expanding over all possible transcripts,
\begin{align*}
&\ex{v_h\mid \sigma^h(T)\leq \alpha} \cdot \pr{\sigma^h(T)\leq \alpha} \\
& \qquad =  \pr{\sigma^h(T)\leq \alpha}\cdot \sum_{\tau \st \sigma^h(\tau)\le \alpha} \ex{v_h\mid T=\tau} \cdot \pr{T=\tau\mid \sigma^h(T)\leq \alpha}. \\
\\ &\qquad \geq \pr{\sigma^h(T)\leq \alpha}\cdot \sum_{\tau \st \sigma^h(\tau)\le \alpha} \left(\frac{\gamma}{2} - \frac{40\cdot \alpha^2}{\lambda(1-p)/p} \right )\cdot\pr{J\neq r\mid T=\tau} \cdot \pr{T=\tau\mid \sigma^h(T)\leq \alpha}  .
\end{align*}
 
\noindent 
Where the last inequality follows by the definition of $v_h$ when $h$ is similar (\cf ). Consequently, 
\begin{align}\label{eq:perfect}
 \ex{v_h\mid \sigma^h(T)\leq \alpha} \cdot \pr{\sigma^h(T)\leq \alpha}  \ge \left(\frac{\gamma}{2} - \frac{40\cdot \alpha^2}{\lambda(1-p)/p} \right )\cdot\pr{J\neq r\land  \sigma^h(T)\leq \alpha}
\end{align} 

\newcommand{\simil}{p_\mathsf{sim}}

\noindent
Next, write $\simil=\pr{\sigma^h(T)\in [\alpha, (1-p)\lambda/p]}  $. For the second term, again by expanding over all possible transcripts, we compute
\begin{align*}
\lefteqn{ \ex{v_h \mid  \sigma^h(T)\in [\alpha, (1-p)\lambda/p]} \cdot \pr{\sigma^h(T)\in [\alpha, (1-p)\lambda/p]}  }\\ 
& = \simil \cdot \sum_{\substack{\tau \st \\ \sigma^h(\tau)\in [\alpha, (1-p)\lambda/p]}} \ex{v_h \mid  T=\tau} \cdot \pr{T=\tau \mid \sigma^h(T)\in [\alpha, (1-p)\lambda/p]}   
\end{align*}
For any fixed $\tau$ such that $\sigma^h(\tau)\in [\alpha, (1-p)\lambda/p]$, by the definition of $v_h$ when $h$ is similar (\cf \cref{thm:Laplace}), it holds that 
\begin{align*}
 \ex{v_h \mid  T=\tau} & \ge \pr{J\neq r\mid T=\tau}\left (\frac{\gamma}{2} - \frac{40}{\lambda(1-p)/p} \cdot \ex{(\sigma^h)^2\mid T=\tau}\right ) \\
 &\ge \pr{J\neq r\mid T=\tau}\cdot  \frac{\gamma}{2} - \frac{40}{\lambda(1-p)/p} \cdot \ex{(\sigma^h)^2\mid T=\tau}
\end{align*}
and we deduce that 
\begin{align}
\lefteqn{ \ex{v_h \mid  \sigma^h(T)\in [\alpha, (1-p)\lambda/p]} \cdot \pr{\sigma^h(T)\in [\alpha, (1-p)\lambda/p]}  } \nonumber\\  
& \qquad\geq    \frac{\gamma}{2}\cdot\pr{J\neq r\,\wedge\, \sigma^h(T)\in [\alpha,(1-p)\lambda/p]} \nonumber\\
& \qquad\quad- \frac{40}{(1-p)\lambda/p} \cdot \ex{(\sigma^h(T))^2\mid\sigma^h(T) \in [\alpha, (1-p)\lambda/p] }\cdot \pr{ \sigma^h(T)\in [\alpha, (1-p)\lambda/p] }.\label{eq:similaro}
\end{align}
Finally, by the definition of $v_h$ when $h$ is non-similar (\cf \cref{thm:Laplace}), 
\begin{align}
\lefteqn{ \ex{v_h\mid \sigma^h(T)\geq (1-p)\lambda/p} \cdot \pr{\sigma^h(T)\geq (1-p)\lambda/p}}  \nonumber\\ 
& \qquad\geq  - 4\cdot \ex{\sigma^h(T) \mid \sigma^h(T)\geq (1-p)\lambda/p} \cdot \pr{\sigma^h(T)\geq (1-p)\lambda/p}.\label{eq:nosimilaro}
\end{align}
Add \cref{eq:perfect,eq:similaro,eq:nosimilaro} and replace the relevant expressions using \eqref{eqn:outlier} and \eqref{eqn:exsquared}:
\begin{align*}
\ex{v_h}\geq \pr{h\in \Unbiased\wedge J\neq r}\cdot \left ( \frac{\gamma}{2} - \frac{40\cdot \alpha^2}{(1-p)\lambda/p}\right ) -40\cdot 4\alpha\beta - 8\alpha\beta \log(1/\lambda)-8\alpha\beta  .
\end{align*}
Next, we lower-bound the quantity $\pr{h\in \Unbiased\wedge J\neq r}$.
\begin{align*}  \pr{h\in \Unbiased\wedge J\neq r}&\geq \pr{h\in \Unbiased\wedge J\neq r \wedge \exists s_i(T)\geq \gamma} \\
&= \pr{  J\neq r \mid  h\in \Unbiased\wedge \exists s_i(T)\geq \gamma}\cdot  \pr{h\in \Unbiased\wedge \exists s_i(T)\geq \gamma}\\
&\geq \frac{1}{6}\cdot (\pr{\exists s_i(T)\geq \gamma}- \pr{h\in \Biased})  \\
&\geq \frac{1}{6}\cdot (\delta- \beta/2)  .
\end{align*}
The last inequality follows from the fact that $\pr{\exists s_i(T)\geq \gamma}\ge \delta$ and $\pr{h\in \Biased}\leq \pr{\sigma^{h}\geq 2\alpha}\leq \beta/2$.
In summary,
\begin{align*}
\ex{v_h}\geq \frac{\delta - \beta/2}{6}\cdot \left ( \frac{\gamma}{2} - \frac{40\cdot \alpha^2}{(1-p)\lambda/p}\right )  -(168\alpha\beta + 8\alpha\beta \log(1/\lambda))  .
\end{align*}
\vspace*{\baselineskip}
\noindent 
The last part of the claim follows from the inequalities below, holding for large enough $r$:
	\begin{itemize}
		 
		\item $\delta-\dfrac{\beta}{2}\geq \dfrac{31}{32}\cdot\delta$, \quad since $\beta\le \sqrt{\frac{p}{1-p}}\cdot\frac{\delta}{16}$ and $p\le 1/2$.
		
		\item $\dfrac{\gamma}{2}-\dfrac{40\alpha^2 p}{\lambda(1-p)}\geq \dfrac{\gamma}{4}$, \quad since $\alpha^2 \le \frac{\gamma\lambda}{64\log(r)}\cdot \frac{1-p}{p} $.
		
		\item $168\alpha\beta + 8\alpha\beta\log(1/\lambda) \leq 8\alpha\beta\log(r)\leq \dfrac{\gamma\delta}{32}$,  \quad since the leading term in the far left summand is $4\alpha\beta\log(r)$ (because of the square root) and $\alpha\beta\le \frac{\gamma\delta}{256\log(r)}$.
	\end{itemize} 
	We conclude that $\ex{v_h}\geq \frac{ \gamma\delta}{25} - \frac{\gamma\delta}{32} \ge  \frac{\gamma\delta}{125}$.
\end{proof} 

\newcommand{\Nugg}{\mathsf{NuggetFinder}}
\newcommand{\itchy}{\mathsf{BuildXLoop}}
\newcommand{\allchi}{\mathsf{BuildX}}
\newcommand{\simm}{\mathsf{similar}}
\newcommand{\issim}{\mathsf{IsSimilar}}

\newcommand{\Xseq}{\mathsf{RunX}}

\section{Biasing Coin-Flipping Protocols}\label{sec:attack}

In this section we prove our main result,  an almost optimal  attack on many-party  coin-flipping protocols.

\def\MainTheorem{
There exists a fail-stop adversary $\Ac$ such that the following holds. Let $\Pi$ be a correct $\np$-party $r$-round coin-flipping protocol, and let $k\in \N$ be the smallest integer such that ${n\choose k}\geq r\log(r)^{2k}$. Then, there exists a party $\Pc$ in $\Pi$ such that $\Ac^\Pi$ controlling all parties but $\Pc$ biases the output of  $\Pc$ by $\Omega(1/\sqrt{r}\log(r)^k)$. The running time of $\Ac^\Pi$  is polynomial in the running time of $\Pi$ and $\np^k$, and it uses oracle only access to $\Pi$'s next-message function.
}

\begin{theorem}[Main theorem]\label{thm:Unimain}
	\MainTheorem
\end{theorem}

\begin{remark}[Interesting choice of parameters] 
Note that $\sqrt{\np}\geq 2\log(r)^2$ implies  ${\np\choose \sqrt{\np}}\geq \sqrt{\np}^{\sqrt{\np}}\geq 2^{\sqrt{\np}}\log(r)^{2\sqrt{\np}}\geq r \log(r)^{2\sqrt{\np}}$,  and therefore there exists $k\in \set{1,\ldots, \sqrt{n}}$ satisfying the hypothesis of the theorem. On the other hand, if $\sqrt{\np}< 2\log(r)^2 $, then it is easy to see that either such  $k$ does not exist, or $\log(r)^k\geq \sqrt{r}$ and in this case \citet{Cleve86}'s bound overtakes and our theorem is trivial.
\end{remark}

Let $\np'=\floor{\np/\noh}$, for some $\noh< n/2$. By noting that any $\np$-party $r$-round coin-flipping protocol is, in particular, an $\np'$-party $r$-round coin-flipping protocol, the theorem below follows \cref{thm:Unimain} by simple reduction.  

\begin{theorem}[Main theorem, fewer  corruptions variant,]\label{thm:FewerUnimain}
	Let $\Pi$ be a correct $\np$-party $r$-round coin-flipping protocol and let $n'=\floor{n/\noh}$, for some $\noh<n/2$. There exists a fail-stop adversary $\Ac$ such that the following holds. Let $k\in \N$ be the smallest integer such that ${n'\choose k}\geq r\log(r)^{2k}$. Then, there exists parties $\Pc_1,\ldots, \Pc_s$ in $\Pi$ such that $\Ac^\Pi$ controlling all parties but $\Pc_1,\ldots, \Pc_\noh$ biases the output of  $\Pc_1,\ldots, \Pc_\noh$ by $\Omega(1/\sqrt{r}\log(r)^k)$. The running time of $\Ac^\Pi$  is polynomial in the running time of $\Pi$ and $\np^k$, and it uses oracle only access to $\Pi$'s next-message function.\footnote{We require $\noh<n/2$, otherwise the resulting protocol has an honest majority, and standard MPC techniques would foil the attack.}
\end{theorem}
For presentation purposes, we will prove our theorem for the special case of non-uniform \ppt Turing machines. Later, in \cref{sec:UniformAttack}, we show how to handle the general case.  In \cref{sec:MartingalesAttack,sec:LapAttack,sec:SingeltonAttack,section:nugget}, we prove the following weaker variant of \cref{thm:Unimain}.
 
\begin{theorem}[Main theorem, non-uniform adversaries variant]\label{thm:main}
	There exists a fail-stop adversary $\Ac$ such that the following holds. Let $\Pi$ be a correct $\np$-party $r$-round coin-flipping protocol, and let $k\in \N$ be the smallest integer such that ${n\choose k}\geq r\log(r)^{2k}$. Then, there exists a party $\Pc$ in $\Pi$ and a string $\adv\in \zs$ such that $\Ac^\Pi(\adv)$ controlling all parties but $\Pc$ biases the output of  $\Pc$ by $\Omega(1/\sqrt{r}\log(r)^k)$. The running time of $\Ac^\Pi$  is polynomial in the running time of $\Pi$ and $\np^k$, and it uses only oracle access to $\Pi$'s next-message function.

\end{theorem}

\paragraph{Proving  \cref{thm:main}.}
Our proof follows the high-level description given in the introduction. Recall that  a backup  value associated with  a subset of  parties  \wrt a given round of a protocol execution  is the common output  these parties  would output if all other parties prematurely abort in this round round. More formally,

\begin{notation}\label{not:backups}
	We  identify the set $[\np]$ with the parties of the  $\np$-party  protocol in consideration. We refer to subset of parties (\ie subset of [\np]) as {\sf tuples}, and denote sets of such tuples using   ``blackboard bold'' (\eg $\mathbb{S}$) rather than calligraphic. For a tuple subset $\bS\subseteq {[\np]\choose k}$ and $h\in [\np]$, let $\bS(h)=\set{\cU\in \bS\colon  h\in \cU}$, i.e.~$\bS(h)$ is the set of tuples in $\bS$ that contain $h$, and   $\bS \rmv{h}=\bS\setminus \bS(h)$.  
\end{notation}

\begin{definition}[Backup values] \label{def:backups}
	The following definitions are \wrt a fixed honest execution of an  $\np$-party, $r$-round correct protocol   (determined  by the parties' random coins).  The {\sf \ith round backup value} of  a subset of parties $\cU \subseteq  [\np]$ at round $i\in [r]$,  denoted    $\backup(\cU,i)$, is defined as the common output the parties in $u$ would output, if all other parties abort in the \ith round (set to $\perp$ if the execution  has not reached this round with all parties of $\cU$ alive). The {\sf average backup value} of  a tuples subset of  $\bS$, is defined by $\AvgBackup(\bS,i) = \frac{1}{\size{\bS}}\sum_{\cU\in \bS} \backup(\cU,i)$. Furthermore, for every $\bS$, we define the random variables $B_1^{\bS},\ldots,B_r^{\bS}$ to denote the value of  $\AvgBackup(\bS,1),\ldots, \AvgBackup(\bS,r)$ in a random execution of $\Pi$.
\end{definition}

Back to the informal proof-sketch. For a subset of parties $\cs \subseteq [\np]$,  consider the average backup value of the tuples in $\bS_1 = \cs^k$, i.e.~$B_1^{\bS_1},\ldots,B_r^{\bS_1}$, where $B_i^{\bS_1}$ denotes the value of  $\AvgBackup(\bS_1,i)$ in a random execution of $\Pi$.  Let $X_0,\ldots,X_r$  be the Doob-like sequence defined by $X_i=\ex{\out\mid B_i^{\bS_1}, X_{i-1}, \sum_{j\le i} (X_j-X_{j-1})^2}$. It is easy to see that $X_0,\ldots,X_r$ is a \sosa weak martingale. In \cref{sec:martingales}, we showed that such sequences have at least one $1/\sqrt{r}$-gap between consecutive variables, with constant probability. In turn, such a gap enables a $1/\sqrt{r}$-attack, unless the sequence $B_1^{\bS_2},\ldots,B_r^{\bS_2}$  for $\bS_2 = \cs^{k-1} \times \wb{\cs}$,  and the above sequence $B_1^{\bS_1},\ldots,B_r^{\bS_1}$ are \emph{non-similar}:  there is a $1/\sqrt{r}$-gap between $B_i^{\bS_2}$ and $B_i^{\bS_2}$ in some round $i$.  If so, we can attempt to exploit the non-similarity by  applying our differential privacy-based attack, dubbed \emph{the oblivious sampling attack}, in the spirit of the oblivious sampling experiment described in \cref{sec:Laplace}. In order for the latest attack to achieve the desired bias, we require that for any two parties  $h,h'\in [\np]$, the \emph{projection} of the sequence $B_1^{\bS_1 },\ldots,B_r^{\bS_1 }$ to $h$ and $h'$, defined by  $B_1^{\bS_1(h)},\ldots,B_r^{\bS_1(h)}$ and $B_1^{\bS_1(h')},\ldots,B_r^{\bS_1(h')}$, respectively, yields similar sequences (and the same for $\bS_2$). If not, i.e.~there is a pair of parties $h$, $h'\in \cH$ and $z\in \ot$ such that $B_1^{\bS_z(h)},\ldots,B_r^{\bS_z(h)}$ and $B_1^{\bS_z(h')},\ldots,B_r^{\bS_z(h')}$ are non similar, then we can invoke the oblivious sampling attack with $\wt{\bS}_1=\bS_z(h)$ and $\wt{\bS}_2=\bS_z(h')$, which yields the desired bias, as long as all relevant pairs of parties induce projected sequences that are similar. If not, we can find another pair of parties that breaks our requirement, i.e.~similarity, and repeat the process.

The iterative process described above terminates by finding a non-similar pair of tuple-sets $(\bS_1',\bS_2')$ such that, either every projection is similar, and thus we can apply the oblivious sampling attack, or, $\bS_1'$ and $\bS_2'$ consist of tuples in which all-but-one parties are fixed, i.e.~every projection describes the distribution of a single bit. If so, we can apply a simple attack that we call the \emph{singletons attack}.  We refer to the ``level'' where the process stops as the \emph{nugget} of $\Pi$. 

The actual proof is significantly more complicated, as we have to use a different similarity measure for every level, and we have to make sure the projected sets of  tuples have the right size.

\medskip

We formally prove the theorem using the following four lemmas, proved in  \cref{sec:MartingalesAttack,sec:LapAttack,sec:SingeltonAttack,section:nugget}. \cref{lemma:UsefulNaget} state that any protocol has a nugget  (formally defined in \cref{def:UsefulNaget}), where \cref{lemma:short:mart,lemma:short:lap,lemma:short:sing} state that there is an effective attack, for all possible values of the nugget.

\begin{notation}\label{not:UsefulNaget}
Let $\coef{\ell}= \frac{(\np-1)\cdot (\np-2)\cdot \ldots \cdot (\np-k+\ell)}{(k-1)\cdot (k-2)\ldots \cdot \ell}$, letting $\coef{k}= 1$.  For $r\in \N$, let  $\Rng (r) = \set{1,1+1/r, 1+2/r, 1+3/r\ldots,r}$. We remark that $\size{\cR}=r(r-1)+1$.
\end{notation}

\def \NuggetDef{
	Let $\Pi$  be an  $\np$-party $r$-round coin-flipping protocol, and let $k\in \N$ be the smallest integer such that ${\np \choose k}\geq r\log(r)^{2k}$. Index $\ks \in [k+1]$ is  a {\sf nugget} for $\Pi$, if  there exists  $ \rho^\ast \in  \Rng(r)$, set  $\H \subseteq [\np]$, and tuple sets  $\bS_1, \bS_0\subseteq {\np \choose k}$ such that the following holds.

For a tuple-set $\bS\subseteq 2^{[\np]}$ and  $i\in[r]$, let $B_i^\bS$ denote  the value of $\AvgBackup(\bS,i)$ in a random execution of $\Pi$.  The following holds according to the value of $\ks$:
 \begin{description}
 	\item  [$\ks =1$:] ~
 	
  \begin{enumerate}
  	
  	\item  $\pr{\max_{i\in [r]}\size{ B_i^{\bS_1}-  B_i^{\bS_0}}\geq  \frac{\rho^\ast}{ 256\sqrt{r}} \cdot \frac{ \coef{\ks} ^{1/2}}{\left (64\log(r)\right)^{k-\ks}} } \geq \frac{1}{2\rho^\ast\log(r)}\cdot \frac{ 64^{-k+\ks}}{ \coef{\ks}^{1/2} }$. \label{eq:nug}

 	\item $\H   \ge n/3$, $\size{\bS_1}=\size{\bS_0}=\size{\H}$, and $\size{\bS_z(h)}=1$ for every $h\in \H$ and $z\in  \zo$.
 	
 \end{enumerate}
 	
 		\item[$\ks \in \set{2,\ldots, k}$:] ~
 		
 \begin{enumerate}
 			
 \item Same as \cref{eq:nug}  for  $\ks =1$.	 		
 		
 	\item  For  every $h, h' \in \H$, $z,z'\in \zo$,  $\cU'\in \bS_z$ and  $\rho\in \Rng(r)$:
 	
 	\begin{enumerate}
 		
 		\item 	  $\pr{\max_{i\in [r]}\size{ B_i^{\bS_z(h)}-  B_i^{\bS_z(h')}}\geq  \frac{\rho}{ 256\sqrt{r}} \cdot \frac{ \coef{\ks-1} ^{1/2}}{\left (64\log(r)\right)^{k-\ks+1}} } \leq \frac{1}{2\rho\log(r)}\cdot \frac{ 64^{-k+\ks-1}}{ \coef{\ks-1}^{1/2} }$.

 		\item $\ppr{\cU\la \bS_{z}}{h\in \cU}=\ppr{\cU\la \bS_{z}}{h'\in \cU}\le \frac{1}{2}$.
 		
 		\item $\ppr{\cU\la \bS_{z}}{h\in \cU}=\ppr{\cU\la \bS_{z'}}{h \in \cU}$.
 		
 		\item $ \ppr{\cU\la \bS_{z}}{h\notin \cU}/\ppr{\cU\la \bS_{z}}{h\in \cU} \geq \frac{1}{4}\cdot \frac{n-k+\ks-1}{\ks-1}$.
 		
 		\item $\ppr{h\la \H, \cU\la \bS_z(h)}{\cU=\cU'}=\ppr{\cU\la \bS_z}{\cU=\cU'}$ . 
 		
 	\end{enumerate}
\end{enumerate}

  \item [$\ks= k+1$:]~
 	 \begin{enumerate}
 	 	
 	 	\item $\bS_1(h)=\emptyset$ for every $h\in \H$.
 	 	
 	 	\item  	$\ppr{h\la \H, \cU\la \bS_0(h)}{\cU=\cU'}=\ppr{\cU\la \bS_0}{\cU=\cU'}$  for every  $\cU'\in \bS_0$. 
 	 	
 	 	\item $\pr{\max_{i\in [r]}\size{ B_i^{\bS_0 }-  B_i^{\bS_1 }}\geq  \frac{\rho}{ 256\sqrt{r}}   } \leq \frac{1}{2\rho\log(r)} $ for every $\rho\in \Rng(r)$.
  \end{enumerate}
 \end{description}	
}
\begin{definition}[The Nugget]\label{def:UsefulNaget}
\NuggetDef
\end{definition}

\begin{lemma}\label{lemma:UsefulNaget}
Let $\Pi$  be an  $\np$-party $r$-round coin-flipping protocol, then $\Pi$ has a nugget.
\end{lemma}

\begin{lemma}\label{lemma:short:mart}
	There exists a fail-stop adversary $\Ac$ such that the following holds. Let $\Pi$ be a correct $\np$-party $r$-round coin-flipping protocol, and let $k\in \N$ be the smallest integer such that ${n\choose k}\geq r\log(r)^{2k}$.

	 Suppose  $\Pi$ admits a nugget  $\ks=k+1$, then exists party $ h \in [\np]$ and a string $\adv$ such that $\Ac^\Pi(\adv)$ controlling all parties but $h$ biases the output of  $h$ by $\Omega(1/\sqrt{r}\log(r)^{k-\ks+1})$. The running time of $\Ac$  is polynomial in the running time of $\Pi$ and $\np^k$, and it uses only oracle access to $\Pi$'s next-message function.
\end{lemma}

\begin{lemma}\label{lemma:short:lap}
	Same as  \cref{lemma:short:mart} \wrt $\ks\in \set{2,\ldots,k}$.
\end{lemma}

\begin{lemma}\label{lemma:short:sing}
	Same as  \cref{lemma:short:mart} \wrt $\ks=1$.
\end{lemma}

\begin{proof}[Proof of \cref{thm:main}]
Immediately follows from  \cref{lemma:UsefulNaget,lemma:MartingalesAttack,lemma:LapAttack,lemma:SingeltonAttack}.
\end{proof}
\noindent 

In the following we  assume \wlg that $r$ is larger than some constant to be determined by the analysis. This latest assumption does not incur any loss of generality, and we use it to make sure that the term $1/\sqrt{r}\log(r)^{k-\ks+1}$ dominates over other terms.
 

\subsection{The Game Value Jump  Attack}\label{sec:MartingalesAttack}

\begin{lemma}[Restatement of \cref{lemma:short:mart}] \label{lemma:MartingalesAttack}
	There exists a fail-stop adversary $\Ac$ such that the following holds. Let $\Pi$ be a correct $\np$-party $r$-round coin-flipping protocol, and let $k\in \N$ be the smallest integer such that ${n\choose k}\geq r\log(r)^{2k}$.  Suppose there exist tuple sets  $\bS,\bS'\subseteq {[\np]\choose k}$ and set of parties $\H\subseteq [\np]$, satisfying
	\begin{enumerate}

		\item $ \bS(h)=\emptyset$ for every $h\in \H$, where $\bS(h)$ is defined according to \cref{not:backups}.
		
		\item  For every $\cu'\in \bS'$, $\ppr{h\la \H, \cu\la \bS'(h)}{\cu=\cu'}=\ppr{\cu\la \bS'}{\cu=\cu'}$.\label{item:thm:mart}

		\item $ \pr{\max_i\size{ B_i^{\bS} -  B_i^{\bS'}}\geq  \frac{\rho}{256\sqrt{r}} }\leq \frac{1}{2\rho\log(r)}$, for every  $\rho\in \set{1, 1+1/r,\ldots, r}$, where  $B_i^\bS = B_i^\bS(\Pi)$  is defined  according to \cref{def:backups}.
	
	\end{enumerate}

	Then  there exists $h\in \H$ and as a string advice $\adv$ such that $\Ac^\Pi(\adv)$ corrupting all parties but $h$ biases the output of $h$ by $\Omega(1/\sqrt{r})$. 
	
	Furthermore, the running time of $\Ac^\Pi(\adv)$ is polynomial in the running time of $\Pi$ and $\np^k$, and only  uses  oracle access to $\Pi$'s next-message function. 
\end{lemma} 

\subsubsection{The Game-Value Sequence}  \label{sec:gamevalue}

The cornerstone of the so-called game value jump  attack is that the adversary computes the expected outcome of the protocol, referred to as the \emph{game-value} and denoted $X_i$ for round $i$. Then, at every round $i\in [e]$, she compares this value to the backup value at hand, and  decides to abort if the backup value deviates from the expected outcome of the protocol significantly. Next, we formally define  the game-value sequence and  follow up with a discussion regarding some of its properties.

\vspace*{\baselineskip}
\noindent 
Define  $g:[0,1]^3\times \set{0,1}\mapsto \zo$  by
\begin{align} \label{eq:trigger}
	 g(x,y,y',\aux )=&    
	\begin{cases} 
		  \aux & \text{if }  |y-x|< 1/64\sqrt{r} \,\lor\, |y'-x|< 1/64\sqrt{r},\\ 
	  	0 & \text{otherwise;}
	\end{cases}
\end{align} 

\newcommand{\GameValue}[1]{
Let $X_0=\ex{\out }$. For $i\in [r]$, define $X_i$ such that
\begin{align*} 
X_i=\rnd{\delta}\left (\ex{\out\mid B^\bS_{i}, B^{\bS}_{i-1}, X_{i-1}, \sum_{\ell< i } (X_{\ell} - X_{\ell-1})^2, G_{i-1} }\right ) , 
\end{align*}
where $G_i=g(B^\bS_{i}, B^{\bS}_{i-1}, X_{i}, G_{i-1})$, letting $G_0=1$ and $g$ be defined according to \cref{#1}.}

\begin{definition}[Game-value sequence]\label{eq:DefX}
\GameValue{eq:trigger}
\end{definition}

\remove{
\noindent Namely, using the terminology of \cref{sec:martingales}, the sequence $X_1,\ldots, X_r$ is the Doob augmented weak martingales of $(Z,f,\hp,\delta)$ where $\hp$ is a normal form function that augments $g$. As the name suggests, the sequence $X_1,\ldots, X_r$ admits the ($2\delta$-weak) martingale property (c.f.~\cref{def:DMartingales}). By \cref{theorem:martingale}, there exists a normal form function $\hp$ that augments $g$ such that \cref{eq:Xjump} holds true, \ie with constant probability, two consecutive points of the the sequence will be at least $1/32\sqrt{r}$ far apart. We elaborate further on this sequence.
}

\noindent
Notice that $X_i$, for $i\in [r]$, is simply the discretized expected outcome of the protocol given a ``short'' aggregated accounted of the history so far. Namely, each variable $X_i$ is an approximation of the expected outcome of the protocol given the two preceding points of the backup sequence $B_i^\bS$ and $B_{i-1}^\bS$, the previous game-value $X_{i-1}$ and the sum-of squares $\sum_{\ell< i } (X_{\ell} - X_{\ell-1})^2 $, as well as a bit $G_{i-1}$ indicating whether either $B^\bS_{j}$ or $B^\bS_{j-1}$ deviated from the value of $X_j$ by more than $1/64\sqrt{r}$, for some $j<i-1$.

 \begin{remark}[Computing $X_1,\ldots, X_r$] \label{remark:advice}
	Each $X_i$ is fully determined by the index $i$ and the value of the $5$-tuple $(B^\bS_{i}, B^{\bS}_{i-1}, X_{i-1}, \sum_{\ell< i } (X_{\ell} - X_{\ell-1})^2, G_{i-1} )$. Recall that  $n^k\ge r$, $\size{\supp(B^\bS_{i})}= {\np \choose k}\in O(\np^k)$ and  $\size{\supp(X_i)}=1/\delta$, for every $i\in [r]$. Hence, there exists a table of size $\log(1/\delta)\cdot  2\cdot{ \np \choose k}^2 \cdot \frac{r^2}{\delta^3}$, such that the value of $X_i$, for all $i\in [r]$, can be computed from $B^\bS_{\le i}$ using this table. Hereafter, we fix $\delta=1/200r$ and thus the table is described by a string of polynomial-size in $n^k$.  
\end{remark} 

\noindent 
We show that the sequence $X_0,\ldots, X_r$ satisfies the martingale property according to \cref{def:SOSADMartingale}. Such sequences typically have large gaps which we will exploit in our attack.

\begin{claim}
The sequence $X_0,\ldots, X_r$ is a \sosa $2\delta$-weak martingale sequence, i.e.~for all $i\in [r]$, 
\begin{align*}
\ex{X_i\mid X_{i-1}, \sum_{\ell<i} (X_i-X_{i-1})^2} \in X_{i-1} \pm 2\delta  .
\end{align*}
\end{claim}

\begin{proof}
Fix $i\in [r]$ and define $\wt{X}_i=\ex{\out\mid B_{i}, B_{i-1}, X_{i-1}, \sum_{\ell< i } (X_{\ell} - X_{\ell-1})^2, G_{i-1})}$, $\wt{X}_{i+1}=\ex{\out\mid B_{i+1}, B_{i}, X_{i}, \sum_{\ell< i+1 } (X_{\ell} - X_{\ell-1})^2, G_{i})}$ i.e.~without rounding. By tower law (\cref{lemma:tower}), $\ex{\wt{X}_{i+1}\mid X_{i}, \sum_{\ell\le i } (X_{\ell} - X_{\ell-1})^2}=\ex{\out\mid  X_{i}, \sum_{\ell\le i } (X_{\ell} - X_{\ell-1})^2}$. Furthermore, by \cref{lemma:z:4}  $\ex{\out\mid  \wt{X}_{i}, \sum_{\ell\le i } (X_{\ell} - X_{\ell-1})^2}=\wt{X}_i$. Consequently, since $X_i$ is the rounded value of $\wt{X}_{i}$ to the closest $\delta$-multiple, by \cref{lemma:z:3}, $\ex{\out\mid  X_{i}, \sum_{\ell\le i } (X_{\ell} - X_{\ell-1})^2}\in X_i \pm \delta$. Finally, since $X_{i+1}$ is the rounded value of $\wt{X}_{i+1}$ to the closest $\delta$-multiple, we deduce that 
\begin{align*}
\ex{X_{i+1}\mid X_{i}, \sum_{\ell\le i } (X_{\ell} - X_{\ell-1})^2}&\in \ex{\wt{X}_{i+1}\mid X_{i}, \sum_{\ell\le i } (X_{\ell} - X_{\ell-1})^2}\pm \delta \\ 
&\in \ex{\out\mid  X_{i}, \sum_{\ell\le i } (X_{\ell} - X_{\ell-1})^2} \pm \delta \\
& \in X_{i} \pm 2\delta   .
\end{align*}
\end{proof}

\noindent 
By \cref{claim:SOSAjumps}, since $X_0,\ldots, X_r$ is a \sosa $2\delta$-weak martingale sequence with $X_0=1/2$, $X_r\in \set{0,1}$, and $\delta<1/200r$, it holds that

\begin{align}\label{eq:Xjump}
\pr{\exists i\in [r]\st  \size{X_i-X_{i-1}}\ge \frac{1}{4\sqrt{r}} }\ge \frac{1}{20}   . 
\end{align}
 
\subsubsection{The Attack}
 
We start with  a high-level overview of the attack. The adversary biasing  party $h\in \H$, to be chosen at random, towards zero is defined as follows (the attack biasing toward one is defined analogously). After receiving the honest party messages for round $i-1$, it computes the values of $y_{i}=B^\bS_i$, $y_{i-1}=B^\bS_{i-1}$, $x_{i}=X_i$ and $g_i=G_i$, for $X_i$ and $G_i$ being according to \cref{eq:DefX}. 

If $y_{i-1}$ is below $x_i$ by more than $1/64\sqrt{r}$, then it  aborts all parties but a random tuple of $\bS'$ that contains $h$, \emph{without} sending the \ith-round messages of the aborting parties. The surviving corrupted parties are instructed to terminate the protocol honestly.

 If $y_{i}$ is below $x_i=X_i$ by more than $1/64\sqrt{r}$, it aborts all parties but a random tuple in $\bS'$ that contains $h$, \emph{after} sending the \ith-round messages of the aborting parties. The surviving corrupted parties are instructed to terminate the protocol honestly.

The  attacker is formally defined as follows.  
\begin{algorithm}[The martingale attack $\GvAttack$] \label{alg:MartingalesAttack}~
	\item Parameters: $\bS$, $\bS'\subseteq {[\np] \choose k }$, $z\in \set{0,1}$, honest party $h\in [\np]$ and a string $\adv\in \set{0,1}^{*}$.
	\item Description:
	\begin{enumerate} 
		\item   Compute $B_{1}^{\bS}$ according to the protocols specifications. If $ (-1)^{1-z}\cdot (B_{1}^{\bS}-   \frac{1}{2}) > 1/64\sqrt{r}$, {\sf without} sending their 1st round messages, abort all parties except a random tuple in $\bS'(h)$.  
			 \begin{itemize}[--]
			\item The remaining corrupted parties are instructed to terminate the protocol honestly.
			\end{itemize}
			
		\item For $i=1,\ldots, r$: 
		\begin{enumerate}
			\item Upon receiving the \ith round messages of $h$, compute $B_{i}^{\bS}$, $B_{i+1}^{\bS}$, $X_{i+1}$ and $G_i$ using the messages received so far and the string $\adv$.
			\item If $(-1)^{1-z}\cdot (B_{i }^{\bS}-   X_{i+1}) > 1/64\sqrt{r}$ and $G_i=1$, {\sf without} sending their messages for round $i$ abort all parties except a random tuple in $\bS'(h)$. \label{ab:i}
			 \begin{itemize}[--]
			\item The remaining corrupted parties are instructed to terminate the protocol honestly.
			\end{itemize}
			
			 \item If $ (-1)^{1-z}\cdot (B_{i+1}^{\bS}-   X_{i+1}) > 1/64\sqrt{r}$ and $G_i=1$, {\sf after} sending their messages for round $i$, abort all parties except a random tuple in $\bS'(h)$. \label{ab:iplusx}
			 \begin{itemize}[--]
			\item The remaining corrupted parties are instructed to terminate the protocol honestly.
			\end{itemize}
		\end{enumerate}
	\end{enumerate} 
 \end{algorithm}

\noindent 
Let $\GvAttack(\bS, \bS', z, h, \adv)$ denote the martingale attacker with parameters $\bS, \bS', z, h, \adv$. We refer to the round in which the adversary instructs some parties in its  control  to abort as  the \emph{aborting round}, set to $r$ if no abort occurred.
 
\subsubsection{Success probability of \cref{alg:MartingalesAttack}.}\label{section:martsuccess} Let $\bS$, $\bS'$ and $\H$ be as in \cref{lemma:MartingalesAttack}, and let $H$ denote an element of $\H$ chosen uniformly at random. Following the discussion of \cref{remark:advice}, let $\adv$ denote a string of size polynomial in $n^k$ that fully describes the sequence $X_1,\ldots, X_r$ that is defined according to \cref{eq:DefX}. We show that either $\Ac_1(H)=\GvAttack(\bS, \bS', 1, H,\adv)$ or $\Ac_0(H)=\GvAttack(\bS, \bS', 0, H, \adv)$ succeeds in obtaining the bias of \cref{lemma:MartingalesAttack}.

Before proceeding with the proof, we introduce a last piece of notation. For $z\in \set{0,1}$, let $J^{z\ast}$ denote the round-index where the adversary $\Ac^{z}$ decided to abort certain parties, and let $J^{z}$ denote the round-index of the last messages sent by those aborting parties. Namely, in \stepref{ab:i} of \cref{alg:MartingalesAttack} we have $J^{z}=J^{z\ast}-1=i$ and in \stepref{ab:iplusx} of \cref{alg:MartingalesAttack} we have $J^{z}=J^{z\ast}=i+1$. If no abort occurred, $J^{z}=J^{z\ast}=r$.

\cref{lemma:MartingalesAttack} follows from the claims below.

\begin{claim}\label{claim:mart:abort}
	$\pr{J^{1}\neq r}+\pr{J^{0}\neq r}\geq 1/20$.
\end{claim}

\begin{claim}\label{claim:mart:exp}
	For $z\in \set{0,1}$, $\ex{X_{J^{z\ast}}}\in 1/2 \pm \frac{1}{200r} $.
\end{claim}

\begin{claim}\label{claim:mart:err}
	$\ex{\max_i |B_{i}^{\bS}-B_{i}^{\bS'}|}\leq 1/128\sqrt{r} $, for $r$ large enough.
\end{claim}

\noindent Before proving each of these claims, we show how to combine them to obtain the lemma.

\begin{proof}[Proof of \cref{lemma:MartingalesAttack}]
	By \cref{claim:mart:abort}, we may assume without loss of generality that $\pr{J^1\neq r}\geq 1/10$. Next, we compute the bias caused by the attacker $\Ac_1(H)$. By \cref{item:thm:mart} of \cref{lemma:MartingalesAttack}, the output of the honest party is identically distributed with $B_{J^1}^{\bS'}$. Compute
\begin{align}
	 \ex{B^{\bS'}_{J^1}}-1/2 &\geq  \ex{B^{\bS}_{J^1}}-1/2  -  \ex{B^{\bS}_{J^1}}+\ex{B^{\bS'}_{J^1}}  \nonumber\\
	&\geq  \ex{B^{\bS}_{J^1}}-\ex{X_{{J^{1\ast}}}}    -  \ex{\max_i \size{B^{\bS}_i - B^{\bS'}_i}}\cdot  \pr{J^1\neq r} - \frac{1}{200r} \label{eq:mart:exp}\\
	&\geq  \pr{J^1\neq r}\cdot \left(\ex{B^{\bS}_{J^1} - X_{{J^{1\ast}}}\mid J^1\neq r}  - \ex{\max_i \size{B^{\bS}_i - B^{\bS'}_i}}\right )  -\frac{1}{200r} \label{eq:mart:err}  \\
	&\geq \pr{J^1\neq r}\left (\frac{1}{64\sqrt{r}} - \frac{1}{128\sqrt{r}} \right )  -\frac{1}{200r} \nonumber \\ &
	\geq   \frac{1}{40} \cdot \frac{1}{128\sqrt{r}}-\frac{1}{200r}. \nonumber
\end{align}  \cref{eq:mart:exp} follows from triangle inequality, union bound, \cref{claim:mart:exp} and the fact that $B^{\bS}_r=B^{\bS'}_r$. \cref{eq:mart:err} follows from \cref{claim:mart:err} and the fact that $B^{\bS}_{J^1} - X_{J^{1\ast}}\geq 1/64\sqrt{r}$, whenever $J^1\neq r$.
\end{proof}

\paragraph*{Proof of \cref{claim:mart:abort}.} First, we lower-bound the probability of abort by the probability of having a large increment in the $X$-sequence alone. For convenience, we introduce the following notation. For $z\in \set{0,1}$, let $\trigg_{i+1}^z$ denote the predicate $(-1)^{1-z}(B_{i}^{\bS} -X_{i+1}) \geq 1/64\sqrt{r} \;\lor\; (-1)^{1-z}( B^{\bS}_{i+1}-X_{i+1}) \geq 1/64\sqrt{r} $ and let $\trigg_{1}^z $ denote the predicate $ (-1)^{1-z}(B_{1}^{\bS} - \frac{1}{2}) \geq 1/64$. We remark that $\trigg_{\ell}^z$ denotes whether the attack biasing towards $z\in \zo$ is potentially ``triggered'' at round $\ell$. Write $\trigg_{i+1}=\trigg_{i+1}^0 \lor \trigg_{i+1}^1$.  Recall that 
\begin{align*}
	\pr{J^z\neq r}&=\pr{\trigg_{1}^z \vee \bigvee_{i=1}^{r-1}\left (G_i=1 \land  \trigg_{i+1}^z\right )} .
\end{align*}
Thus, by union bound,
\begin{align*}
	\pr{J^0\neq r}+\pr{J^1\neq r}&\geq \pr{\trigg_1\vee \bigvee_{i=1}^{r-1}\left (G_i=1\land \trigg_{i+1} \right ) }.
\end{align*}

\noindent Recall that $G_i=1$ is equivalent to $\bigwedge_{j=1}^{i}  \neg \trigg_j\equiv \neg\left ( \bigvee_{j=1}^{i} \trigg_j\right )$. It follows that 
\begin{align*}
	\bigvee_{i=1}^{r-1}\left (G_i=1\land \trigg_{i+1} \right )&\equiv \bigvee_{i=1}^{r-1}\left (\trigg_{i+1} \land \neg\left ( \bigvee_{j=1}^{i} \trigg_j\right ) \right )\\
	& \equiv \bigvee_{i=1}^{r-1} \trigg_{i+1}.
\end{align*}

\noindent We can thus lower-bound $\pr{J^0\neq r}+\pr{J^1\neq r}$ by $\pr{\bigvee_{i=1}^{r-1} \trigg_{i+1}}$.

\begin{align*}
	\pr{J^0\neq r}+\pr{J^1\neq r}&\geq\pr{\trigg_{1}\vee \bigvee_{i=1}^{r-1} \trigg_{i+1}}\\
	&=\pr{\size{B_1^\bS-X_1}\vee \bigvee_{i=1}^{r-1}\left ( \size{B_{i}^{\bS} -X_{i+1}}\geq 1/64\sqrt{r}\vee \size{B^{\bS}_{i+1}-X_{i+1}}\geq 1/64\sqrt{r}\right ) }\\
	&=\pr{ \bigvee_{i=1}^{r-1}\left ( \size{B_{i}^{\bS} -X_{i}}\geq 1/64\sqrt{r}\vee \size{B^{\bS}_{i}-X_{i+1}}\geq 1/64\sqrt{r}\right ) }\\
	&\geq \pr{ \bigvee_{i=1}^{r-1} \size{X_{i+1}  -X_{i} }\geq 1/32\sqrt{r} }
\end{align*}
By \cref{eq:Xjump}, we conclude that $\pr{ \bigvee_{i=1}^{r-1} \size{X_{i+1}  -X_{i} }\geq 1/32\sqrt{r} }\geq 1/20 $.

\qed

\paragraph*{Proof of \cref{claim:mart:exp}.} Recall that $\delta=1/200r$ and $ X_{i}  =   \rnd{\delta}(\ex{\out\mid B_{i}^\bS,B_{i-1}^\bS, X_{i-1}, \sum_{\ell<i} (X_\ell-X_{\ell-1})^2 , G_i}) $. For conciseness, write $\hist_i$ for the $5$-tuple $ (B_{i}^\bS,B_{i-1}^\bS, X_{i-1}, \sum_{\ell<i} (X_\ell-X_{\ell-1})^2 , G_i)$. We compute $\ex{X_{J^{z\ast}}}=\sum_i \ex{X_{i}\mid J^{z\ast}=i}\cdot \pr{J^{z\ast}=i}$. Let us focus on the term $\ex{X_{i}\mid J^{z\ast}=i}$.
\begin{align}
	\ex{X_{i}\mid J^{z\ast}=i}&= \ex{ \rnd{\delta}(\ex{\out\mid \hist_i})\mid J^{z\ast}=i}\\
	 &\in \ex{  \ex{\out\mid  \hist_i}\mid J^{z\ast}=i} \pm \delta. \nonumber
\end{align}
Since $ \hist_i$ fully determines $X_{i}$, $B^\bS_i$, $B^\bS_{i-1}$ and $J^{z\ast}\geq i$, it follows that $\hist_i$ fully determines $J^{z\ast}=i$, which implies that 
\begin{align}
	\ex{\ex{\out\mid \hist_i}\mid J^{z\ast}=i}=\ex{\out \mid J^{z\ast}=i}.
\end{align}
Since,  by assumption,  $\ex{\out}=1/2$, it follows that $\ex{X_{J^{z\ast}}}\in\sum_{i\in r}\ex{\out \mid J^{z\ast}=i}\cdot \pr{J^{z\ast}=i}\pm \delta = \ex{\out} \pm \delta=1/2\pm \frac{1}{200r}$.
\qed

\paragraph*{Proof of \cref{claim:mart:err}.}

From the hypothesis of \cref{lemma:MartingalesAttack}, it holds that 

\begin{align*}
	\forall \rho\in \set{1,1+1/r,\ldots, r}, \quad \pr{\max_{i\in [r]} \size{B_i^\bS-B_i^{\bS'}}\geq \rho\cdot\frac{1}{ 256\sqrt{r}} }\leq \frac{1}{2\rho \log(r)}  .
\end{align*} 
For convenience, write $B_{\max}=\max_{i\in [r]} \size{B_i^\bS-B_i^{\bS'}}$ and let us compute $\ex{B_{\max}}$. 
\begin{align*}
	 \ex{B_{\max}}&=\ex{B_{\max}\mid B_{\max}\leq 1/256\sqrt{r}}\cdot \pr{ B_{\max}\leq 1/256\sqrt{r}}\\
	& +\sum_{j=1}^{\log(256\sqrt{r})} \ex{B_{\max}\mid 256\sqrt{r}\cdot B_{\max}\in [2^{j-1}, 2^j]}\cdot \pr{256\sqrt{r}\cdot B_{\max}\in [2^{j-1}, 2^j]} \\
	&\leq \frac{1}{256\sqrt{r}} + \sum_{j=1}^{\log(256\sqrt{r})} \frac{2^{j }}{256\sqrt{r}}\cdot \frac{1}{2^{j-1}\cdot 2\log(r)} \\
	& = \frac{1}{256\sqrt{r}} + \frac{1}{256\sqrt{r}\log(r)} \cdot \left (\frac{1}{2}\cdot \log(r)+ \log(256) \right ) \\
	&\leq \frac{1}{128\sqrt{r}}   ,
\end{align*} where the last inequality holds for large enough $r$. 
\qed

\subsection{The Differential Privacy Based Attack}\label{sec:LapAttack}

\begin{lemma}[Restatement of \cref{lemma:short:lap}]\label{lemma:LapAttack}
	There exists a fail-stop adversary $\Ac$ such that the following holds. Let $\Pi$ be a correct $\np$-party $r$-round coin-flipping protocol, and let $k\in \N$ be the smallest integer such that ${n\choose k}\geq r\log(r)^{2k}$.   Suppose there exists $\ks\in \set{2,\ldots, k}$,  $\rho^\ast\in \set{1,1+1/r, \ldots, r}$,	tuple sets $\bS_1$ and $\bS_0\subseteq {[\np]\choose k}$ and party set $\H\subseteq [\np]$, such that
	\begin{itemize} 
		 \item For every $h, h' \in \H$, $z,z'\in \zo$ and $\cU'\in \bS_z$:
			
			\begin{itemize}
				\item $\ppr{\cU\la \bS_{z}}{h\in \cU}=\ppr{\cU\la \bS_{z}}{h'\in \cU}\le \frac{1}{2}$.
				\item $\ppr{\cU\la \bS_{z}}{h\in \cU}=\ppr{\cU\la \bS_{z'}}{h \in \cU}$.
				
				\item $ \ppr{\cU\la \bS_{z}}{h\notin \cU}/\ppr{\cU\la \bS_{z}}{h\in \cU} \geq \frac{1}{4}\cdot \frac{n-k+\ks-1}{\ks-1}$.
				
				\item $\ppr{h\la \H, \cU\la \bS_z(h)}{\cU=\cU'}=\ppr{\cU\la \bS_z}{\cU=\cU'}$ .
				
			\end{itemize}
		\item Letting  $B^{\bS_z}_i = B_i^{\bS_z}(\Pi)$ and $\coef{\cdot}$ be according is according to \cref{def:backups,not:backups}:
		\begin{align}\label{eq:lemma:lapgap}
		\pr{\max_{i\in [r]}\size{ B_i^{\bS_1}-  B_i^{\bS_0}}\geq  \frac{\rho^\ast}{ 256\sqrt{r}} \cdot \frac{ \coef{\ks} ^{1/2}}{\left (64\log(r)\right)^{k-\ks}} } \geq \frac{1}{2\rho^\ast\log(r)}\cdot \frac{ 64^{-k+\ks}}{ \coef{\ks}^{1/2} }.
		\end{align}
		\item Letting $B^{\bS_z (h)}_i=B^{\bS_z(h)}_i(\Pi)$  be according to \cref{def:backups}, for every $z\in\zo$,  $h,h'\in \H$ and  $\rho\in \set{1,1+1/r, \ldots, r}$, it holds that:
		\begin{align}\label{eq:lemma:lapsim}
			&\pr{\max_{i\in [r]}\size{ B_i^{\bS_z(h)}-  B_i^{\bS_z(h')}}\geq  \frac{\rho}{ 256\sqrt{r}} \cdot \frac{ \coef{\ks-1} ^{1/2}}{\left (64\log(r)\right)^{k-\ks+1}} } \leq \frac{1}{2\rho\log(r)}\cdot \frac{ 64^{-k+\ks-1}}{ \coef{\ks-1}^{1/2} }.
		\end{align} 
	\end{itemize} Then, there exists $h\in \H$ such that $\Ac^\Pi(\bS_1,\bS_0,\H,k^\ast,\rho^\ast)$ corrupting all parties but $h$ biases the output of $h$ by $\Omega(1/\sqrt{r}\log(r)^{k-\ks+1})$.

 Furthermore, the running time of $\Ac^\Pi(\bS_1,\bS_0,\H,k^\ast,\rho^\ast)$ is polynomial in the running time of $\Pi$ and $\np^k$, and it uses only oracle access to $\Pi$'s next-message function.
 \end{lemma}

\subsubsection{The Attack}

We start with  a high-level overview of the attack using the notation of \cref{lemma:LapAttack}.  
  
The adversary corrupts all parties except a random  party $h\in \H$. After receiving the honest party \ith  message, it adds Laplace noise to the quantity $B_i^{\bS_1\rmv{h}}-B_i^{\bS_0\rmv{h}}$, \ie the difference between the average backup values for those tuples that do not contain $h$. If the resulting quantity is above some value $\gamma$, the adversary aborts all parties except a random tuple in $\bS_z(h)$, for $z\in  \zo$ being  the direction of the bias the adversary wishes to attack towards\footnote{The choice of $\gamma$ and of the Laplace parameter is dictated by the magnitude of the gap between $B_i^{\bS_1}$ and $B_i^{\bS_0}$ as stated in \cref{eq:lemma:lapgap}.}.

Since, by assumption, the values  $B_i^{\bS_z\rmv{h}}$ and $B_i^{\bS_z(h)}$  are not too far apart, adding Laplace noise  ``decorrelates'' the abort decision from the identity of the honest party party $h$.  Thus, $B_i^{\bS_z(h)}$ is roughly distributed like the mean $B_i^{\bS_z }$ (and by extension $B_i^{\bS_z\rmv{h} }$ as well). Therefore, either the adversary biasing towards one or the adversary biasing towards zero succeeds in its attack, since either $\ex{B_J^{\bS_1}}> 1/2$ or $\ex{B_J^{\bS_0}}<1/2$, where $J$ denote the aborting  round.

The formal description of the attack is given below.

\begin{algorithm}[$\DpAttack$: The differential privacy based attack]\label{alg:LapAttack}~ 
	\item Parameters: $\bS_1$, $\bS_0\subseteq {\np \choose k }$, $z\in \zo$, party $h\in [\np]$ and $\gamma\in [0,1]$.
	\item Notation: Let $\lambda=\gamma/4\log(r)$.
	\item Description:
	\begin{enumerate} 
		\item For $i=1,\ldots, r$: 
		\begin{enumerate}
		 	\item Upon receiving the \ith-round messages of $h$, compute $B_i^{\bS_1\rmv{h}}$ and $B_i^{\bS_0\rmv{h}}$.
			\item Sample $\nu_i \leftarrow \Lap{\lambda}$.
			\item If $B_i^{\bS_1\rmv{h}}-B_i^{\bS_0\rmv{h}} + \nu_i> \gamma$, without sending their messages for round $i$, abort all parties except a random tuple in $\bS_z(h)$. 
			\begin{itemize}
			\item[--] The remaining corrupted parties are instructed to terminate the protocol honestly. 
			\end{itemize} 
		\end{enumerate}
	\end{enumerate} 
\end{algorithm} 

Let $\DpAttack(\bS_1, \bS_0, z, h, \gamma)$ denote the above attacker with parameters $\bS_1, \bS_0, z, h, \gamma$. We refer to the round in which the adversary instructs some parties in its  control  to abort as  the \emph{aborting round}, set to $r$ is not abort happen. 
\subsubsection{Success probability of \cref{alg:LapAttack}} 

Let $H$ be a uniform element of $\H$, and let $\gamma =   \frac{\rho^\ast}{256\sqrt{r}}  \cdot \frac{\coef{\ks}^{1/2}}{\left  (64\log(r)\right )^{k-\ks}}  $. We show that either $\Ac_1(H)=\DpAttack(\bS_1, \bS_0, 1, H, \gamma)$ or $\Ac_0(H)=\DpAttack(\bS_1, \bS_0, 0, H,\gamma)$ succeeds in obtaining the bias of \cref{lemma:LapAttack}. Let $J$ denote the smallest round $i$ such that $B_i^{\bS_1\rmv{H}}-B_i^{\bS_0\rmv{H}}+\Lap{\lambda}\geq \gamma$, and $J=r$ if no such round exists. \cref{lemma:LapAttack} follows from the next claim.

\begin{claim}\label{claim:obsampattack} 
	$\ex{B_{J}^{\bS_1(H)}-B_{J}^{\bS_0(H)}}\geq \frac{1}{2^{16}}\cdot \frac{1}{\sqrt{r}\log{r}}\cdot \left (\frac{1}{64^2\log(r)}\right )^{k-\ks}$.
\end{claim}

\begin{proof}[Proof of \cref{lemma:LapAttack}]
	If $\ex{B_{J}^{\bS_1(H)}-B_{J}^{\bS_0(H)}}\geq \eps$, then either $\ex{B_{J}^{\bS_1(H)}}\geq 1/2 +\eps/2$ or $\ex{B_{J}^{\bS_0(H)}}\leq 1/2 - \eps/2$. By using the appropriate $\eps$ from \cref{claim:obsampattack} and observing that, under adversary $\Ac_z$, the honest party's output is identically distributed with $B^{\bS_z(H)}_{J}$, we obtain the desired statement.
\end{proof}

\paragraph*{Proof of \cref{claim:obsampattack}.} Define $\delta=\frac{1}{2}\cdot \frac{1}{2\rho^\ast\log(r)} \cdot   \frac{64^{ -k+\ks }}{\coef{\ks}^{1/2}}$, $\alpha = \frac{\gamma}{32\log(r)}\cdot \frac{\sqrt{\np-k+\ks-1}}{\sqrt{\ks-1}} $ and $\beta = \frac{\delta}{16}\cdot \frac{ \sqrt{\ks-1}}{\sqrt{\np-k+\ks-1}}$. From the hypothesis of \cref{lemma:LapAttack} and the definition of $\alpha$, $\beta$, $\gamma$ and $\delta$, it holds that $\pr{\max_i \size{B_i^{\bS_1}- B_i^{\bS_0}}\geq \gamma }\geq 2\delta$ and $\pr{\max_i \size{B_i^{\bS_z(h)}- B_i^{\bS_z }}  \geq  \rho \cdot  \alpha/2\rho^\ast}\le  \beta\cdot \rho^\ast/2\rho $, for every $z\in \zo$, $h\in \cH$ and $\rho\in \cR$. Thus, by triangle inequality, union bound and the fact that $\rho\in \cR$ can be chosen arbitrarily, the following inequalities hold without loss of generality.

\begin{align}
	&\pr{\max_i B_i^{\bS_1}- B_i^{\bS_0}\geq \gamma }\geq \delta , \\ 
	\forall h\in \H,\forall \rho\in \cR\colon \qquad &\pr{\max_i \size{\left (B_i^{\bS_1(h)}- B_i^{\bS_0(h)}\right ) - \left ( B_i^{\bS_1}- B_i^{\bS_0}\right )}\geq  \rho \cdot \alpha}\leq  \beta/\rho  .\label{eq:attack:promise}
\end{align}

\noindent Let $\tau$ denote an arbitrary transcript of $\Pi$ and let $s^h_i(\tau)$ and $s^{\rmv{h}}_i(\tau)$ denote the value of $B_i^{\bS_1(h)}- B_i^{\bS_0(h)}$ and $B_i^{\bS_1\rmv{h}}- B_i^{\bS_0\rmv{h}}$, respectively, for transcript $\tau$. Further define $s_i(\tau)=\frac{1}{n}\sum_{h\in \H} s^h_i(\tau)$, and, for arbitrary $h\in \H$ and $z\in \zo$, let $p=\ppr{\cU\la \bS_z}{h\in \cU}$. We remark that the value of $p$ does not depend on $h$ or $z$. Next, by the definition of $s_i^h(\tau)$ and the hypothesis of the theorem, we observe that  
\begin{enumerate}
\item $	p\cdot s^h_i(\tau) + (1-p)\cdot s^{\rmv{h}}_i(\tau)= s_i(\tau)$, and 
\item  $\frac{1-p}{p} \geq \frac{1}{4}\cdot\frac{n-k+\ks-1}{\ks-1}   $.
\end{enumerate}

 By definition, the adversary $\Ac_z$ aborts (some parties)  if it finds out that  $s^{\rmv{h}}_{i}(\tau) + \Lap{\lambda}\geq \gamma$. Let $T$ be the value of $\tau$,  and $J$ be the aborting round in a random execution of $\Pi$ in which the adversary $\Ac_z$ attacking the honest party $H$. Using the terminology of \cref{sec:Laplace}, the value of $s^{H}_{J}(T)$ is equal to the output of an oblivious sampling experiment with parameters $\H$, $\set{s^h_i(\tau\la T)}_{h,i}$, $\gamma$, $p$, $\lambda$. From the choice of $\alpha$ and $\beta$, and under the guarantee of \cref{eq:attack:promise}, \cref{thm:LaplaceGen} yields that $\ex{s^H_{J}(T)}\geq \gamma\delta/125 - \frac{1}{2r}\in \Omega(1/\sqrt{r}\log(r)^{k^\ast - k +1})$, for $r$ large enough.

\qed

 
\subsection{The Singletons Attack} \label{sec:SingeltonAttack}

\begin{lemma}[Restatement of \cref{lemma:short:sing}]\label{lemma:SingeltonAttack}
	 There exists a fail-stop adversary $\Ac$ such that the following holds. Let $\Pi$ be a correct $\np$-party $r$-round coin-flipping protocol, and let $k\in \N$ be the smallest integer such that ${n\choose k}\geq r\log(r)^{2k}$.  Suppose there exists $\rho^\ast\geq 1$, tuple sets $\bS_1,\bS_0\subseteq{[\np]\choose k}$ and party set $\H\subseteq[\np]$   such that:
	 \begin{enumerate} 
	 	\item $\cH\ge n/3$ and  $\size{\bS_0}=\size{\bS_1}=\size{\H}$.
	 	
	 	\item[]   \hspace{-.2 in} For every $h\in \H$ and $z\in  \zo$: 
		\item $\size{\bS_z(h)}=1$, letting $\bS_z(h)$ be according to \cref{not:backups}.

		\item $\pr{\max_{i\in [r]}\size{ B_i^{\bS_1}-  B_i^{\bS_0 }}\geq  \frac{\rho^*}{256\sqrt{r}} \cdot \dfrac{   {\np-1 \choose k-1}^{1/2}}{\left (64\log(r)\right)^{k-1}}} \geq \frac{1}{2\rho^*\log(r)}\cdot   \frac{64^{-k+1} }{  {\np -1\choose k-1}^{1/2} } $ \label{eq:lemma:approxgap}
	
		letting $B_i^{\bS_z}$ be according to \cref{def:backups}.
	\end{enumerate} 
	Then, there exists $h\in \H$ such that $\Ac^\Pi(\bS_1,\bS_0,\H,k^\ast,\rho^\ast)$ corrupting all parties but $h$, biases the output of $h$ by $\Omega(1/\sqrt{r}\log(r)^k)$. 
	
	Furthermore, the running time of $\Ac^\Pi(\bS_1,\bS_0,\H,k^\ast,\rho^\ast)$ is polynomial in the running time of $\Pi$ and $\np^k$, and it uses only oracle access to $\Pi$'s next-message function.
 \end{lemma}

\subsubsection{The Attack}

We start with  a high-level overview of the attack. The adversary biasing  a party    $h\in \H$, to be chosen at random, towards zero is defined as follows (the attack biassing toward one is defined analogously). Before the protocols starts, the adversary samples \emph{half} of the tuples in $\bS_1$ and $\bS_0$ not containing $h$, denoted $\bE_1 $ and $\bE_0$ respectively.  Upon receiving  the \ith message from  $h$, it computes the difference between the average backup values of the tuples in $\bE_1 $ and $\bE_0$, denoted $B_i^{\bE_1 } - B_i^{\bE_0}$. If the resulting quantity is above  $3\gamma/4$,  it aborts all parties except the unique tuple in $\bS_0(h)$.\footnote{The choice of $\gamma$ is dictated by the magnitude of the gap between $B_i^{\bS_1}$ and $B_i^{\bS_0}$ as stated in Assumption \eqref{eq:lemma:approxgap}.} For  the attack to go through, it is required that  $B_i^{\bE_z }$ and $B_i^{\bS_z }$  are not too far apart. Thankfully, standard concentration bounds guarantee that to be the case.

The formal description of the attack is given below.
\begin{algorithm}[The singletons attacker $\ApproxAttack$] \label{alg:SingeltonAttack}~ 
	\item Parameters: tuple subsets $\bS_1,\bS_0\subseteq {\np \choose k }$, $z\in  \zo$, honest party  $h\in [\np]$ and $\gamma\in [0,1]$.
	\item Description:
	\begin{enumerate} 
		\item For $z\in \zo$, let $\bE_z \subseteq\bS_z\rmv{h}$ denote random subset of size $\size{\bS_z}/2$.
		\item For $i=1,\ldots, r$: 
		\begin{enumerate}
			\item Upon receiving the \ith-round messages of $h$, compute $B_i^{\bE_1}$ and $B_i^{\bE_0}$.
			\item If $B_{i }^{\bE_1}- B_{i }^{\bE_0}  > 3\gamma/4$, without sending their messages for round $i$, abort all parties except the unique random tuple in $\bS_z(h)$.
			\begin{itemize}
			\item[--] The remaining corrupted parties are instructed to terminate the protocol honestly. 
			\end{itemize} 
		\end{enumerate}
	\end{enumerate} 
\end{algorithm}

\noindent 
Let $\ApproxAttack(\bS_1, \bS_0, z, h,\gamma)$ denote the singletons attacker with parameters $\bS_1, \bS_0, z, h$. We refer to the round in which the adversary instructs some parties in its  control  to abort as  the \emph{aborting round}, set to $r$ is not abort happen.

\subsubsection{Success probability of \cref{alg:SingeltonAttack}}  Let $H$ be a uniform element of  $\H$. Let  $\gamma =  \alpha/\sqrt{\np}$ letting $\alpha=  \frac{\rho^*}{256\sqrt{r}}\cdot\frac{\sqrt{n}\cdot {\np-1 \choose k-1}^{1/2}}{ \left (64\log(r)\right )^{k-1}} $. We show that either $\Ac_1(H)=\ApproxAttack(\bS_1, \bS_0, 1, H,\gamma)$ or $\Ac_0(H)=\ApproxAttack(\bS_1, \bS_0, 0, H,\gamma)$ succeeds in obtaining the bias of \cref{lemma:SingeltonAttack}. Let $J$ denote the smallest round $i$ such that $B_i^{\bE_1}-B_i^{\bE_0}\geq 3\gamma/4$, and $J=r$ if no such round exists. Furthermore, define $\beta=\frac{1}{2\rho^*\log(r)}\cdot   \frac{64^{-k+1} }{  {\np -1\choose k-1}^{1/2} } $,  let $G_{r,\alpha}$ denote the event $\max_i \set{B_i^{ \bS_1 }-B_i^{\bS_0}}\geq \alpha/\sqrt{\np}$, let $E_{r,\alpha}$ denote the event $(\max_i \set{\size{B_i^{ \bE_1 }-B_i^{\bS_1}}}\geq \alpha/8\sqrt{\np}) \lor (\max_i \set{\size{B_i^{ \bE_0 }-B_i^{\bS_0}}}\geq \alpha/8\sqrt{\np})$. \cref{lemma:SingeltonAttack} follows from \cref{claim:singleton:gap,claim:singleton:err}.

\begin{claim}\label{claim:singleton:gap}
	$\pr{J\neq r\mid G_{r,\alpha} \land \neg E_{r,\alpha}} = 1$.
\end{claim}

\begin{claim}\label{claim:singleton:err}
	$\pr{E_{r,\alpha}}\leq 4r\cdot\exp(-\alpha^2/192)\leq \frac{1}{r}$, for $r$ large enough. 
\end{claim}

\noindent We prove \cref{lemma:SingeltonAttack} assuming the two claims above.

\begin{proof}[Proof of \cref{lemma:SingeltonAttack}]
First we observe that, under adversary $\Ac_z(H)$, the honest party's output is identically distributed with $B^{\bS_z(H)}_{J}$. Thus, like in the proof of \cref{lemma:LapAttack}, it suffices to lower-bound $\ex{B^{\bS_1(H)}_{J}-B^{\bS_0(H)}_{J}}$. By the choice of $\alpha$ and $\beta$ and \cref{eq:lemma:approxgap} of \cref{lemma:SingeltonAttack}, it holds that $\pr{G_{r,\alpha}}\geq \beta$. Consequently,
\begin{align*}
	\ex{B_J^{\bS_1(H)}-B_J^{\bS_0(H)} } &\geq \ex{B_J^{\bS_1(H)}-B_J^{\bS_0(H)}\mid G_{r,\alpha} \land \neg E_{r,\alpha} }\cdot\pr{G_{r,\alpha}\land\neg E_{r,\alpha}} - \pr{E_{r,\alpha}}\\
	& \geq  \left (\ex{B_J^{\bE_1}-B_{J}^{\bE_0}\mid G_{r,\alpha} \land \neg E_{r,\alpha} } - \frac{\alpha}{4\sqrt{\np}}\right )\cdot \pr{G_{r,\alpha}\land\neg E_{r,\alpha}} - \pr{E_{r,\alpha}} \\
	& \geq \frac{1}{2}\cdot \frac{\alpha}{\sqrt{\np}} \cdot   \pr{G_{r,\alpha}}  - 2\cdot \pr{E_{r,\alpha}}\geq \frac{\alpha\beta}{2\sqrt{\np}} -2\cdot\pr{E_{r,\alpha}}\\
	& \geq \frac{1}{1024\sqrt{r}\log(r)}\cdot \left (\frac{1}{64^2\cdot \log(r)}\right )^{k-1} - \frac{2}{r}  .
\end{align*}
\end{proof}

\paragraph*{Proof of \cref{claim:singleton:gap}.} If $E_{r,\alpha}$ did not occur, then $B_i^{ \bE_1 }-B_i^{\bE_0}$ differs from  $B_i^{ \bS_1 }-B_i^{\bS_0}$ by at most $\frac{\alpha}{4\sqrt{\np}}$. If the latter is greater than $\alpha/\sqrt{\np}$, then the former is greater than $\frac{3\alpha}{4\sqrt{\np}}=3\gamma/4$.

\paragraph*{Proof of \cref{claim:singleton:err}.} By assumption, ${\np \choose k}\geq r\log(r)^{2k}$. It follows that 

\begin{align*}
	\alpha&= \frac{\rho^*}{256\sqrt{r}}\cdot\frac{\sqrt{k}\cdot {\np \choose k}^{1/2}}{ \left (64\log(r)\right )^{k-1}}  \geq \frac{\rho^*\sqrt{k}\cdot \log(r)}{2^{ 6k+8}}  .
\end{align*}
Thus, by noting that $\size{\H}\geq \np/3$, apply union bound and Hoeffding's inequality (\cref{fact:Hoeffding})\footnote{Hoeffding's inequality holds for any fixing of the random inputs, and thus it also holds over the probability space of those random inputs}, and deduce that
\begin{align*}
	\pr{E_{r,\alpha}}\leq 4r\cdot\exp(-\alpha^2/192)\leq 4r\cdot\exp(-2\log(2r))  ,
\end{align*}
where the last inequality holds for $r$ large enough, since $\exp(-\alpha^2/192)\le e^{O(-\log(r)^2)}$.

\subsection{Proof of \cref{lemma:UsefulNaget} }\label{section:nugget}

\begin{notation} The concatenation  of two tuple subsets $\bS_1, \bS_0 \subseteq 2^{[\np]}$,  denoted  $\bS_1 \lVert \bS_0$, is defined by $\set{\cU_1 \cup \cU_0 \colon\cU_1  \in \bS_1 , \cU_0  \in \bS_0}$.   
\end{notation}

\noindent 
For reference, we recall of the nugget \cref{def:UsefulNaget}.

\begin{definition}[Restatement of \cref{def:UsefulNaget}]\label{sef}
\NuggetDef
\end{definition}

Next we prove that any protocol admits a nugget.

\begin{proof}[Proof of \cref{lemma:UsefulNaget}]

	We prove the lemma by explicitly constructing the sets (in \cref{figure:nugget}). The algorithm stops as soon as it finds $\bS_1$, $\bS_0$ and $\cH$ such that $\max_i \size{B^{\bS_z(h)}_i - B^{\bS_z(h)}_i}$ is small with the probability specified by \cref{figure:nugget}, for every $z\in \zo$ and $h\in \cH$. Furthermore, if $k^\ast <k$, the construction guarantees that $\max_i \size{B^{\bS_1}_i - B^{\bS_0}_i}$ is large with probability specified by \cref{figure:nugget}. We verify that all the other technical requirements of \cref{lemma:UsefulNaget} are met. By construction, there exists $\cQ\subseteq \cP$ of size $(k-k^\ast-1)$, parties $p_1, p_0\in (\cP\setminus \cQ)$ and tuple set  $\cC\in \set{\cA_1,\cA_0}$, such that $\bS_z$ and $\H$ are of the form
	\begin{align*}	 
	\bS_z&= 
	\begin{cases} 
	 {\cP\choose k-1}\lVert{\cA_{z}\choose 1}  & \text{ if } \ks\in \set{k, k+1} \\
	\cQ \lVert \set{p_z} \lVert { \cP\setminus (\cQ \cup \set{p_z}) \choose k^\ast-1}\lVert {\cC\choose 1} & \text{ if } \ks\in \set{1,\ldots, k-1}\\ 
	\end{cases} \\
	\H&=  
	\begin{cases} 
	\cA_{0} & \text{ if } \ks\in \set{k+1} \\ 
	\P \setminus (\cQ\cup\set{p_1,p_0}) & \text{ if } \ks\in \set{2,\ldots, k}  \\
	\cC  & \text{ if } \ks=1 
	\end{cases}
	\end{align*} 
	
	It is easy to verify that the case $\ks= k+1$. For $\ks =1$, we remark that because $\bS_z$ is of the form $\cQ \lVert \set{p_z} \lVert {\cH \choose 1}$, for some fixed set of parties $\cQ$ and party $p_z$, it is immediate that $\size{\bS_z(h)}=1$, for every $h\in \cH$. It remains to verify the conditions for $\ks \in \set{2,\ldots,k}$. We remind the reader that $\ks\le k <\sqrt{\np}$. Clearly, for every $h\in \cH$, $\ppr{\cu\la \bS_{z}}{h\in \cu}=\frac{\ks-1}{\size{\H}}$ or $\frac{\ks-1}{\size{\H}+1}\leq \frac{1}{2}$, and $\ppr{\cU\la \bS_{z}}{h\in \cU}=\ppr{\cU\la \bS_{z'}}{h \in \cU}$. Furthermore, 
	\begin{align*}
		\dfrac{\ppr{u\la \bS_{z}}{h\notin u}}{\ppr{u\la \bS_{z}}{h\in u}} &\geq \dfrac{1-(\ks-1)/\size{\H}}{(\ks-1)/\size{\H}}  = \dfrac{\size{\H}}{\ks-1} -1\\
		& \ge \dfrac{n/3 - (k-\ks-1) -2}{\ks-1}-1 = \dfrac{n/3 - k+\ks-1}{\ks-1} -1\\
		& \ge \frac{1}{4}\cdot \dfrac{n  - k+\ks-1}{\ks-1}  ,
	\end{align*} 
	where the last follows from $n/3\ge 4k +1$, for $n$ large enough since $\sqrt{n}>k$. Finally, $\ppr{h\la \H, \cU\la \bS_z(h)}{\cU=\cU'}=\ppr{\cU\la \bS_z}{\cU=\cU'}$  follows immediately from the definition $\bS_z$ and $\H$.
\end{proof}


\begin{figure} [H]
	
	\vspace{.1in}
	
		Let $\cA_1,\cA_0, \cP\subset [\np]$ denote an arbitrary equal-size partition of $[\np]$  (\ie $\cA_1$, $\cA_0$ and $\cP$ are pairwise disjoint and ${\cA_1}\cup{\cA_0}\cup{\cP}=[\np]$,  \wlg  $\np$ is a multiple of $3$). \smallskip
		
		 Define $\ks\in [k+1]$, $ \bS_{1} , \bS_{0}\subseteq {[\np]\choose k}, \H\subseteq [\np]$ and $\rho^\ast\in \Rng(r) $ by the following iterative process:
		\begin{enumerate}
			\item Let $\bS^{k+1}_{1} ={\cA_1\choose 1}\lVert  {\cP\choose k-1}$, $\bS_0^{k+1}={\cA_0\choose 1} \lVert  {\cP\choose k-1}$, $\H_{k+1}=\cA_0$.
			
			\item Let $\bS^{k}_{1} ={\cA_1\choose 1}\lVert  {\cP\choose k-1}$, $\bS_0^{k}={\cA_0\choose 1} \lVert  {\cP\choose k-1}$, $\H_k=\cP$, and $c^{k}_1=c^{k}_0=\emptyset$.
			
			\item If  $\exists\rho\in \Rng(r)$ such that	$	\pr{\max_{i\in [r]}\size{ B_i^{\bS^{k+1}_1} -  B_i^{\bS^{k+1}_0  }}\geq  \frac{\rho}{256\sqrt{r}} }\geq \frac{1}{2\rho\log(r)}$:\label{figure:item:findrho}
			 
			 \begin{enumerate}
			  \item Set $\rho_{k}=\rho$.

				 \item  For $\ell=k,\ldots, 2$:   
				 \begin{itemize}
					 \item[]\hspace{-0.27in} If  $\exists  z\in \{1,0\}$, $h, h'\in \H_{\ell}\setminus c^{\ell}_1\cup c^{\ell}_0$,  $\rho\in \Rng(r) $ such that \label{figure:item:goDOWN}
					 \begin{align*}
				 	\hspace{-0.1in}	&\pr{\max_{i\in [r]}\size{ B_i^{\bS_z^{\ell}(h)} \! - \! B_i^{\bS_z^{\ell}(h')}}\geq  \frac{\rho}{ 256\sqrt{r}} \cdot \dfrac{ \coef{\ell-1} ^{1/2}}{\left (64\log(r)\right)^{k-\ell+1}} }   \geq \frac{1}{2\rho\log(r)}\cdot \frac{64^{-k+\ell-1}}{ \coef{\ell-1}^{1/2} }
				 	\end{align*} 
				 	
				 	\hspace{-0.18in}  define:
				 	\begin{enumerate}
			 			\item $\bS^{\ell-1}_{1}=\bS^{\ell}_{z}(h)$, $\bS^{ \ell-1}_{0}=\bS^{ \ell}_{z}(h')$,
			 			\item $\H_{\ell-1}=\H_\ell\setminus c^\ell_z$,
			 			\item $c^{\ell-1}_1=\set{h}$ and $c^{\ell-1}_0=\set{h'}$,
			 			\item $\rho_{\ell-1}=\rho$.
				 	\end{enumerate}  
					 \item[] \hspace{-0.27in}  Else, define $\ks= \ell$, $\rho^*=\rho_{\ell}$, $(\bS_1,\bS_0)=(\bS_1^\ell,\bS_0^\ell)$ and  $\H=\H_{\ks} \setminus c_{\ks}^1\cup c_{\ks}^0$.
	 			\end{itemize}
	 	
	 	         \item If $\ks$ was  not  assigned, set $\ks=1$, $\rho^\ast=\rho_1$, $(\bS_1,\bS_0)=(\bS_1^1,\bS_0^1)$, and let $\H=\cA_1$ if $\bS_1$ and $\bS_0$ are obtained as a concatenation of $\A_1$ with some other tuple set, and  $\H=\cA_0$ otherwise.
	 	         
	 	         \end{enumerate}
			 
			 Else, let $\ks=k+1$, $\rho^\ast=1$ and $(\bS_1,\bS_0,\H)=(\bS_1^{k+1},\bS_0^{k+1},\H_{k+1})$. 
		\end{enumerate} 
		\caption{The Nugget} \label{figure:nugget}
\end{figure}


 
\subsection{Uniform Adversaries} \label{sec:UniformAttack} 
In this section, we show how to replace the nonuniform adversary with a uniform one. For reference, we restate our main theorem.
 
\begin{theorem}[Restatement of \cref{thm:Unimain}]
	\MainTheorem
\end{theorem}

There are two barriers when   emulating the proof of the nonuniform case; the first one is  finding \emph{uniformly} the nugget and its parameters $k^*$ and $\rho^*$ according to \cref{def:UsefulNaget}.  The second barrier is mounting the ``martingale'' attack with a uniform variant of the game-value sequence $X_0,\ldots, X_r$, in the case that $k^*=k+1$. In \cref{sec:nuggfind}, we show that there is an algorithm $\Nugg$ that finds the nugget with some ``$\eps$-loss'', i.e.~the gap and similarity are guaranteed modulo $\eps$, with probability $1-e^{-1/\eps}$. That being said, because $\eps$ can be taken arbitrarily small, say $1/r^{1000}$, it bears no consequence to the analysis of our attacks. Thus, if $k^*\neq k+1$, we can already deduce a bias of magnitude $1/\sqrt{r}\log(r)^k$ for the uniform adversary setting. The case $k^*= k+1$ requires careful treatment because of the game-value sequence, and we address it  in \cref{sec:uniGV}.

\subsubsection{Finding the Right Nugget}\label{sec:nuggfind}

We show that there is an algorithm $\Nugg$ that finds the nugget with some ``$\eps$-loss'', i.e.~the gap and similarity are guaranteed modulo $\eps$, with probability $1-e^{-1/\eps}$.

\begin{proposition}
There exists an algorithm $\Nugg$ taking input a coin-tossing protocol $\Pi$ and a number $\eps\in (0,1)$ such that the following holds. With probability $1-e^{-1/\eps}$, $\Nugg(\Pi,\eps)$ outputs $(k^*, \rho^*, \bS_0, \bS_1,\cH)$ such that the following holds according to the value of $\ks$:
 \begin{description}
 	\item  [$\ks =1$:] ~
 	
  \begin{enumerate}
  	
  	\item  $\pr{\max_{i\in [r]}\size{ B_i^{\bS_1}-  B_i^{\bS_0}}\geq  \frac{\rho^\ast}{ 256\sqrt{r}} \cdot \frac{ \coef{\ks} ^{1/2}}{\left (64\log(r)\right)^{k-\ks}} } \geq \frac{1}{2\rho^\ast\log(r)}\cdot \frac{ 64^{-k+\ks}}{ \coef{\ks}^{1/2} } - \eps$. \label{eq:nugEPS}

 	\item $\H   \ge n/3$, $\size{\bS_1}=\size{\bS_0}=\size{\H}$, and $\size{\bS_z(h)}=1$ for every $h\in \H$ and $z\in  \zo$.
 	
 \end{enumerate}
 	
 		\item[$\ks \in \set{2,\ldots, k}$:] ~
 		
 \begin{enumerate}
 			
 \item Same as \cref{eq:nugEPS}  for  $\ks =1$.	 		
 		
 	\item  For  every $h, h' \in \H$, $z,z'\in \zo$,  $\cU'\in \bS_z$ and  $\rho\in \Rng(r)$:
 	
 	\begin{enumerate}
 		
 		\item 	  $$\pr{\max_{i\in [r]}\size{ B_i^{\bS_z(h)}-  B_i^{\bS_z(h')}}\geq  \frac{\rho}{ 256\sqrt{r}} \cdot \frac{ \coef{\ks-1} ^{1/2}}{\left (64\log(r)\right)^{k-\ks+1}} } \leq \frac{1}{2\rho\log(r)}\cdot \frac{ 64^{-k+\ks-1}}{ \coef{\ks-1}^{1/2} } + \eps  .$$

 		\item $\ppr{\cU\la \bS_{z}}{h\in \cU}=\ppr{\cU\la \bS_{z}}{h'\in \cU}\le \frac{1}{2}$.
 		
 		\item $\ppr{\cU\la \bS_{z}}{h\in \cU}=\ppr{\cU\la \bS_{z'}}{h \in \cU}$.
 		
 		\item $ \ppr{\cU\la \bS_{z}}{h\notin \cU}/\ppr{\cU\la \bS_{z}}{h\in \cU} \geq \frac{1}{4}\cdot \frac{n-k+\ks-1}{\ks-1}$.
 		
 		\item $\ppr{h\la \H, \cU\la \bS_z(h)}{\cU=\cU'}=\ppr{\cU\la \bS_z}{\cU=\cU'}$  
 		
 	\end{enumerate}
\end{enumerate}

  \item [$\ks= k+1$:]~
 	 \begin{enumerate}
 	 	
 	 	\item $\bS_1(h)=\emptyset$ for every $h\in \H$.
 	 	
 	 	\item  	$\ppr{h\la \H, \cU\la \bS_0(h)}{\cU=\cU'}=\ppr{\cU\la \bS_0}{\cU=\cU'}$  for every  $\cU'\in \bS_0$. 
 	 	
 	 	\item $\pr{\max_{i\in [r]}\size{ B_i^{\bS_0 }-  B_i^{\bS_1 }}\geq  \frac{\rho}{ 256\sqrt{r}}   } \leq \frac{1}{2\rho\log(r)} +\eps $ for every $\rho\in \Rng(r)$.
  \end{enumerate}
 \end{description} 
 The running time $\Nugg$ is polynomial in the running time of $\Pi$ and $1/\eps$, and it uses only oracle access to $\Pi$'s next-message function.
\end{proposition}

\noindent 
The proof, sketched below,  follows standard approximation via sampling argument . 
\begin{proof}[Proof's sketch]
	
$\Nugg$ samples a random partition $\cA_1,\cA_0,\P$ of $[n]$ and simply follows the steps of \cref{figure:nugget} with the following caveats: \begin{itemize}
\item In \cref{figure:item:findrho}, $\Nugg$ approximates the value of $\pr{\max_{i\in [r]}\size{ B_i^{\bS^{k+1}_1} -  B_i^{\bS^{k+1}_0  }}\geq  \frac{\rho}{256\sqrt{r}} }$, by running $\Pi$ a number of $1/\eps^s$ times, for some constant $s$.
\item Similarly, in \cref{figure:item:goDOWN}, $\Nugg$ approximates the value of $\pr{\max_{i\in [r]}\size{ B_i^{\bS_z^{\ell}(h)} \! - \! B_i^{\bS_z^{\ell}(h')}}\geq  \frac{\rho}{ 256\sqrt{r}} \cdot \dfrac{ \coef{\ell-1} ^{1/2}}{\left (64\log(r)\right)^{k-\ell+1}} } $, by running $\Pi$ a number of $1/\eps^s$ times, for some constant $s$.
\end{itemize} 
For suitable choice of $s$, Hoeffding's inequality guarantees that, with probability $1-e^{-1/\eps}$, $\Nugg$'s approximation is within $\eps$ of the ``true'' value, at every step.  
\end{proof}

\newcommand{\Gu}{\widehat{G}}
\subsubsection{Computing the Game-Value Sequence}\label{sec:uniGV}
We now explain how to give  a uniform variant of the game-value sequence to mount a successful ``martingale'' attack. At the heart of the non-uniform attack is the sequence of random variable $X= (X_0,\ldots,X_r)$ defined as follows, \wrt to the sequence of backup values $B^\bS_{1},\ldots,B^{\bS}_{r}$ and output of the protocol $\out$. Recall function $g:[0,1]^3\times \zo\mapsto \zo$ defined by 
\begin{align} \label{eq:trigger:recall}
g(x,y,y',\aux )=&    
\begin{cases} 
\aux & \text{if }  |y-x|< 1/64\sqrt{r} \,\lor\, |y'-x|< 1/64\sqrt{r},\\ 
0 & \text{otherwise;}
\end{cases}
\end{align}

\begin{definition}[Restatement of \cref{eq:DefX}]
	\GameValue{eq:trigger:recall}
\end{definition}

The attacked used the following two properties of  $X$:
\begin{enumerate}
	\item $X_i \in \ex{\out\mid B^\bS_{i}, B^\bS_{i-1},X_{i-1}, \sum_{\ell< i } (X_{\ell} - X_{\ell-1})^2, G_{i-1}} \pm 1/200r$, and
	\item $\pr{\exists i \in [r] \colon \size{X_i - X_{i-1}} \ge 1/32 \sqrt{r} } \in \Omega(1)$.
\end{enumerate}
The first item holds by definition. The second item follows by the \sosa weak martingale property of $X$. Hence, to prove the uniform case, all we need to find is a  uniformly constructed sequence $\Xu = (\Xu_1,\ldots,\Xu_r)$ for which the above two properties hold.

We show how to construct a sequence that almost achieves the above properties, and still suffices for our purposes. Specifically, we show that with high probability (\ie $1- e^{-r}$) over a choice of some initialization randomness $\Mu$, there exists a a sequence of random variables $\Xu = (\Xu_1,\ldots,\Xu_r)$ where each $\Xu_i \in [0,1]$ is efficiently constructed from $B^{\bS}_{\le i}$, and the following holds for every $i\in [r]$: 

\begin{align}\label{def:Xu}
\pr{\Xu_i \notin \ex{\out\mid B^\bS_{i}, B^\bS_{i-1}, \Xu_{i-1}, \sum_{\ell< i } (\Xu_{\ell} - \Xu_{\ell-1})^2, \Gu_{i-1}} \pm 1/200r} \le 1/r^2
\end{align}
letting  $\Gu_i = g(\Xu_i, B^\bS_{i}, B^{\bS}_{i-1}, \Gu_{i-1})$.   Namely, $\Xu$ is close to being a \sosa weak martingales sequence. Fortunately, we prove that  such a sequence still  has a jump with constant probability. Specifically,  \cref{claim:SOSAjumpsWHP} yields that
\begin{align}\label{eq:XuHasJump}
\pr{\exists i \in [r] \colon \size{\Xu_i - \Xu_{i-1}} \ge 1/32 \sqrt{r} } \in \Omega(1)
\end{align}
It is easy to verify that the proof of \cref{lemma:MartingalesAttack} still goes through \wrt such a sequence. In the rest of this section we define   a uniformly constructed sequence $\Xu$ for which \cref{def:Xu} holds.

\begin{remark} We emphasize that our goal is to construct a sequence $(\Xu_1,\ldots,\Xu_r)$ satisfying \cref{def:Xu,eq:XuHasJump}, with little regard to how close it is to the ``real'' sequence $X_1,\ldots, X_r$. As mentioned in the introduction, because of the recursive nature of $X_1,\ldots, X_r$, approximating such a sequence may be hopeless. For further discussion, we refer the reader to \cref{sec:intro:martingales}. 
\end{remark}

\begin{notation}
Let $\cB$ and $\cD$ denote the sets $\supp(B_{\cdot}^\bS)$ and $\set{0,\frac{1}{(200r)^2}, \frac{2}{(200r)^2},\ldots , r}$, respectively. 
\end{notation}

\newcommand{\numiter}{r^{50}}

\begin{algorithm}[Algorithm $\allchi$ for constructing $\set{\Mu_i}_{i=0,\ldots, r}$] \label{algo:allchi}  
\item Parameters:   $\bS\subseteq {[n] \choose k}$ 
\item Description:
	\begin{enumerate}
		\item Sample $\set{(b_{i,1}^\ell,\ldots, b_{i,r}^\ell)\leftarrow (B^{\bS}_1,\ldots, B^{\bS}_r) }_{\substack{  i=0,\ldots, r \\ \ell=1,\ldots, \numiter}}$, by running $r\cdot \numiter$ instances of protocol $\Pi$ \Inote{ using the next message function of $\Pi$?}\Nnote{see edit}.
		\item For $i=0,\ldots,,  r$:

		\quad Compute $\Mu_i = \itchy\left (i ,   \set{\Mu_{j}}_{j<i}, \set{(b_{i,1}^\ell,\ldots, b_{i,r}^\ell)}_{\ell=1,\ldots, \numiter} \right )$.
		
	\end{enumerate} 
		\item Output: $\set{\Mu_{i} }_{i=0}^r$  \\
\end{algorithm}

\begin{algorithm}[Algorithm $\itchy$] \label{algo:chi} 
	\item  Parameters: $i\in [r]$, $\set{\Mu_{j}:\cB^2\times \cD^2\times \set{0,1}\mapsto \cD}_{j<i}$, $\set{(b_{1}^\ell,\ldots, b_{r}^\ell)}_{\ell=1,\ldots, \numiter}$ 
	\item Description:
	\begin{enumerate} 
		\item Set $\Mu_i= 0$ 
		\item For $\ell=1,\ldots, \numiter$  
		\begin{enumerate}  
		\item Set $\sigma_0^\ell=0$, $\tau_0^\ell=1$.
		\item For $j=1\ldots,   i-1  $, compute  
		\begin{enumerate}   
		\item $x_{j}^\ell=\Mu_{j}(b_{j}^\ell,b_{j-1}^\ell, x_{j-1}^\ell,\sigma_{j-1}^\ell,\tau_{j-1}^\ell )$,
		\item $\sigma_{j}^\ell= (x_{j}^\ell-x_{j-1}^\ell)^2 + \sigma_{j-1}^\ell$,
		\item $\tau_{j}^\ell=g(b_{j}^\ell,b_{j-1}^\ell,x_{j}^\ell,\tau_{j-1}^\ell)$.
		\end{enumerate}
		\end{enumerate}

		\item For every $\c=(b,b',x,\sigma,\tau)\in  \cB^2\times \cD^2\times \set{0,1}$, compute 
		\begin{enumerate}
		\item  $q_{\c}=\size{\set{\ell\in [\numiter]\,:\, ( b_i^\ell,b_{i-1}^\ell, x_{i-1}^\ell,\sigma_{i-1}^\ell,\tau_{i-1}^\ell)=(b,b',x,\sigma, \tau)}}$
		
		 $ p_{\c}=\size{\set{\ell\in [\numiter]\,:\, ( b_i^\ell,b_{i-1}^\ell,x_{i-1}^\ell, \sigma_{i-1}^\ell,\tau_{i-1}^\ell)=(b,b',x,\sigma, \tau)\land b_r^\ell=1}}$.
		\item If $q_\c\neq 0$,  set $ \Mu_{i}(b,b',x,\sigma,\tau) =\rnd{1/200r}\left ( p_{\c}/q_{\c}\right )$ . 
		\end{enumerate} 
	\end{enumerate}
	\item Output: $\Mu_i$.

\end{algorithm}  

Namely, \cref{algo:allchi}  operates in a sequence of $r$ iterations as follows. At iteration $i$, the algorithm outputs a function of $\mu_i$ such that $\mu_i(b,b',x,\sigma, \tau)$ approximates $\ex{\out\mid B^\bS_i=b,B^\bS_{i-1}=b', \hist_{\Mu_{<i}}(B^\bS_{<i})=(x,\sigma, \tau)}$, where $\Mu_{<i}=\Mu_{0},\ldots,\Mu_{i-1}$ corresponds to the functions constructed at the previous iterations and the function $\hist_{\Mu_{<i}}$ maps sequences $B^\bS_{<i}$ to $3$-tuples in $\cD^2\times \set{0,1}$. Intuitively, the function $\hist_{\Mu_{<i}}$ encodes an aggregated account of the sequence $B^\bS_{<i}$ consisting of -- the previous (approximated) game-value -- the sum of (approximated) squares -- and the ``trigger'' of the attack, i.e.~whether a backup value diverged significantly from the (approximated) game-value at any given round of the protocol.

\begin{remark}[Running time of $\allchi$] It is immediate to see that \cref{algo:chi} runs in time polynomial in the running time of $\Pi$ and $r$.
\end{remark}

Formally, for an output  $\Mu = \set{\Mu_i}_{i\in [r]}$ of $\allchi$, we define the sequence $\Xu^{\Mu} = (\Xu_0^{\Mu},\ldots, \Xu^{\Mu}_r)$ as follows.

\begin{definition}[$\Xu^{\Mu}$] \label{algo:Xseq} 
For  fixed value of $\Mu = \set{\Mu_i}_{i\in [r]}$ output  by  $\allchi$, the  sequence  $\Xu^{\Mu}= (\Xu^{\Mu}_0,\ldots, \Xu^{\Mu}_r )$ is defined as follows. Let $\Xu^{\Mu}_0 = 1/2$ . For $i\in [r]$, define $\Xu^{\Mu}_i$ such that 
\begin{align*} 
\Xu^{\Mu}_i &= \Mu_{i}(B^\bS_i,B^\bS_{i-1}, \Xu^{\Mu}_{i-1}, \sum_{j\le i-1}(\Xu^{\Mu}_{j}-\Xu^{\Mu}_{j-1} )^2, \Gu^{\mu}_{i-1})   
\end{align*}
where $\Gu^{\Mu}_i  =g(B^\bS_i,B^\bS_{i-1},\Xu^{\Mu}_i, \Gu^{\mu}_{i-1})$, letting $\Gu_{0}=1$ and $ g $ be defined according to \cref{eq:trigger:recall}.
\end{definition} 

We conclude the section by proving the following claim.
\begin{claim}
	With save but probability $O(e^{-r})$ over $\Mu  \la \allchi(\Pi, \bS)$, the following holds for every $i\in [r]$:
\begin{align} 
\pr{\Xu^{\Mu}_i \notin \ex{\out\mid B^\bS_{i},B^\bS_{i-1}, \Xu^{\Mu}_{i-1}, \sum_{\ell< i } (\Xu^{\Mu}_{\ell} - \Xu^{\Mu}_{\ell-1})^2, \Gu^{\Mu}_{i-1}} \pm 1/200r} \le 1/r^2.
\end{align} 
\end{claim}

\begin{proof} Let $\eps=r^{-10}$. For conciseness, write $Z_i^\Mu$ for the $5$-tuple $(B^\bS_{i},B^\bS_{i-1}, \Xu^{\Mu}_{i-1}, \sum_{\ell< i } (\Xu^{\Mu}_{\ell} - \Xu^{\Mu}_{\ell-1})^2, \Gu^{\Mu}_{i-1})$ and notice that $\size{\supp(Z_i^\Mu)}\le \size{\cB^2\times \cD^2\times \set{0,1}} \le r^{8}$. Using the notation from  \cref{algo:chi}, for every $i\in [r]$ and $\c\in \supp(Z_i)$, it holds that
\begin{align*}
\ppr{\Mu_i}{\size{ \pr{\out=1\land  Z_i^\Mu=\c} - \frac{p_{\c}}{\numiter}}\ge \eps^2}&\le 2\cdot \exp(-2\cdot \numiter\cdot \eps^4),\\
\ppr{\Mu_i}{\size{ \pr{ Z_i^\Mu=\c} - \frac{q_{\c}}{\numiter}}\ge \eps^2}&\le 2\cdot \exp(-2\cdot \numiter\cdot \eps^4) .
\end{align*}
Both inequalities follow by Hoeffding's inequality. Consequently, for every $\c\in \supp(Z_{\cdot}^\Mu)$, we deduce that
\begin{align*} 
\ppr{\Mu_0,\ldots, \Mu_r}{\exists  i\in [r] :\size{ \pr{\out=1\land  Z_i^\Mu=\c} - \frac{p_{\c}}{\numiter}}\ge \eps^2 \lor \size{ \pr{ Z_i^\Mu=\c} - \frac{q_{\c}}{\numiter}}\ge \eps^2 }&\le e^{-r}
\end{align*}
Hereafter, we fix a mapping $\Mu=(\Mu_0,\ldots, \Mu_r)$ satisfying, for every $i\in [r]$ and $\c\in \supp(Z_i)$, 
\begin{align}
\size{ \pr{\out=1\land  Z_i^\Mu=\c} - \frac{p_{\c}}{\numiter}}< \eps^2, \nonumber \\ 
\size{ \pr{ Z_i^\Mu=\c} - \frac{q_{\c}}{\numiter}}< \eps^2. \label{eq:goodmap}
\end{align} 
Next, we fix $i$ and $\c$ such that $\pr{ Z_i^\Mu=\c}\ge \eps$. Using the fact that $1/(1+x)\in 1\pm 2x$, for small enough $x$, we deduce that 
\begin{align}
	\frac{p_{\c}}{q_{\c}}&\in \frac{\pr{\out=1\land  Z_i^\Mu=\c}\pm \eps^2}{\pr{    Z_i^\Mu=\c}\pm \eps^2 } \\
	& \in \frac{\pr{\out=1\land  Z_i^\Mu=\c}\pm \eps^2}{\pr{  Z_i^\Mu=\c} } \pm 2\cdot \eps\nonumber\\
	&\in  \ex{\out\mid Z_i^\Mu=\c } \pm 3\eps\nonumber
	\end{align}
	Our choice of $\eps$ yields that $\Xu^\mu(\c)=\rnd{1/200r}(p_{\c}/q_{\c})\in \ex{\out\mid Z_i^\Mu=\c } \pm 1/200r$. We conclude by noticing that the probability of running into an element $\c$ at round $i$ such that $\pr{ Z_i^\Mu=\c}< \eps$ is bounded above by $\eps \cdot\size{\supp(Z_i^\Mu)}\le \frac{1}{r^2}$.
\end{proof}

\bibliographystyle{abbrvnat}
\bibliography{crypto}

\appendix

\newcommand{\wh}[1]{\widehat{#1}}

\section{Missing Proofs}\label{sec:MissinProofs}

\begin{proof}[Proof of \cref{fact:laplace}]
 We distinguish four cases, depending on the signs of $\gamma$ and $\gamma'$.
	
	\paragraph*{Case 1. $(\gamma \geq 0, \gamma'\geq 0)$.} $p/p'=\frac{\frac{1}{2}\cdot e^{-\gamma}}{\frac{1}{2}\cdot e^{-\gamma+\eps}}= e^{-\eps}\in 1\pm 2\eps$.
	
	\paragraph*{Case 2. $(\gamma \geq 0, \gamma'< 0)$.}  $p/p'=\frac{\frac{1}{2}\cdot e^{-\gamma}}{1-\frac{1}{2}\cdot e^{\gamma-\eps}}= \frac{1}{2e^{\gamma}-e^{2\gamma-\eps}}$. Since $\gamma\geq 0$ and $\gamma-\eps< 0$, it follows that $0\leq \gamma\leq \eps <1 $ and thus $-\eps< \gamma-\eps<\eps$. Thus $\frac{1}{2e^{\gamma}-e^{2\gamma-\eps}}=e^{-\eps}\cdot \frac{1}{2e^{\gamma-\eps}-e^{2(\gamma-\eps)}}\in e^{\eps}\cdot (1\pm \eps^2)\in 1\pm 5\eps$.
	
	\paragraph*{Case 3. $(\gamma < 0, \gamma'\geq 0)$.} $p/p'=\frac{1-\frac{1}{2}\cdot e^{\gamma}}{\frac{1}{2}\cdot e^{-\gamma+\eps}}= 2\cdot e^{\gamma-\eps} - e^{2\gamma-\eps} $. Similarly to the previous case, since $\gamma< 0$ and $\gamma-\eps\geq 0$, it follows that $0> \gamma\geq \eps >-1 $ and thus $\eps<\gamma<-\eps$. Thus $2\cdot e^{\gamma-\eps} - e^{2\gamma-\eps}=e^{-\eps}\cdot (2e^{\gamma}-e^{2\gamma})\in e^{-\eps}\cdot (1\pm \eps^2)\in 1\pm 5\eps$.
	
	\paragraph*{Case 4. $(\gamma < 0, \gamma'< 0)$:} $p/p'=\frac{1-\frac{1}{2}\cdot e^{\gamma}}{1-\frac{1}{2}\cdot e^{\gamma-\eps}}= \frac{1-\frac{1}{2}\cdot e^{\gamma'-\eps'}}{1-\frac{1}{2}\cdot e^{\gamma'}}=\frac{1-\frac{1}{2}\cdot e^{\gamma'} \cdot e^{-\eps'}}{1-\frac{1}{2}\cdot e^{\gamma'}}$. Let $\mu=1-\frac{1}{2}\cdot e^{-\gamma'}$ and notice that $\mu\in [1/2, 1]$. Compute $
	\frac{1-\frac{1}{2}\cdot e^{\gamma'-\eps'}}{\mu}= 1+ \frac{1-\mu}{\mu} - \frac{1-\mu}{\mu} \cdot e^{-\eps'}\in 1+ \frac{1-\mu}{\mu} - \frac{1-\mu}{\mu} \cdot (1\pm 2\eps)
	\in 1 \pm 2\eps$. 
\end{proof}
\begin{lemma}\label{fact:basic3}
	Consider an iterative sequence of $r$ independent Bernoulli trials, where the success probability of the \ith trial is $p_i\in [0,1]$. Assume that $p_r=1$. For $i\in[r]$, let $q_i=  p_i\cdot \prod_{j<i}(1-p_j)$ be the probability of the first success occurring in the \ith trial. It holds that $\sum_{i=1}^{r} q_i\cdot (\sum_{j\leq i} p_j)=1$.
\end{lemma}

\begin{proof} 
	We prove the claim by proving a stronger statement. Namely, for arbitrary $p_r\in [0,1]$, we show that 
	\begin{equation}\label{eq:sumbiases}
	\sum_{i=1}^r q_i \left (\sum_{j\leq i} p_j \right )= 1- \left (\prod_{i\leq r} (1-p_i) \right)\left (1+\sum_{i\leq r} p_i\right )\enspace . 
	\end{equation} 
	Notice that our claim is a special case of \cref{eq:sumbiases} for $p_r=1$. We proceed to prove the equation by induction on $r$. For $r=1$, take arbitrary $p_1\in [0,1]$ and notice that $q_1 p_1=1-(1-p_1)(1+p_1)$. Next, assume that \cref{eq:sumbiases} is true and let $p_{r+1}\in[0,1]$. The calculation below concludes the proof. 
	\begin{align*}
	\sum_{i=1}^{r+1} q_i \left (\sum_{j\leq i} p_j \right )&= 1- \left (\prod_{i\leq r} (1-p_i) \right)\left (1+\sum_{i\leq r} p_i\right ) + p_{r+1} \left (\prod_{i\leq r} (1-p_i) \right)\left (p_{r+1}+\sum_{i\leq r} p_i\right ) \\
	&=1- \left (\prod_{i\leq r} (1-p_i) \right)\left (1-p_{r+1}^2+ (1-p_{r+1})\sum_{i\leq r} p_i\right )\\
	&=1- \left (\prod_{i\leq r+1} (1-p_i) \right)\left (1+ \sum_{i\leq r+1} p_i\right ),
	\end{align*}
	where the last transition follows using $1-p_{r+1}^2 = (1-p_{r+1})(1+p_{r+1})$.
\end{proof}

\begin{lemma}\label{fact:basic4}  
		Consider two iterative sequences, each of $r$ independent Bernoulli trials.  Let $p_i,p'_i\in [0,1]$ denote the success probability of the \ith trial of the first and second sequence, respectively. Assume that $p_r=p'_r=1$. 
		Let $\eps$ be such that  for all $i\in[r]$, it holds that $\frac{p_i}{p'_i},\frac{p'_i}{p_i}, \frac{(1-p'_i)}{(1-p_i)},\frac{(1-p_i)}{(1-p'_i)}\in (1\pm \eps)$. Then, for every $i\in [r]$,
	\begin{equation}\label{eq:active}
	\size{\prod_{j\le i} (1-p_j') - \prod_{j\le i} (1-p_j)}\leq 3\eps\left (\prod_{j\le i} (1-\min(p_j,p_j')) \right )\left (\sum_{j\le i} \min(p_j,p_j') \right ) \enspace .
	\end{equation}
	\Enote{I am guessing that you want to go over $j\le i$ and not $j<i$.}
\end{lemma}

\begin{proof}
	 First, observe that $1-p_i  \in (1\pm 3\eps\cdot  p_i)(1-p_i')$ and $1-p_i'  \in (1\pm 3\eps\cdot  p_i')(1-p_i)$, for every $i\in [r]$. We hint on how to verify the former (the latter is symmetric). If $p_i\ge 1/3$ or if $p'_i\le 2/3$, then verifying $1-p_i  \in (1\pm 3\eps\cdot  p_i)(1-p_i')$ is easy. Otherwise, if $p_i< 1/3$ and $p'_i > 2/3$, then $\eps p_i >1/3$, and hence $3\eps p_i > 1$. Thus, $(1\pm 3\eps\cdot  p_i)(1-p_i')>1> 1-p_i$.  
	
	We prove \cref{eq:active} by induction on $i$. For every $j\in [i]$, let $\wt{p}_j=\min(p_j,p_j')$. For the base case, $\size{(1-p_1)-(1+p_1')}\leq 2\eps \wt{p}_1(1-\wt{p}_1)$. Next, assume that \cref{eq:active} is true up to some $i\in [r]$. Without loss of generality, further assume that $\wt{p}_{i+1}=p_{i+1}$ and let $u\in [0,1]$ such that $1-p_{i+1}=(1+3u \eps p_{i+1})(1-p'_{i+1})$. For the induction step, compute 
	\begin{align*}
	\size{ \prod_{j\le i+1} (1-p_j') -  \prod_{j\le i+1} (1-p_j)}&\leq (1-p_{i+1}')  \size{ \prod_{j\le i} (1-p_j')- (1+3\eps u p_{i+1})\prod_{j\le i} (1-p_j)}\\
	&\le (1-p'_{i+1})\size{\prod_{j\le i} (1-p_j')-\prod_{j\le i} (1-p_j)}+\size{3u\eps p_{i+1}(1-p'_{i+1})\prod_{j\le i} (1-p_j) } \\
	&\le (1-\wt{p}_{i+1})3\eps\left (\prod_{j\le i} (1-\wt{p}_j) \right )\left (\sum_{j\le i} \wt{p}_j \right )+3\eps p_{i+1}\prod_{j\le i+1} (1-\wt{p}_j) \\
	&= 3\eps \left (\prod_{j\leq i+1} (1-\wt{p}_j)  \right ) \left (\wt{p}_{i+1} +\sum_{j\le i} \wt{p}_j \right )\enspace .
	\end{align*} The second inequality is by the triangle inequality. The third inequality follows by the induction hypothesis and the fact that for every $j\in[i+1]$ it holds that $1-p_j, 1-p'_j \leq 1-\wt{p}_j$.  The last transition is true by the assumption that $\wt{p}_{i+1} ={p}_{i+1}$.
\end{proof}

\begin{lemma}\label{fact:basic5} Consider two iterative sequences, each of $r$ independent Bernoulli trials.  Let $p_i,p'_i\in [0,1]$ denote the success probability of the \ith trial of the first and second sequence, respectively. Assume that $p_r=p'_r=1$. For $i\in[r]$, let $q_i=  p_i\cdot \prod_{j<i}(1-p_j)$ and $q'_i=  p'_i\cdot \prod_{j<i}(1-p'_j)$. 
		Let $\eps$ be such that  for all $i\in[r]$, it holds that $\frac{p_i}{p'_i},\frac{p'_i}{p_i}, \frac{(1-p'_i)}{(1-p_i)},\frac{(1-p_i)}{(1-p'_i)}\in (1\pm \eps)$. 
		Then, for every $i\in [r]$, it holds that  
	\begin{equation}
	\size{q_i-q_i'}\leq 3\eps\cdot\min(p_i,p_i')\cdot \left (\prod_{j<i} (1-\min(p_j,p_j')) \right )\left (\frac{1}{3}+\sum_{j\leq i} \min(p_j,p_j') \right )\enspace .
	\end{equation}  
\end{lemma}
\begin{proof}
For every $j\in [i]$, let $\wt{p}_j=\min(p_j,p_j')$.  Without loss of generality, assume that $\wt{p}_{i}=p_{i}$ and let $u\in [0,1]$ such that $p'_{i}=(1+u\eps) p_{i}$ (there exists such $u$ since $p_i'\in p_i(1\pm \eps)$).  
	\begin{align*}
	\size{ p_i\prod_{j<i} (1-p_j) -  p'_i\prod_{j<i} (1-p'_j)}&\le p_{i}  \size{ \prod_{j<i} (1-p_j) - \prod_{j< i} (1-p'_j)}+\size{\eps u p_i\prod_{j< i} (1-p'_j)}\\
	&\le 3\eps\wt{p}_{i}\left (\prod_{j< i} (1-\wt{p}_j) \right )\left (\sum_{j< i} \wt{p}_j \right )+\eps \wt{p}_{i}\prod_{j< i} (1-\wt{p}_j)  \\
	&\le 3\eps \wt{p}_i\left (\prod_{j< i} (1-\wt{p}_j)  \right ) \left ( \frac{1}{3}+\sum_{j< i} \wt{p}_j     \right )\enspace .
	\end{align*} The first inequality is by the triangle inequality. The second inequality follows by \cref{fact:basic4} and the fact that for every $j\in[i]$ it holds that $1-p'_j \leq 1-\wt{p}_j$.  
\end{proof}

\begin{lemma}[Restating \cref{fact:basic2}]\label{fact:basic2-appendix}
	Consider two iterative sequences, each of $r$ independent Bernoulli trials.  Let $p_i,p'_i\in [0,1]$ denote the success probability of the \ith trial of the first and second sequence, respectively. Assume that $p_r=p'_r=1$. For $i\in[r]$, let $q_i=  p_i\cdot \prod_{j<i}(1-p_j)$ and $q'_i=  p'_i\cdot \prod_{j<i}(1-p'_j)$. 
		Let $\eps$ be such that  for all $i\in[r]$, it holds that $\frac{p_i}{p'_i},\frac{p'_i}{p_i}, \frac{(1-p'_i)}{(1-p_i)},\frac{(1-p_i)}{(1-p'_i)}\in (1\pm \eps)$. Then, $\sum_{i=1}^{r-1} \size{q_i -q_i'} \le 4\eps(1 - q_r)$.
\end{lemma} 

\begin{proof}
For every $j\in [r]$, let $\wt{p}_j=\min(p_j,p_j')$, and for every $i\in[r]$ let $\wt{q}_i=  \wt{p}_i\cdot \prod_{j<i}(1-\wt{p}_j)$. Since the $\wt{p}_j$s define an iterative sequence of Bernoulli trials, from \cref{fact:basic5} and \cref{fact:basic3} it follows that,
 \begin{align}
\sum_{i=1}^{r-1} \size{q_i -q_i'} &\le \sum_{i=1}^{r-1}  3\eps\cdot\wt{p}_j\cdot \left (\prod_{j<i} (1-\wt{p}_j) \right )\left (\frac{1}{3}+\sum_{j\leq i} \wt{p}_j \right )\\
&\le 3\eps\cdot\sum_{i=1}^{r-1} \wt{q}_j\left (\frac{1}{3}+\sum_{j\leq i} \wt{p}_j \right )\\
&= \eps \left(\sum_{i=1}^{r-1} \wt{q}_j\right) + 3\eps\cdot\left(\sum_{i=1}^{r} \wt{q}_j\left (\sum_{j\leq i} \wt{p}_j \right )\right ) - 3\eps\cdot\wt{q}_r\left (\sum_{j\leq r} \wt{p}_j \right )\\
&\le 4\eps -4\eps\wt{q}_r\le 4\eps(1-q_r). 
\end{align}
The second to last inequality uses the fact that $\wt{p}_r = p_r=p_r'=1$. The last inequality follows since $1-\wt{q}_r\leq 1-q_r$.	
\end{proof}


\begin{proof}[Proof of \cref{lemma:tower}]
	Straightforward computation. Fix $b\in \supp(B)$
	\begin{align*}
	\ex{\ex{A\mid B, C}\mid B=b}&=\sum_{c}  \ex{A\mid B=b, C=c} \pr{C=c\mid B=b}\\ 
	&=\sum_{c} \sum_{a} a\cdot \pr{A=a\mid B=b, C=c}  \pr{C=c\mid B=b}\\ 
	&=\sum_{a }a\cdot\sum_{c}  \pr{A=a\land  C=c \mid B=b }   \\ 
	&=\sum_{a}   a\cdot \pr{A=a  \mid B=b }   \\
	&=\ex{A \mid B=b}
	\end{align*}
\end{proof}

\begin{proof}[Proof of \cref{lemma:z:4}]
	Immediate consequence of the fact that $\ex{A\mid \ex{A\mid B}, f(B)}=\ex{\ex{A\mid B}\mid \ex{A\mid B}, f(B)}$, since $B$ fully determines $\ex{A\mid B}$ and $f(B)$.
\end{proof}

\begin{proof}[Proof of \cref{lemma:z:3}.]
	Let $B'= \rnd{\delta}(B)$ and fix $b'\in \supp(B)\subseteq \R$ and $c\in \supp(C)$. 
	\begin{align*}
	\ex{A\mid B'=b' \land C=c }& =\sum_a a\cdot \pr{A=a\mid B'=b'\land C=c}\\
	&= \sum_a a \cdot \sum_{b\in [b', b'+\delta]}\pr{A=a\land B=b\mid B'=b'\land C=c}\\
	& =  \sum_a a \cdot \frac{1}{\pr{B'=b'\mid C=c}}\sum_{b\in [b', b'+\delta]}\pr{A=a\land B=b\mid   C=c} \\
	&= \frac{1}{\pr{B'=b'\mid C=c}} \sum_{b\in [b', b'+\delta]}\pr{B=b\mid C=c}\cdot \sum_a a \cdot \pr{A=a\mid  B=b \land C=c}\\
	\\
	&= \frac{1}{\pr{B'=b'\mid C=c}} \sum_{b\in [b', b'+\delta]}\pr{B=b\mid C=c}\cdot   \ex{A \mid  B=b , C=c}\\
	&= \frac{1}{\pr{B'=b'\mid C=c}}\sum_{b\in [b', b'+\delta]}b\cdot \pr{B=b\mid C=c} \\
	& \in [b', b'+\delta] 
	\end{align*}
\end{proof}







 \end{document}